\documentclass[manuscript,screen,nonacm]{acmart}

\usepackage{hyperref}
\usepackage{xspace}

\usepackage{amsmath}
\usepackage{amsthm}
\usepackage{enumerate}
\usepackage{paralist} 
\usepackage{cleveref}
\usepackage{pifont}

\newcommand{\FO}{\mbox{\rm FO}}
\newcommand{\FOt}{\mbox{$\mbox{\rm FO}^2$}}
\newcommand{\TGF}{\mbox{$\mbox{\rm TGF}$}}

\newcommand{\GFU}{\mbox{$\mbox{\rm GFU}$}}

\newcommand{\GFTG}{\mbox{\rm GF+TG}}
\newcommand{\GFtTG}{\mbox{\rm GF$^2$+TG}}

\newcommand{\TGFTG}{\mbox{\rm TGF+TG}}

\newcommand{\GF}{\mbox{\rm GF}}
\newcommand{\GFt}{\mbox{$\mbox{\rm GF}^2$}}

\newcommand{\NP}{\textsc{NP}}
\newcommand{\PTime}{\textsc{PTime}}

\newcommand{\ExpTime}{\textsc{ExpTime}}

\newcommand{\NExpTime}{\textsc{NExpTime}}
\newcommand{\TwoExpTime}{2\textsc{ExpTime}}
\newcommand{\TwoNExpTime}{\textsc{N2ExpTime}}

\newcommand{\str}[1]{{\mathfrak{#1}}}
\newcommand{\restr}{\!\!\restriction\!\!}
\newcommand{\N}{{\mathbb N}}
 
\newcommand{\sss}{\scriptscriptstyle}

\newcommand{\type}[2]{{\rm tp}^{#1}(#2)}

\newcommand{\hD}{h}
\newcommand{\hS}{h'}

\newcommand{\UU}{\mathsf{U}}

\newcommand{\width}{\text{width}}

\usepackage{scalerel}
\newcommand{\AAA}{{\scaleobj{1.25}{\boldsymbol\alpha}}}

\newcommand{\fo}{\FO}
\newcommand{\tgf}{\TGF}
\newcommand{\fot}{\FOt}
\newcommand{\gf}{\GF}
\newcommand{\gftimes}{\GFU}
\newcommand{\gfu}{\gftimes}
\newcommand{\gfcross}{\GF^{\times_2}}
\newcommand{\ur}{\UU}

\newcommand{\ran}{\mathit{ran}}

\newcommand{\Smc}{\mathcal{S}}

\newcommand{\nexptime}{\NExpTime\xspace}
\newcommand{\twoexptime}{\TwoExpTime\xspace}
\newcommand{\twonexptime}{\TwoNExpTime\xspace}
\newcommand{\exptime}{\ExpTime\xspace}
\newcommand{\np}{\NP\xspace}

\renewcommand{\vec}[1]{\bar{#1}}

\newtheorem{theorem}{Theorem}
\newtheorem{lemma}[theorem]{Lemma}
\newtheorem{claim}[theorem]{Claim}

\newtheorem{proposition}[theorem]{Proposition}

\crefname{proposition}{Proposition}{Propositions}
\Crefname{proposition}{Proposition}{Propositions}

\newcommand{\defend}{\hfill$\Diamond$}

\newtheoremstyle{myexample}%
{3pt}%
{3pt}%
{\rm}%
{}%
{\it}%
{.}%
{.5em}%
{}%
\theoremstyle{myexample}
\newtheorem{xample}[theorem]{Example}

\usepackage{framed}
\definecolor{shadecolor}{rgb}{0.93,0.93,0.93}
\newenvironment{example}[1][]
{
	\begin{shaded*}	
		\vspace{-1.5ex}
		\begin{xample}[#1]
		}{
		\end{xample}
		\vspace{-1.5ex}
	\end{shaded*}
}

\newcommand{\footnotetextbefore}[1]{\addtocounter{footnote}{1}\footnotetext{#1}\addtocounter{footnote}{-1}
}

\usepackage{xcolor}
\definecolor{mygreen}{rgb}{0, 0.6, 0}

\setcopyright{acmcopyright}
\copyrightyear{2026}
\acmYear{2026}
\acmDOI{TBD}

\begin{document}

\title{The Triguarded Fragment}

\author{Emanuel Kiero\'nski}
\email{emanuel.kieronski@cs.uni.wroc.pl}
\orcid{0000-0002-8538-8221}
\affiliation{%
  \institution{Institute of Computer Science, University of Wroc\l{}aw}
  \country{Poland}
}

\author{Sebastian Rudolph}
\email{sebastian.rudolph@tu-dresden.de}
\orcid{0000-0002-1609-2080}
\affiliation{%
	\institution{Computational Logic Group, TU Dresden and ScaDS.AI Dresden/Leipzig}
    \country{Germany}
}

\author{Mantas \v{S}imkus}
\email{mantas.simkus@tuwien.ac.at}
\orcid{0000-0003-0632-0294 }
\affiliation{%
      \institution{Institute of Logic and Computation, TU Wien}
	\country{Austria}
}

\begin{abstract}
A prominent research question in computational logic is how to restrict first-order predicate logic ($\fo$) in such a way that the satisfiability problem becomes decidable. Among others, past efforts have identified two prominent decidable $\fo$ fragments of high expressivity: the \emph{guarded fragment} ($\gf$), and the \emph{two-variable fragment} ($\fot$). 
These fragments are of high interest and crucial importance as they provide significant insights into decidability and expressiveness of other prominent (computational) logics like \emph{Modal Logics (MLs)} and various \emph{Description Logics (DLs)}, which play a central role in  Verification, Knowledge Representation, and other areas.
In this article, we show that $\gf$ and $\fot$ can be combined into a new fragment that subsumes both, while maintaining decidability of the satisfiability problem. 
This fragment, called the \emph{triguarded fragment} (denoted $\TGF$), is obtained by relaxing the standard definition of $\gf$ by requiring guardedness of quantification only for subformulae with \emph{three} or more free variables.
We show that, when restricting the use of equality, satisfiability in $\TGF$ is \twonexptime-complete,
dropping to \nexptime-complete when the maximum predicate arity is fixed (a natural assumption in the context of MLs and DLs). We further establish that the problem is \np-complete in terms of \emph{data complexity}, which is again in line with data complexity results for basic expressive DLs.
We observe that many natural extensions of $\TGF$, including the liberal use of equality, lead to undecidability. 
We also establish that TGF has the \emph{finite model property} (providing a tight doubly exponential bound on the model size), whence finite satisfiability coincides with satisfiability. 

\end{abstract}

\keywords{decidability, computational complexity, finite model property}

\maketitle

\section{Introduction}\label{s:intro}
First-order predicate logic (abbreviated by $\fo$ in this article) plays a pivotal role in logic and exhibits many favorable properties in terms of expressivity and model theory. 
However, a well-known drawback is that satisfiability checking of arbitrary $\fo$ formulae is undecidable, which limits the use of full $\fo$ in contexts that require complete automated reasoning. 
This fact also motivates the search for expressive, yet decidable $\fo$ fragments.%
\footnote{In this paper, whenever we call a logic (un)decidable, we mean that the satisfiability problem for formulae of this logic is (un)decidable. 
This will avoid cumbersome phrases. 
There is no danger of confusion, since all logics considered here have a straightforward syntactic definition, thus deciding a formula's membership in the logic will always be trivial.} 
Such fragments play a crucial role in the study of various practically relevant logics, including modal logics (MLs) and the broader family of description logics (DLs).
The latter are logic-based knowledge representation languages, whose expressivity is usually suitably limited to ensure the decidability of basic reasoning problems~\cite{dlhandbook,DBLP:books/daglib/0041477,rudolph2011fodl}. 
They are used for ontological modeling in a variety of domains and supported by efficient automated reasoners.

\footnotetextbefore{We will not go into details regarding the syntax of DLs; the interested reader is referred to the literature. 
Equivalent FO counterparts of the complex DL statements in this example will be provided in Example 2.}
\begin{example}[Geographical Knowledge in Description Logics]
Let us consider an example for modeling knowledge by means of logical statements in DLs.\footnotemark{} 
Factoid statements can be expressed by ground literals just as in $\fo$. 
For instance, we may state (incomplete) information about certain geographic entities and their spatial relationships:
$$\begin{array}{c}
Country(germany) \qquad BalticState(lithuania) \qquad Country(poland) \qquad Sea(balticSea) \qquad Sea(northSea)\\[1ex] 
Borders(germany,northSea) \qquad  Borders(poland,balticSea) \qquad Borders(lithuania,poland))
\end{array}$$
DLs can also express terminological knowledge, that is, semantic relationships of the used vocabulary, e.g., that the ``borders'' relationship only holds between geographical entities ($GE$) and that it is symmetric and irreflexive:
\begin{align*}
	\exists Borders.\top \sqsubseteq GE \qquad \top \sqsubseteq \forall Borders.GE \qquad 
	Borders \sqsubseteq Borders^- \qquad \top \sqsubseteq \neg \exists Borders.\mathsf{Self}.
\end{align*}
More advanced DLs allow for the usage of ``nominals'', defining one-element classes of named entities. 
For example, we may use them to express properties of the Baltic states: 
First, they all border the Baltic Sea and this is also the only sea they border. 
Second, each of them borders a(nother) Baltic state, which does not border Poland:
\begin{align*}
	BalticState & \sqsubseteq \exists Borders.\{balticSea\} \sqcap \forall Borders.(\neg Sea \sqcup \{balticSea\})\\
	BalticState & \sqsubseteq \exists Borders.(BalticState \sqcap \forall Borders.\neg \{poland\}).	
\end{align*}
This knowledge, together with the inequality 
$northSea \neq balticSea$, allows us to infer the following statements:
\begin{align*}
	 \quad\!\!\!\!\! \neg Sea(poland), \qquad \neg BalticState(poland), \qquad \neg BalticState(germany), \qquad \neg Borders(lithuania,northSea).
\end{align*}
\end{example}

Many MLs and DLs can be conceived as ``semantic fragments'' of $\fo$, meaning that straightforward, semantically faithful translations into $\fo$ exist.%
\footnote{Notable exceptions are logics that use fixpoint expressions, such as the modal $\mu$-calculus or DLs allowing for regular expressions, in particular DLs of the $\mathcal{Z}$ family.}
For most decidable such MLs and DLs, the translations can be chosen such that they are \PTime-computable and the obtained formulae fall into well-known decidable fragments of $\fo$, implying not only decidability, but also complexity results, model-theoretic properties, and limits of expressiveness. 

For instance, many standard MLs and  DLs admit a straightforward equivalent translation into $\fot$, the fragment of $\fo$ with at most two variables~\cite{blackburn:inria-00119856,DBLP:journals/ai/Borgida96}. 
Intuitively, in $\fot$, expressivity is restricted to %
pairwise interrelationships between domain elements.
\begin{example}[Translation of DL Statements into $\fot$]
\newcommand{\p}{\hspace{-0.4pt}}
Revisiting Example 1, we find that the ground statements are already in valid $\fot$ syntax; the others can be translated as follows (note the variable re-use in the third line):
\begin{align*}
&\forall xy\p. Borders(\p x,\!y) {\,\Rightarrow\,} GE(\p x) \quad \forall xy\p. Borders(\p x,\!y) {\,\Rightarrow\,} GE(\p y) \quad \forall xy\p. Borders(\p x,\!y) {\,\Rightarrow\,} Borders(\p y,\!x) \quad \forall x.\! \neg Borders(\p x,\!x)\\
&\qquad\forall x. BalticState(x) \Rightarrow \Big(
\exists y.\big(Borders(x,y)\wedge y{\,=\,}balticSea\big)   
\wedge
\forall y.\big(Borders(x,y)\Rightarrow \neg Sea(y) \vee y{\,=\,}balticSea\big) 
\Big)\\
&\qquad\forall x. BalticState(x) \Rightarrow \Big(
\exists y.\big(Borders(x,y)\wedge BalticState(y)   
\wedge
\forall x.(Borders(y,x)\Rightarrow x{\,\neq\,}poland)\big) 
\Big).
\end{align*}
However, %
not every $\fot$ sentence has a counterpart in mainstream DLs: for instance, expressing the statement
$$\left.\right.\qquad\qquad\qquad\qquad\forall xy. \big( Country(x) \wedge Sea(y) \Rightarrow Borders(x,y) \vee Detached(x,y) \big) \qquad\qquad\qquad\qquad (\dag)$$
would require a highly peculiar DL with non-safe Boolean role constructors \cite{RudolphKH08}. 
\end{example}
For $\fot$ without equality, the satisfiability problem has been known to be decidable for over six decades thanks to Scott~\cite{scott1962decision}.  
The decidability of satisfiability in $\fot$ in the presence of equality was established in 1975 by Mortimer~\cite{DBLP:journals/mlq/Mortimer75} and the corresponding \nexptime upper complexity bound has been  known since over two decades~\cite{DBLP:journals/bsl/GradelKV97}. 

\smallskip

Another explanation for the decidability of many MLs and DLs is the fact that they can often be translated into the \emph{guarded fragment} ($\gf$) of $\fo$~\cite{ABN98:GF} -- see also Grädel's discussion on that ``robust'' relationship \cite{DBLP:conf/dlog/Gradel98}. 
Intuitively, $\gf$ requires quantification to be ``relativized'' to tuples of elements that are jointly co-occurring in a relationship.
\begin{example}
The formulae in Example 2 except $(\dag)$ are in $\gf$, but $\gf$ also admits formulae with more than 2 variables. For example the statement
\begin{align*}
\forall xy. Borders(x,y) \Rightarrow \big( Country(x) \wedge Sea(y) \Rightarrow \exists z. BorderBetween(z,x,y) \wedge Coast(z) \big) 
\end{align*} 
is in $\gf$, where the atom $Borders(x,y)$ is called a \emph{guard} for the preceding universal quantifier whereas the atom $BorderBetween(z,x,y)$ is the guard for the preceding existential quantifier.
\end{example}	
Satisfiability checking in $\gf$ is \twoexptime-complete in general, but it is just \exptime-complete under the assumption that the maximum arity of predicates is fixed~\cite{Gra99,DBLP:journals/jolli/CateF05}. 
The latter is relevant: it implies the  \exptime upper bound for consistency checking in many standard DLs, as their $\gf$-translations use predicate symbols of arity ${\leq}\,2$. 

\smallskip

Both \FOt{} and \GF{} possess the \emph{finite model property} (FMP), meaning that any satisfiable sentence has a finite model. 
\begin{example}
Considering all the statements of Example 2 and Example 3 together, we note that any model must contain more than just the elements which are named by constants (since none of the named elements but Lithuania can be a Baltic state, yet there must be another Baltic state neighboring Lithuania). 
In fact, one can even infer that there must be at least two additional Baltic states. 
However, it should be clear that taking all the (finitely many) countries, seas, and coasts on earth with their mutual relationships provides us with a finite model for the noted logical statements. 
Recall that the FMP does not hold true for $\fo$ in general as demonstrated by the sentence
\begin{align*}
SouthOf(southPole,northPole)\  & \wedge\ \forall xyz.\big( SouthOf(x,y) \wedge SouthOf(y,z) \Rightarrow \ SouthOf(x,z)) \big)\\ 
\wedge \ \ \neg \exists x.SouthOf(x,x) \ & \wedge\ \forall xz. \big(SouthOf(x,z) \Rightarrow  \exists y.(SouthOf(x,y) \wedge SouthOf(y,z))\big).
\end{align*}
\end{example}	
As a consequence of $\fot$ and $\gf$ exhibiting the FMP, the satisfiability and the finite satisfiability problems coincide for these fragments. 
For \FOt, existence of a finite model of only exponential size with respect to the sentence was actually the way to establish the above mentioned complexity. 
For \GF, the original FMP result gave rise to a triply exponential bound on the model size \cite{Gra99}, whereas  a tight doubly-exponential bound was established more recently \cite{BGO14}.

\bigskip

Given the importance of $\gf$ and $\fot$, this paper takes a deeper look at them, and proposes a new, very expressive fragment of $\fo$ that subsumes both, while still enjoying decidability and even the finite model property. 
The fragment, called the \emph{triguarded fragment} (denoted $\TGF$), is obtained by relaxing the standard definition of $\gf$.  
As illustrated in the examples, in $\gf$, existential and universal quantification can only be used in (sub)formulae of the form $\exists \vec{x}.(R(\vec{t})\land \psi)$ or $\forall\vec{x}.(R(\vec{t})\Rightarrow \psi)$, where $R(\vec{t})$ is an atomic formula such that $\vec{t}$ contains \emph{all} free variables of $\psi$ (one then says that the atom $R(\vec{t})$ ``guards'' the formula
$\psi$). 
In contrast, in $\TGF$, guardedness of quantification is required only in case $\psi$ has \emph{three} or more free variables (hence the name ``triguarded''). 
This entails that quantification can be used in an unrestricted way for formulae with at most two free variables, and hence $\fot$ gets included in $\TGF$ seamlessly.
Yet, \TGF{} also brings about a new dimension of expressivity, as it allows one to formulate properties expressible neither in \FOt{} nor in \GF{}.
\footnotetextbefore{Note that any subsequence of the described quantifier prefix is also permissible, since it can be extended by ``spurious quantifiers'' over fresh variables not occurring in $\varphi$.}
\begin{example}
The following sentences are not in \FOt{} nor in \GF{} but in \TGF{}: 
\begin{align*}
\forall xz.Location(x) \wedge Location(y) \Rightarrow & \ \exists z. PathBetween(z,x,y) \\
\forall xz.Location(x) \wedge Location(y) \Rightarrow & \ \big(\,\, \forall z_1.(PathBetween(z_1,northPole,x) \Rightarrow Crosses(z_1,equator)   \\[-1ex]
& \vee \forall z_2.(PathBetween(z_2,x,y) \Rightarrow Crosses(z_2,equator)   \\[-1ex]
& \vee \forall z_3.(PathBetween(z_3,y,southPole) \Rightarrow Crosses(z_3,equator) \, \big)
\end{align*}
Even higher expressivity can be obtained when using auxiliary predicates.
Recall that the expressive yet decidable $\fo$ fragment commonly referred to as \emph{Gödel's class} comprises all $\fo$ sentences $\psi$ whose prenex form has the shape 
$$\exists x_1 \ldots \exists x_n \forall y_1 \forall y_2 \exists z_1 \ldots \exists z_m. \varphi,$$
where $\varphi$ is an arbitrary quantifier-free $\fo$ formula not using the equality predicate.\footnotemark{}
For instance, it can be readily checked that the sentence
\begin{align*}
& \exists x_1x_2.Location(x_1) \wedge Location(x_2) \wedge \forall y_1y_2.\big(Location(y_1) \wedge \mathit{CircleOfLatitude}(y_2)\big)\\[-1ex] 
& \left.\right.\qquad\qquad\Rightarrow  \exists z_1. (ShortestPathBetween(z_1,x_1,y_1) \vee ShortestPathBetween(z_1,y_1,x_2)) \wedge Crosses(z_1,y_2)
\end{align*}
falls in Gödel's class. It is not too hard to see that the set of models of any such sentence $\psi$ is projectively characterized by the sentence $\psi'$ obtained from $\psi$ by Skolemizing all $x_1,\ldots,x_n$ and preceding  
$\varphi$ with an auxiliary guard containing all of $y_1,y_2,z_1,\ldots,z_m$. The $\psi'$ thus obtained is in fact in $\TGF{}$ and a model-conservative extension of $\psi$.  
For the example sentence above, this procedure would yield (using telling Skolem names and moving quantifiers for better readability)
\begin{align*}
	& Location(northPole) \wedge Location(southPole) \wedge \forall y_1y_2.\big(Location(y_1) \wedge \mathit{CircleOfLatitude}(y_2)\big)\\[-1ex] 
	& \left.\right.\Rightarrow  \exists z_1. Aux(y_1,y_2,z_1) \wedge (ShortestPath(z_1,northPole,y_1) \vee ShortestPath(z_1,y_1,southPole)) \wedge Crosses(z_1,y_2).
\end{align*}

\end{example}

This article will proceed as follows.
After \Cref{s:notation} provides a few basic remarks on the notation used, \Cref{s:syntax} formally defines $\TGF$ based on the ideas presented above. For reasons that will become clear later in the paper, the use of the equality predicate in $\TGF$ must be restricted. Furthermore, we introduce a convenient syntactic variant of $\TGF$: 
$\gf$ endowed with a built-in predicate $\ur$ that  must be interpreted as the set of all pairs of domain elements. 
In DL parlance, this corresponds to extending $\gf$ with the \emph{universal role}, and thus this fragment is denoted $\gftimes$. 
Since the predicate $\ur$ can be used to provide ``spurious'' guards to formulae with up to two free variables, $\gftimes$ adds to $\gf$ precisely the expressivity needed to capture
$\TGF$. 
In fact, in the rest of the article, we will then mainly focus on $\gftimes$ instead~of~$\TGF$, as this considerably simplifies our technical considerations.
\Cref{s:syntax} concludes by introducing a Scott-like normal form for $\gftimes$ sentences.
\Cref{s:tools} provides some more technical equipment: the well-known machinery of \emph{types}, which are meant to characterize structural configurations of a bounded number of domain elements, as they may occur in a model. Moreover we will introduce some useful model-theoretic constructions.

\Cref{s:compgfu} then provides the exact upper complexity bounds, showing that satisfiability of formulae in $\gftimes$ (and thus in $\TGF$) is decidable in \twonexptime.
This result comes at the cost of a more elaborate development, using a characterization of the satisfiability of a formula in $\gftimes$ via \emph{mosaics}, where a mosaic is a special (finite) collection of types that can be used to build a model for the input formula. The upper bound is then established via a procedure that guesses and verifies an appropriate mosaic. Along the same lines, we proceed to show that the problem is in \twoexptime if constants are disallowed and in \nexptime whenever the predicate arities are bounded by a constant.
Regarding the latter, we note that $\fot$ is already \nexptime-hard (even without equality), which means that in the bounded-arity setting, $\TGF$ and $\gftimes$ do not have higher complexity than their sublogic $\fot$. Last but not least, we also establish satisfiability to be in \NP{} in terms of data complexity.
On top of showing these upper complexity bounds, this approach establishes decidability from first principles and does not hinge on pre-existing results for $\gf$.
The matching lower bounds are easy to derive except for \twonexptime-hardness in the general case, which we obtain by a reduction from the tiling problem for a doubly exponential grid.  

\Cref{sec:datacomplexity} deals with a question inspired by the field of ontology-based querying: what is the data complexity of $\TGF$, that is the complexity of satisfiability of $\TGF$ theories of the shape $\mathcal{T} \cup D$, where $\mathcal{T}$ is arbitrary but fixed while $D$ is a set of ground facts (a ``database'') which is allowed to vary. Using techniques from earlier sections, we establish that $\TGF$ is \textsc{NP}-complete in data complexity. In a similar vein, \Cref{sec:disjdatalog} characterizes the expressivity of $\TGF$ from a querying perspective, by establishing that for any $\TGF$ theory, there is a disjunctive datalog theory which ``recognizes'' the same databases.

In \Cref{s:undec}, we briefly review the possibilities to generalize the obtained decidability result and find them quite limited. Some limitations are inherited from either $\fot$ (undecidability of conjunctive query answering) or $\gf$ (impossibility to decidably add functionality constraints or counting quantifiers), but others originate from the interplay of the components. In particular, in contrast to the fact that satisfiability in $\gf$ and $\fot$ is decidable (with the same complexity) regardless of the presence of equality, the satisfiability of $\TGF$ and $\gftimes$ formulae becomes undecidable when allowing for unrestricted use of equality. While we provide a direct proof, this observation also follows from an analogous result for Gödel's class.

In preparation to the model-theoretic considerations to follow, \Cref{s:decgfu} provides as a warm-up an argument that decidability of $\gftimes$ (and hence $\TGF$) can be established by a reduction to $\gf$ (yielding an alternative decidability proof, from which, however, it is cumbersome to extract tight complexity bounds).
The proof is conceptually rather simple and should help the reader familiarize themselves with the model-theoretic setting.

Finally, in \Cref{s:fmpgfu}, we turn to the question if \TGF{} has the finite model property (and thus, if finite model reasoning and the associated complexity is any different from the arbitrary-model case). Notably, neither technique used for establishing the FMP for \FOt{} and \GF{} directly lends itself to solving the question for \TGF{}. Indeed one of this article's core contributions is to answer this question in the positive. The corresponding results are established through rather elaborate, carefully crafted model constructions; while more complicated, they generally follow the scheme from \Cref{s:decgfu}. By coupling these with meticulous inspections of existing proofs, we are able to extract tight bounds for the model size, which is doubly exponential with respect to the size of the input formula. 

Summing up, in this article, we significantly advance the state of the art in computational logic by:
\begin{itemize}
	\item
	introducing the Triguarded Fragment: a new, highly expressive fragment of \FO{} subsuming \FOt{} and \GF{} with restricted use of equality,
	\item
	establishing decidability and tight complexities for satisfiability checking in \TGF,
	\item 
	characterizing complexity and expressivity of \TGF{} when seen as a query language,    
	\item 
	showing that a free use of equality in \TGF{} and several other attempted generalizations lead to undecidability,
	\item 
	showing the finite model property of \TGF{}, and thus that finite satisfiability coincides with satisfiability, and
	\item %
	providing a tight upper bound on the size of the finite model.
\end{itemize}
This article is a consolidated, revised, and significantly extended version of earlier work published at LPAR \cite{RS18} and LICS \cite{DBLP:conf/lics/KieronskiR21}.

\section{Notation}\label{s:notation}

\paragraph{Logics}
We assume the reader to be familiar with first-order predicate logic (FO).
We will write $x$, $y,\ldots$ for
variables, $c$ for constants, and $t$ for terms (comprising all variables and constants), all of these possibly with decorations.
We write $\bar{x}$, $\bar{y},\ldots$ for tuples of variables and $\bar{t}$ for tuples of terms. For convenience, such tuples will sometimes be conceived as sets, justifying statements like $\bar{x} \subseteq \bar{y}$. We use lower case Greek letters $\varphi, \psi, \ldots$ to denote logical formulae and we sometimes write $\psi(\bar{x})$ to indicate that all free variables of the formula (subformula) $\psi$ are contained in $\bar{x}$. Given a tuple $\vec{x}$ of variables, an
 \emph{$\vec{x}$-substitution} is any function $f$ from terms to terms such that $f(t)=t$ for all $t\not \in \vec{x}$. Given a tuple
 $\vec{t}= \langle t_1,\ldots,t_n \rangle$ of terms and an
 $\vec{x}$-substitution $f$, we let
 $f(\vec{t}) = \langle f(t_1),\ldots,f(t_n)\rangle $. %

\paragraph{Structures}
We work with finite signatures (typically denoted by $\sigma$, possibly with decorations) containing relation symbols %
of arbitrary nonnegative %
arity %
and, possibly, constant symbols (we do not consider functions of arity greater than~0). 
We refer to structures over finite signatures using Fraktur capital letters $\str{A}, \str{B}, \str{C}, \ldots$ while
the corresponding Roman capital letters $A, B, C, \ldots$ are used to denote their domains. Given a structure $\str{A}$ and some $D \subseteq A$ we
denote by $\str{A} \restr D$  the restriction of $\str{A}$ to its subdomain $D$. 

Given a structure $\str{A}$ interpreting a signature $\sigma$%
, we call the subset $\hat{A} \subseteq A$  consisting of all the elements interpreting the constants of $\sigma$ the \emph{named part} of $\str{A}$. %
The \emph{unnamed} part is defined as $\check{A}  := A \setminus \hat{A}$.

We usually use $a, b$ to denote domain elements of structures, and $\bar{a}$, $\bar{b}$ for tuples of domain elements.
As above, tuples of elements may sometimes be used to denote the sets consisting of their members.

Given a tuple $\vec{x}=\langle x_1,\ldots,x_n \rangle$ of variables,  
an \emph{$\vec{x}$-assignment} is a function $g: \{x_1,\ldots,x_n\} \to A$ mapping the variables of $\vec{x}$ to the domain of some structure $\str{A}$. As usual, we then write $\str{A},g \models \psi(\vec{x})$ to indicate that $\str{A}$ satisfies $\psi(\vec{x})$ under the assignment $g$. Given $\vec{x}=\langle x_1,\ldots,x_n \rangle$ and an $n$-tuple $\vec{a}=\langle a_1,\ldots,a_n \rangle \in A^n$, we let $\vec{x} \mapsto \vec{a}$ denote the  $\vec{x}$-assignment mapping $x_i$ to $a_i$ for every $i \in \{1,\ldots,n\}$.

\section{Syntax of the Triguarded Fragment}\label{s:syntax}

\subsection{Definition and Subsumption of Other Fragments}

We are now ready to introduce the \emph{triguarded fragment} of
$\fo$. Essentially, it is a relaxed variant of the well-known guarded fragment ($\gf$), with the difference that guards are
only required when one quantifies over (sub)formulae with three or more free
variables. The use of equality, on the other hand, is somewhat more restricted than in $\gf$.

\begin{definition} \label{def:tgf} The \emph{triguarded fragment}
	$\TGF$ of first-order logic is defined as the smallest
	set of formulae closed under the following rules:
\begin{enumerate}
\item All atomic formulae of the shape $R(\bar{t})$, $c=c'$, $x=c$, $c=x$, and $x=x$ belong to $\TGF$.
\item If $\varphi$ and $\psi$ are formulae in $\TGF$, then so are $\neg \varphi$, $\varphi \wedge \psi$, $\varphi \vee \psi$, $\varphi \Rightarrow \psi$, and $\varphi \Leftrightarrow \psi$.  		
\item If $\bar{x} \subseteq \bar{y} \subseteq \bar{z}$ are tuples of variables, 
   $\varphi(\bar{y})$ is a formula in $\TGF$, and
   $\gamma(\bar{z})$ is an atomic formula,
   then the formulae $\exists \bar{x}.\gamma(\bar{z}) \wedge \varphi(\bar{y})$
   and $\forall \bar{x}.\gamma(\bar{z}) \Rightarrow \varphi(\bar{y})$ also belong to
   $\TGF$.
\item If $\bar{x}\subseteq\bar{y}$ are tuples of variables,
	$\varphi(\bar{y})$ is a formula in $\TGF$, and $|\bar{y}|\leq 2$, then the formulae
	$\exists \bar{x}.\varphi(\bar{y})$
	and $\forall \bar{x}. \varphi(\bar{y})$ also belong to
	$\TGF$. \defend
\end{enumerate}
\end{definition}

As usual, the atoms $\gamma$ relativising quantifiers in item (3) of the above definition are called the \emph{guards} of the quantifiers.

We observe that, indeed, if we restricted the definition to the items (1), (2), and (3), we would obtain a version of the guarded fragment $\gf$, where equality statements between distinct variables (such as $x=y$) are disallowed.
In what follows, when dealing with $\gf$ formulae, we sometimes take the liberty to omit guards from quantified subformulas having at most one free variable, noting that any formula without free variables or with only one free variable $x$ can be trivially provided with the guard $x=x$.

We further observe that, if we restricted Definition~\ref{def:tgf} to the items (1), (2), and (4), we obtain 
$\fot$, denoting all formulae of $\fo$ that use at most $2$ variables (with the use of equality restricted as above),\footnote{Strictly speaking, we also obtain formulae using more than $2$ variables like $\exists y (R(x,y) \wedge \exists z R(y,z))$, but all these are syntactic variants of $\fot$ formulae as can be shown by an easy variable renaming.} and thus
$\fot\subseteq \TGF$. %

Beyond all $\gf$ and $\fot$ formulae, the syntax of $\TGF$ also allows us to build formulae that are contained in neither of the two fragments, witnessed by formulae like
\begin{equation}\label{eq:ex-tgf-formula}
	\forall x\forall y.( (R_1(x,c)\land R_2(y,c')) \Rightarrow \exists
	z.R_3(x,y,z,c'')).
\end{equation}

It is also noteworthy that, as demonstrated in the introduction, $\TGF$ subsumes one of the most important
\emph{prefix classes}, namely G\"odel's class without equality, consisting of prenex sentences of the shape $\exists \bar{x} \forall y_1 y_2 \exists \bar{z} \varphi$.
We note that, given such a sentence, we can construct a model-theoretic conservative extension in \TGF{} by eliminating the initial prefix of existential quantifiers, replacing the variables $\bar{x}$ by constants (a straightforward application of the technique of Skolemisation), and guarding the block of quantifiers $\exists \bar{z}$ by a ``dummy'' guard $Aux(y_1, y_2, \bar{z})$ using a fresh predicate $Aux$. Along the same lines, decidability of satisfiability of \TGF{} sentences also provides a positive answer to an open question by ten Cate and Franceschet \cite{DBLP:journals/jolli/CateF05}, who asked if the class of $\fo$ sentences of the shape $\exists \bar{x} \forall y_1 y_2 \exists \bar{z} \psi$ with $\psi$ a guarded formula has a decidable satisfiability problem. As can be easily observed, the construction of the model-conservative extension sketched above for sentences of Gödel's class easily extends to ten Cate's and Franceschet's class, and it will produce sentences in $\TGF$. Then, by means of the results in this article, this fragment is decidable and even has the finite model property.

\subsection{$\gfu$: a Convenient Alternative Representation of $\tgf$}\label{ss:gfu}

One main goal in this paper is to investigate model-theoretic properties and the computational
complexity of satisfiability in $\TGF$. To this end, we find it helpful to consider 
a slightly different logic, which is effectively equivalent to $\TGF$,
but which facilitates presentation and understanding. 
To motivate this choice, we observe that there is a simple extension of $\gf$ that allows us to capture $\TGF$.
As we have seen, the reason why $\TGF$ is strictly bigger than $\gf$ is that $\TGF$ allows
``unguarded'' quantification in front of any formula $\varphi$ that has no more than 2 free variables. Now, if we could endow $\gf$ with a ``built-in'' binary predicate whose extension is fixed and always contains all pairs of domain elements, we could use it to provide a guard for any such $\varphi$. 
Consequently, it is convenient to assume that our signature $\sigma$ contains a binary \emph{universal role} predicate
$\ur$, satisfying 
\begin{equation}\label{eq:condition-ur}
\ur^{\str{A}}= A \times A 
\end{equation}
for any
structure $\str{A}$. 
Structures interpreting $\UU$ in this desired way will be called \emph{$\UU$-biquitous}.
Naturally, $\ur$ is not allowed to be used in classical $\gf$ formulae, but note that in $\fo$, $\fot$ and $\TGF$, such a 
built-in predicate $\ur$ would not add expressiveness, because it can be
axiomatized using a fresh ordinary binary predicate $U$ and asserting the $\fot$ sentence
\begin{equation}\label{eq:axiom-ur}
	\phi = \forall x\forall y.U(x,y)
\end{equation}
Thus, we may safely allow the usage of $\ur$ as a predicate symbol in formulae of $\fo$, $\fot$ and
$\TGF$. Since $\phi$ is not in $\gf$, the addition of the built-in
$\ur$ to $\gf$ makes a big difference (as will also become clear from our complexity
results). We now formally define $\gftimes$, which extends $\gf$ with
$\ur$, and in fact adds to $\gf$ the necessary expressivity to capture
$\TGF$.

\begin{definition}
	Let $\gftimes$ be the set of formulae of $\TGF$ that can be built
	using the items (1), (2) and (3) of Definition~\ref{def:tgf} only,
	while possibly using the predicate $\ur$ in atomic formulae. \defend
\end{definition}

By using the $\ur$ predicate as a guard for formulae with at most 2
free variables, we can convert any $\TGF$ formula into an equivalent
formula in $\gftimes$. For instance, the above example formula \eqref{eq:ex-tgf-formula} can be transformed into the equivalent $\gftimes$ 
formula
\begin{equation}
	\forall x\forall y. (\ur(x,y)\Rightarrow ((R_1(x,a)\land R_2(y,b))
	\Rightarrow \exists z.R_3(x,y,z,c))).
\end{equation}

\begin{proposition}\label{prop:twoguarded}
	For any $\varphi\in \TGF$, we can build in polynomial
	time a formula $\varphi'\in \gftimes$ such that $\varphi$ and
	$\varphi'$ are equivalent over $\UU$-biquitous structures and (except for $\UU$) the signature elements used by $\varphi$ and $\varphi'$ coincide.
\end{proposition}

\begin{proof} (Sketch)
	We obtain $\varphi'$ from $\varphi$ by replacing 
	\begin{itemize}
		\item every subformula of the form $\exists x.\varphi(x)$ by
		$\exists x.(\ur(x,x) \wedge \varphi(x))$,
		\item every subformula of the form $\exists x.\varphi(x,y)$ by
		$\exists x.(\ur(x,y) \wedge \varphi(x,y))$,
		\item every subformula of the form $\exists xy.\varphi(x,y)$ by
		$\exists xy.(\ur(x,y) \wedge \varphi(x,y))$,
		
		\item every subformula of the form $\forall x.\varphi(x)$ by
		$\forall x.(\ur(x,x) \Rightarrow \varphi(x))$,
		\item every subformula of the form $\forall x.\varphi(x,y)$ by
		$\forall x.(\ur(x,y) \Rightarrow \varphi(x,y))$,
		\item every subformula of the form $\forall xy.\varphi(x,y)$ by
		$\forall xy.(\ur(x,y) \Rightarrow \varphi(x,y))$.
	\end{itemize}
	In other words, we replace every unguarded quantification
	by a $\ur$-guarded one. It is easy to see that this translation does not
	change the meaning of the formula.
\end{proof}
Conversely, we note that  any $\gftimes$ theory can be efficiently  
translated---while preserving satisfiability---into $\TGF$ by replacing $\ur$ by $U$ everywhere and adding the axiomatization of the binary universal predicate (\ref{eq:axiom-ur}) above. 
Note that the transformations do not just preserve satisfiability but also
modelhood (modulo the freshly introduced $U$). Hence $\TGF$ and $\gftimes$
are both equally expressive and equally succinct.   

In our subsequent constructions, it will sometimes be convenient to interpret \GFU{} formulae over non-$\UU$-biquitous structures. In this case, they are treated as usual \GF{} formulae with $\UU$ being an arbitrary predicate.

\subsection{The Guarded Normal Form}\label{ss:nf}

To simplify the structure of $\gftimes$  formulae, we will make use of a
suitable \emph{(guarded) 
normal form}, which is not much different from normal forms 
used in prior work
\cite{scott1962decision,Gra99}. As this normal form does not treat the symbol $\UU$ in any special way, it can be also
used for ordinary $\gf$.

\begin{definition}[Guarded Normal Form]\label{def:normalform}
		A sentence $\varphi\in \gftimes $ (\GF) is in \emph{guarded normal form} (or just \emph{normal form}, if there is no danger of confusion) if it has
		the form
		$\bigwedge_{\psi \in \mathbf{A}} \psi \land \bigwedge_{\psi \in
			\mathbf{E}} \psi $, where $\mathbf{A}$ contains sentences of the form
		\begin{equation}
			\forall \bar{x}. ( 
			R(\bar{t})   \Rightarrow (\neg H_1(\bar{v}_1)\lor \ldots \lor \neg H_n(\bar{v}_n) \lor H_{n+1}(\bar{v}_{n+1})\lor \ldots \lor H_m(\bar{v}_m))),\label{def:A}
		\end{equation}
		\noindent
		and 
		$\mathbf{E}$ contains sentences of the form
		\begin{equation}
			\forall \bar{x}. (R(\bar{u}) \Rightarrow \exists
			\bar{y}. H(\bar{v})).\label{def:E}
		\end{equation}  %
        Thereby, $R$, $H$, and $H_i$ are predicates from the underlying signature or equality. %
		We use $\mathbf{A}(\varphi)$ and $\mathbf{E}(\varphi)$ to denote the
		sets $\mathbf{A}$ and $\mathbf{E}$ of a formula $\varphi$ as
		above. For a sentence
		$\psi = \forall \bar{x}. (R(\bar{u}) \Rightarrow \exists
		\bar{y}. H(\bar{v})) $ from $\mathbf{E}$, we let $\width(\psi)$ denote the number of
		variables that appear in $\bar{v}$. For a formula $\varphi$ as above,
		$\width(\varphi)$ is the maximal $\width(\psi)$ over all
		$\psi\in \mathbf{E}(\varphi) $. \defend
\end{definition}
	As usual, in case $m=0$, the empty disjunction in
	(\ref{def:A}) stands for $\bot$. Note that since (\ref{def:A}) and (\ref{def:E}) are in $\gftimes$,  each variable that appears in $\bar{v}_1,\ldots,\bar{v}_m$ also appears in $\bar{t}$, and  each variable that appears in $\bar{v}$ also appears in $\bar{u}$. 
	Observe that the sentence in (\ref{def:A}) can be equivalently written as 
	\begin{equation}
		\forall \bar{x}. (
		R(\bar{t})\land   H_1(\bar{v}_1)\land \ldots \land  H_n(\bar{v}_n)   \Rightarrow  H_{n+1}(\bar{v}_{n+1})\lor \ldots \lor H_m(\bar{v}_m)).\label{eq:newnormal}
	\end{equation}
	For presentation reasons, in what follows we may sometimes use the form
	(\ref{eq:newnormal}) instead of (\ref{def:A}) when speaking about
	sentences in $\mathbf{A}$. %
	The informed reader might notice that (\ref{eq:newnormal}) closely
	resembles a (guarded) disjunctive Datalog rule with guard atom $R(\bar{t})$.   
The following statement establishes that in our investigations of $\gftimes$, we can indeed restrict our attention to formulae in normal form.
\begin{proposition}\label{prop:normalization}
	For any sentence $\varphi\in \gftimes \; (\GF)$, we can construct in polynomial
	time a sentence $\varphi'\in \gftimes \; (\GF)$ in normal form (over an extended signature) such that 
\begin{enumerate}[(a)]	
	\item 
	$\varphi'$ is a conservative extension of $\varphi$, i.e., every model of $\varphi'$ is a model of $\varphi$ and every model of $\varphi$ can be expanded to a model of $\varphi'$, and  
	\item 
	the
	translation does not increase the arity of predicates, i.e., there is
	no predicate symbol in $\varphi'$ whose arity is strictly greater
	than the arity of every predicate in $\varphi$.
\end{enumerate}
\end{proposition}

In order to prove Proposition~\ref{prop:normalization}, we provide the following normalization:  

\begin{definition}\label{def:normalization}
	Given a $\gftimes$ ($\gf$) sentence $\chi$, we let $\mathsf{gnf}(\chi)$ denote the sentence $\bigwedge_{\psi \in \mathbf{A}_\chi} \psi \land \bigwedge_{\psi \in
			\mathbf{E}_\chi} \psi $, where the sets $\mathbf{A}_\chi$ and $\mathbf{E}_\chi$ are obtained as follows (introducing fresh predicates $p_\psi$ of arity
			$|\bar{z}|$ for certain subformulae $\psi[\bar{z}]$ with
	free variables $\bar{z}$):
	
	\begin{itemize}
		\item Let $\chi'$ denote the result of eliminating $\Rightarrow$ and $\Leftrightarrow$  from $\chi$ in the standard way and putting the result into negation normal form. 	
		\item We let $\mathbf{A}_\chi$ contain the sentence 
		\begin{equation}\label{gnf:base}
			\Rightarrow p_{\chi'}
		\end{equation}
		\item For every subformula
		$\varphi_1[\bar{x}] \wedge \varphi_2[\bar{y}]$ of $\chi'$, we add
		to $\mathbf{A}_\chi$ the sentences
		\begin{eqnarray}
			\forall \bar{x} \bar{y}. (p_{\varphi_1 \wedge \varphi_2}(\bar{x} \cup \bar{y}) \Rightarrow p_{\varphi_1}(\bar{x}))\\
			\forall \bar{x} \bar{y}. (p_{\varphi_1 \wedge \varphi_2}(\bar{x} \cup \bar{y}) \Rightarrow p_{\varphi_2}(\bar{y}))
		\end{eqnarray}
		
		\item For every subformula
		$\varphi_1[\bar{x}] \vee \varphi_2[\bar{y}]$ of $\chi'$, we add to
		$\mathbf{A}_\chi$ the sentence
		\begin{equation}
			\forall \bar{x}\bar{y}. (p_{\varphi_1 \vee \varphi_2}(\bar{x} \cup \bar{y}) \Rightarrow p_{\varphi_1}(\bar{x}) \vee p_{\varphi_2}(\bar{y}))
		\end{equation}
		
		\item For every subformula of $\chi'$ that is a non-guard atom
		$\alpha(\bar{x})$ we add to $\mathbf{A}_\chi$ the sentence
		\begin{equation}
			\forall \bar{x}.(p_{\alpha(\bar{x})}(\bar{x}) \Rightarrow \alpha(\bar{x}))
		\end{equation}		
		
		\item For every subformula of $\chi'$ that is a negated non-guard atom
		$\neg\alpha(\bar{x})$ we add to $\mathbf{A}_\chi$ the sentence
		\begin{equation}
			\forall \bar{x}.(p_{\neg\alpha(\bar{x})}(\bar{x}) \Rightarrow \neg\alpha(\bar{x}))
		\end{equation}		
		
		\item For every subformula of $\chi'$ with the shape
		$\forall \bar{x}. \neg\alpha(\bar{x},\bar{y}) \vee
		\varphi(\bar{x},\bar{y})$ we add to $\mathbf{A}_\chi$ the sentence
		\begin{equation}
			\forall \bar{x}\bar{y}.(\alpha(\bar{x},\bar{y}) \Rightarrow \neg p_{\forall \bar{x}. \neg\alpha(\bar{x},\bar{y}) \vee \varphi[\bar{x},\bar{y}]}(\bar{y})  \vee  p_{\varphi[\bar{x},\bar{y}]}(\bar{x},\bar{y}))
		\end{equation}		
		
		\item For every subformula of $\chi'$ with the shape
		$\exists \bar{x}. \varphi[\bar{x},\bar{y}]$ we add to $\mathbf{E}_\chi$ the sentence
		\begin{equation}
			\forall \bar{y}.(p_{\exists \bar{x}. \varphi[\bar{x},\bar{y}]}(\bar{y})  \Rightarrow \exists \bar{x}. p_{\varphi[\bar{x},\bar{y}]}(\bar{x},\bar{y}))
		\end{equation}		
		~\defend
	\end{itemize}
	
\end{definition}

\begin{proof}(Proof sketch for Proposition~\ref{prop:normalization})
	We show that setting $\varphi'=\mathsf{gnf}(\varphi)$, which is clearly polytime computable, satisfies all the conditions of Proposition~\ref{prop:normalization}.
    First, we readily observe that $\varphi'$ is in guarded normal form.
	Second, as can be shown by an easy structural induction over GFU formulae, $\mathsf{gnf}(\varphi)$ ensures that every fresh predicate $p_\psi$ occurring in it implies the formula in its subscript, that is, in technical terms:
	  \begin{equation}\label{gnf:ind}
	  \mathsf{gnf}(\varphi) \models \forall \bar{z}.\left( p_\psi(\bar{z}) \Rightarrow \psi[\bar{z}]\right).
	  \end{equation}
    The fact that $\mathsf{gnf}(\varphi)$ is a model-theoretic conservative extension of $\varphi$  follows from the following two insights: 
\begin{itemize}	
	\item[(i)]
	Every model of $\mathsf{gnf}(\varphi)$ is a model of $\varphi$, that is, $\mathsf{gnf}(\varphi) \models \varphi$. This is an easy consequence from (\ref{gnf:base}) and (\ref{gnf:ind}).
	\item[(ii)]	
	Every model $\str{A}$ of
	$\varphi$ can be extended to a model $\str{A}'$ of $\mathsf{gnf}(\varphi)$ by choosing the
	interpretation $p_\psi^{\str{A}'}$ of each auxiliary predicate $p_\psi$ such that it coincides
	with the valid variable assignments of the corresponding subformula, that is, we let 
	$p_{\psi[\bar{z}]}^{\str{A}'}$ contain all those tuples $\bar{a}$  for which $\str{A}$ satisfies 
	$\psi[\bar{z}]$ under the assignment $\bar{z} \mapsto \bar{a}$. 
\end{itemize}
	
	Preservation of the maximal arity follows from the fact that by
	definition, for any set of variables that occur freely in some
	subformula of $\varphi$, there must be a joint guard for all of
	them, hence, for each freshly introduced predicate in $\varphi'$ we
	find a guard predicate in $\varphi$ with the same or higher arity.
\end{proof}

We note that the output of the normalization procedure provided by \Cref{def:normalization} would allow for the definition of an even more constrained normal form, but for our purposes, the one given in \Cref{def:normalform} is sufficient. Moreover, for some technical reasons,
it will be convenient for us to have additionally an even weaker variant.

\begin{definition}[Weak guarded normal form]\label{def:weaknormalform}
We say that a $\GF{}$ ($\GFU$) sentence is in \emph{weak normal form} if
it is a conjunction of sentences of the form $\forall \bar{x} (R(\bar{t}) \Rightarrow \psi)$ and $\forall \bar{x} (R(\bar{u}) \Rightarrow \exists \bar{y} (H(\bar{v}) \wedge \psi))$,
for some quantifier-free $\psi$.
\end{definition}

Obviously, any normal form sentence is a weak normal form sentence, so Proposition \ref{prop:normalization} works also for
weak normal form.

\color{black}

\section{Some Technical Tools}\label{s:tools}

\subsection{The Standard Name Assumption} \label{ss:sna}

The common FO semantics allows for assigning distinct constants to the same domain element, that is, in general, in a structure $\str{A}$, it may well be the case that $c^\str{A} = d^\str{A}$ for constant names $c \neq d$. This fact blocks the convenient move to let constant names interpret themselves (i.e., $c^\str{A} = c$), which would make constant names both syntactic and semantic objects at the same time. %
However, we now show that in our setting it is possible to impose such a self-interpreting treatment of constants, referred to as the \emph{standard name assumption (SNA)} -- at the reasonable cost of some extra nondeterminism. In return, the SNA will simplify our arguments regarding the upper bounds on the complexity of satisfiability and the size of minimal models.%

Assume any FO
formula $\varphi$ over a signature $\sigma$. We say $\varphi$ is
\emph{satisfiable under the SNA}, if $\varphi$ has a model $\str{A}$
such that $c\in A$ and $c^{\str{A}}=c $ for all constants $c\in
\sigma$. We argue that, for our upper complexity bounds, we can focus on satisfiability in $\gftimes$
under the SNA, thanks to the following observation.
\begin{proposition}\label{prop:SNA}
There is a
  non-deterministic polynomial-time transformation mapping $\gftimes$ formulae to $\gftimes$ formulae such that: $\varphi$ is satisfiable without SNA %
  iff there is a run of the transformation that produces a
  formula $\varphi^\sim$ that is satisfiable under the SNA. %
\end{proposition}

\begin{proof} 
  Assume a formula $\varphi$ in \gftimes, and suppose $C$ is the set of all
  constants that appear in $\varphi$. The desired non-deterministic
  polynomial transformation  is as follows:
  \begin{enumerate}
  \item Guess an equivalence relation $\sim$ over $C$.
  \item For every equivalence class $D$ induced by $\sim$, introduce a new
    constant $b_{D}$.
  \item Construct a new formula $\varphi^\sim$ from $\varphi$ by replacing
    every constant $c$ with the constant $b_{[c]_\sim}$, where ${[c]_\sim}$ is the
    (unique) $\sim$-equivalence class that contains $c$. Let $\sigma^\sim$ be the signature of all constant and predicate names occurring in $\varphi^\sim$.
    
  \end{enumerate}
  It is not difficult to see the following:

  \begin{enumerate}[-]
  \item If $\varphi$ is satisfiable, then there is a run of the above
    procedure such that the resulting $\varphi^\sim$ is satisfiable under the
    SNA. Indeed, assume $\varphi$ is satisfiable in a model
    $\str{A}$. Take the equivalence relation $\sim$ such that
    $c\sim d$ iff $ c^{\str{A}} = d^{\str{A}}$, where $c,d\in C$. It
    is easy to see that the $\varphi^\sim$ constructed based on this 
    specific relation $\sim$ is satisfiable under the SNA. In particular, $\varphi^\sim$ is
    satisfied in the structure $\str{A}^\sim$ obtained from $\str{A}$ by, 
    for all $d \in C$, renaming 
    $d^{\str{A}}$ into $b_{[c]_\sim}$.

  \item If $\varphi^\sim$ is satisfiable under the SNA for some
    equivalence relation $\sim$ over $C$, then the original formula
    $\varphi$ is satisfiable. Indeed, if $\str{A}$ is a structure that
    witnesses satisfiability of $\varphi^\sim$ under the SNA, then
    $\varphi$ is satisfied by $\str{A}$ that is modified by setting
    $c^{\str{A}}=b_{[c]_\sim}$ for all constants $c$ occurring in $\varphi$.\qedhere
  \end{enumerate}
\end{proof}

We will impose the SNA in the rest of this paper, with the exception 
of Sections \ref{s:compgfu} and \ref{s:undec} as far as lower bounds are concerned (for our lower bounds, it is not relevant whether or not the SNA is imposed). %

\subsection{Types}\label{ss:types}

Let $\sigma$ be a finite signature and $c_1, \ldots, c_m$ be the list of all constants in $\sigma$.
For some natural number $\ell >0$, an ({\em atomic}) $\ell$-{\em type} (or just \emph{type}, if $\ell$ is irrelevant or clear from the context) over $\sigma$
is a maximal consistent set of atomic or negated atomic formulae over $\sigma$, using $\ell$ variables $x_1, \ldots, x_\ell$ and containing $x_i \not= x_j$, $x_i\not=c_g$ and $c_g \not=c_h$ for all $1 \le i, j \le \ell, 1\le g, h \le m$, $i \not=j, g \not=h$.
We often identify a type with the formula obtained by taking the conjunction over its elements. 
A type $\alpha$ induces also a structure $\overline{\alpha}$ such that (i) the domain of $\overline{\alpha}$ is is the set $\{x_1, \ldots, x_\ell\} \cup \{c_1, \ldots, c_m\}$, and (ii) for each predicate $P \in \sigma$, we have
$P^{\overline{\alpha}}=\{\vec{t}\mid P(\vec{t})\in \alpha\}$.
 
  If $x_i$ is a variable in a type $\alpha$, we let $\alpha\restr x_i$ denote the restriction of $\alpha$ to those formulae containing only terms that are $x_i$ or constants.  We say a type $\alpha_1$ \emph{agrees with a type $\alpha_2$ on} $(x_i,x_j)$, where $x_i$ is a variable in $\alpha_1$ and $x_j$ is a variable in $\alpha_2$, if $\alpha_1\restr x_i$ and $\alpha_2\restr x_j$ coincide after remaining the only  variable in either expression is replaced by $x$. %
We say that a set $\AAA$ of types over a signature $\sigma$ is \emph{consistent}  if all the types in $\AAA$ agree on the ground atoms, that is, for any tuple of (not necessarily distinct) 
constant symbols  $\bar{c}$, a relation symbol $P$ of arity $|\bar{c}|$,
and every $\alpha_1, \alpha_2 \in \AAA$ we have that $P(\bar{c}) \in \alpha_1$ iff $P(\bar{c}) \in \alpha_2$,
and for every nullary relation symbol $P$ we have $P \in \alpha_1 \text{ iff } P \in \alpha_2$.
We say that an $\ell$-type $\alpha$ is $\UU$-\emph{biquitous}, 
if the structure $\overline{\alpha}$ is $\UU$-biquitous.

Let $\str{A}$ be a structure, and let $\bar{a}$ be a tuple of its distinct unnamed elements with $|\bar{a}|=\ell$. 
We denote by $\type{\str{A}}{\bar{a}}$ the unique $\ell$-type \emph{realized} in $\str{A}$ by the tuple  $\bar{a}$, \emph{i.e.}, the unique type $\alpha(\bar{x})$ such that $\str{A},\vec{x}\mapsto \vec{a} \models \alpha(\bar{x})$. We say that $\bar{a}$ is \emph{guarded} in  $\str{A}$ if  it is of length $1$ or  there is a tuple of (not necessarily distinct) elements $\bar{b}$ containing all the
 elements of $\bar{a}$ and a relation symbol $P \in \sigma$ such that $\bar{b} \in P^{\str{A}}$. A tuple of distinct unnamed elements $a_1, \ldots, a_k$ in a structure $\str{A}$ is $\UU$-\emph{biquitous} if $(a_i,a_j) \in \UU^{\str{A}}$ for all $1 \le i,j \le k$.

We will be often interested in types over signatures $\sigma$ consisting of the
relation symbols (and constants, if present) used in some given formula. A particularly important role will be played by $1$-types and $2$-types.
Observe that, in the absence of constants, the number of $1$-types is bounded by a function which is exponential in $|\sigma|$, and hence also in the length
of the formula. This is because any $1$-type essentially just corresponds to a subset of $\sigma$. On the other hand, 
when at least one constant $c$ is present, then the number of $1$-types may be doubly exponentially large. This is because a $1$-type must completely  describe the substructure on a given element and the interpretation of $c$,
and  there are $2^{2^n}$ relations of arity $n$ on a pair of elements.
By an analogous argument, the number of $2$-types may be doubly exponential in the length of the formula even
in the absence of constants.

\subsection{Doublings and joins}

In this subsection we present two basic transformations that, given some models of a \GF{} sentence in weak normal form,
allow us to create another model. The first of these transformations, called ``(partial) doublings'',
is rather typical -- similar constructions were already used in other contexts, see for instance Lemma 6.2.26 in \cite{BGG}. The second one, called ``joins'' is, up to our knowledge, not present in the literature in the form we formulate it, though, 
as we believe, it is quite natural for guarded logics and its correctness proof is rather straightforward.

\begin{definition}
Let $\str{A}$ be a $\sigma$-structure and $A_0 \subseteq \check{A}$. The  \emph{partial doubling} of $\str{A}$ \emph{duplicating} $A_0$, is the structure $\str{B}$
with domain $B:=(A \setminus A_0) \times \{ 0 \} \cup A_0 \times \{0, 1\}$ in which 
(i) for each $P \in \sigma$ and $((a_1, \ell_1), \ldots, (a_k, \ell_k)) \in B$, we have $((a_1, \ell_1), \ldots, (a_k, \ell_k)) \in P^{\str{B}}$ iff $(a_1, \ldots, a_k) \in P^{\str{A}}$, and (ii) if $c \in \hat{A}$ is a constant
then $(c,0)$ interprets $c$ in $\str{B}$.\footnote{Technically, $\str{B}$ does not respect the SNA, since its domain 
does not contain constant symbols. However, as different constant symbols are interpreted by different domain elements, SNA can
easily be restored: it just suffices to rename $(c,0)$ into $c$ for every constant symbol $c$.} %
For $a \in A_0$ the elements $(a, 0)$ and $(a, 1)$ from $B$ are called \emph{twins}.
If $A_0=\check{A}$ then $\str{B}$ is called just the  \emph{doubling} of $\str{A}$.
\end{definition}

\begin{lemma}\label{l:doubling}
	Let $\varphi$ be a  $\GFU{} \; (\GF{})$ sentence in weak normal form 
	using the equality symbol only as allowed in point (1) of Definition \ref{def:tgf},
	let $\str{A}$ be a model of $\varphi$, and let $A_0 \subseteq \check{A}$. Then the partial doubling $\str{B}$ of $\str{A}$ duplicating $A_0$ is still a model of $\varphi$.
\end{lemma} 

\begin{proof}
Consider any $\mathbf{A}$-conjunct $\zeta$ of $\varphi$, $\zeta = \forall \bar{x} (R(\bar{t}) \Rightarrow \psi)$. 
Take an assignment $f:\bar{x} \rightarrow B$ such that $\str{B},f \models R(\bar{t})$. Let $g:B \rightarrow A$ be the function
returning for an element $(a,i)$ its ``origin'' $a$, so $gf$ is an assignment $gf:\bar{x} \rightarrow A$. By the definition of the doublings
we have $\str{A}, gf \models R(\bar{t})$.
In particular, note that this also holds if $R$ is an equality in a permitted form $x=x$ or $x=c$.
As $\str{A} \models \zeta$, we have that $\str{A}, gf \models \psi$. 
As, by the definition of the partial doublings, $\str{A}, gf \models \delta$ iff $\str{B}, f \models \delta$ for any atom $\delta$ of $\psi$,
we have that $\str{B}, f \models \psi$ and we conclude that $\str{B}, f \models \zeta$. Again, the argument about $\delta$ works also if $\delta$ is an equality $x=c$ or $x=x$, but would not work if it was of the form $x=y$ (consider, e.g., the formula $\forall xy (P(x,y) \Rightarrow x=y)$, which is not preserved by doublings).  
		
		Consider now any $\mathbf{E}$-conjunct $\zeta$ of $\varphi$, $\zeta=\forall \bar{x} (R(\bar{u}) \Rightarrow \exists \bar{y} (H(\bar{v}) \wedge \psi))$.
		Take an assignment $f:\bar{x} \rightarrow B$ such that $\str{B}, f \models R(\bar{u})$. Let $g:B \rightarrow A$ be as above. By the definition of the partial doublings, we have that $\str{A}, gf \models R(\bar{u})$.
As $\str{A} \models \zeta$ there is an assignment $g':\bar{x} \cup \bar{y} \rightarrow A$ extending $g f$ such that $\str{A}, g' \models H(\bar{v}) \wedge \psi$. Let $f':\bar{x} \cup \bar{y} \rightarrow B$ be the assignment such that $f'(x)=f(x)$ for $x \in \bar{x}$ (so $f'$ extends $f$)  and $f'(y)=(g'(y), 0)$ for $y \in \bar{y}$. 
Again, reasoning about $H$ and the atoms of $\psi$ as in the case of $\mathbf{A}$-conjuncts  we have that $\str{B}, f' \models H(\bar{v}) \wedge \psi$, so $\str{B} \models \zeta$.  
\end{proof}

We next establish one straightforward but important property of doublings. For a structure $\str{A}$, we call elements $a, a' \in A$  \emph{indistinguishable} in $\str{A}$ if for any relation symbol $P \in \sigma$,
any tuple $\bar{a} \subseteq A$ and any tuple $\bar{a}'$ obtained from $\bar{a}$ by replacing some occurrences of $a$ by $a'$ and 
some occurrences of $a'$ by $a$ we have that $\bar{a} \in P^{\str{A}}$ iff $\bar{a}' \in P^{\str{A}}$. With this, we obtain the following lemma.
\begin{lemma}\label{c:uconnected} Let $\alpha$ be an unnamed $1$-type realized by a domain element $b$ in a (not necessarily $\UU$-biquitous) structure $\str{B}$  which is 
the doubling of some structure $\str{A}$ duplicating the whole $\check{A}$. Then there is a distinct domain element $b'$ realizing $\alpha$ in $\str{B}$ such 
that $b, b'$ are indistiguishable. 
Moreover, if $\alpha$ is $\UU$-biquitous then $b'$ can be chosen so that 
$(b,b'),(b',b)\in \UU^{\str{B}}$.
\end{lemma}
\begin{proof} 
  Assume w.l.o.g. that $b=(a,0)$ for some $a \in \check{A}$ . Take $b'=(a,1)$.
\end{proof}

We remark that Lemma~\ref{l:doubling} actually holds for any first-order sentence using equality in the restricted way as allowed in point (1) of Definition~\ref{def:tgf}. 
The construction in the next subsection is more specific to \GF{}.

The next construction allows one, under some natural conditions, to ``glue'' a collection of models with overlapping domains into a single model.

\begin{definition}
Let $\str{A}_1$ and $\str{A}_2$ be structures over the same signature sharing a common named domain $\hat{A}$, and let $B$ be the intersection of their
unnamed domains. We say that $\str{A}_1$ and $\str{A}_2$ are \emph{joinable} if:
\begin{enumerate}[(i)]
	\item $\str{A}_1 \restr \hat{A} = \str{A}_2 \restr \hat{A}$ and $\str{A}_1$ and $\str{A}_2$ agree on $0$-ary predicates,
  \item for every tuple of distinct unnamed elements $\bar{b} \subseteq {B}$ it holds that $\type{\str{A}_1}{\bar{b}} = \type{\str{A}_2}{\bar{b}}$ or $\bar{b}$ is unguarded in at least one of $\str{A}_1$, $\str{A}_2$. 
	\end{enumerate}
A (possibly infinite) sequence of structures $\str{A}_1, \str{A}_2, \ldots$  sharing a common named domain is \emph{joinable} if every pair of structures in it is joinable.
\end{definition}

\begin{definition}
Let $\str{A}_1, \str{A}_2, \ldots$ be a joinable sequence of structures, and $\hat{A}$ be their common named part. 
We define their \emph{join} as the structure $\str{A}$ with domain $\bigcup_{i} A_i$ 
such that:
\begin{enumerate}
\item  $\str{A} \restr \hat{A} = \str{A}_1 \restr \hat{A}$,  
\item  for any $0$-ary predicate $Q$ we have $Q^{\str{A}} = Q^{\str{A}_1}$,
\item  for every tuple of distinct unnamed elements $\bar{a} \subseteq A$,
if $\bar{a}$ is guarded in some $\str{A}_i$ then  $\type{\str{A}}{\bar{a}} = \type{\str{A}_i}{\bar{a}}$, 
\item  the types of the tuples that are not fully defined in the previous points are completed by setting the truth values of all undefined atoms to $\tt false$.
\end{enumerate}
\end{definition}

It is straightforwad to check that the definition of the join is sound, that is the types of all tuples are defined without conflicts.

\begin{lemma}\label{l:join}
Let $\varphi$ be a \GF{} sentence in weak normal form.
Let $\str{A}_1, \str{A}_2, \ldots$ be a sequence of models of $\varphi$ sharing their named domain.
If they are joinable then their join  $\str{A}$ is also a model of $\varphi$.
\end{lemma}
\begin{proof} 
Consider any $\mathbf{A}$-conjunct $\zeta$ of $\varphi$, $\zeta = \forall \bar{x} (R(\bar{t}) \Rightarrow \psi)$. 
Take an assignment $f:\bar{x} \rightarrow A$ such that $\str{A},f \models R(\bar{t})$. 
Let $\bar{a}$ be the tuple of distinct unnamed elements contained in $f(\bar{t})$ (in some order).
As $\bar{a}$ is guarded in $\str{A}$ (by $R$), by the definition of
the join there is $i$ such that $\bar{a} \subseteq A_i$ and $\type{\str{A}_i}{\bar{a}}=\type{\str{A}}{{\bar{a}}}$. 
This means that in particular
$\str{A}_i, f \models R(\bar{t})$. As $\str{A}_i \models \zeta$ we have that
$\str{A}_i, f \models \psi$. As the set of the free variables of $\psi$ is contained in $\bar{t}$ we have that $\str{A}, f \models \psi$. Hence $\str{A} \models \zeta$.

Consider now any $\mathbf{E}$-conjunct $\zeta$ of $\varphi$, $\zeta=\forall \bar{x} (R(\bar{u}) \Rightarrow \exists \bar{y} (H(\bar{v}) \wedge \psi))$.
Take an assignment $f:\bar{x} \rightarrow B$ such that $\str{A}, f \models R(\bar{u})$. 
Let $\bar{a}$ be the tuple of distinct unnamed elements contained in $f(\bar{u})$ (in some order).
 As $\bar{a}$ is guarded in $\str{A}$ (by $R$), by the definition of
the join there is $i$ such that $\bar{a} \subseteq A_i$ and
$\type{\str{A}_i}{\bar{a}}=\type{\str{A}}{\bar{a}}$. In particular
$\str{A}_i, f \models R(\bar{u})$. As $\str{A}_i \models \zeta$ there is $f': \bar{x} \cup \bar{y} \rightarrow A_i$, extending $f$,
such that $\str{A}_i, f' \models H(\bar{v}) \wedge \psi$.
Let $\bar{b}$ be the tuple of distinct unnamed elements contained in $f'(\bar{v})$ (in some order). 
As $\bar{b}$ is guarded in $\str{A}_i$ (by $H$) we have, by the definition of the join, that $\type{\str{A}}{\bar{b}}=\type{\str{A}_i}{\bar{b}}$.
And since the set of the free variables of $\psi$ is contained in $\bar{v}$ we infer that $\str{A} , f' \models H(\bar{v}) \wedge \psi$. Hence $\str{A} \models \zeta$.
We note that the above reasoning works without problems if $R$, $H$ or some atoms of $\psi$ are equalities in the permitted form $x=x$ or $x=c$, since $\str{A}$ and all the $\str{A}_i$
agree on their named parts.
\end{proof}

We remark that the joins of $\UU$-biquitous structures are usually not $\UU$-biquitous, so in contrast to Lemma \ref{l:doubling},
Lemma \ref{l:join} does not work for \GFU. %
One may also note that Lemma \ref{l:join}  works even if (a guarded use of) equality between variables is permitted.

\section{Tight Complexity Bounds for Satisfiability of TGF (GFU)} \label{s:compgfu}

We recall that any type $\alpha$ induces a
structure $\overline{\alpha}$ such that (i) the domain of $\overline{\alpha}$ is is the set of terms that
appear in $\alpha$, and (ii) for each predicate $P\in \sigma$, we have
$P^{\overline{\alpha}}=\{\vec{t}\mid P(\vec{t})\in \alpha\}$.

\begin{definition}[Type overlap]
    Assume types $\alpha_1,\alpha_2$ over a signature with constants $C$ and let $\vec{x}=(x_{i_1},\ldots,x_{i_k})$ be a tuple of variables from
  $\alpha_1$. We say  $\alpha_1$ \emph{overlaps with  $\alpha_2$  at $\vec{x}$}, if the following conditions are satisfied:
  \begin{enumerate}[(i)]
  \item If $P(\vec{t})\in\alpha_1$ with 
    $\vec{t}\subseteq C\cup \vec{x}$, then
    $P(\vec{v})\in \alpha_2$, where $\vec{v}$ is obtained from $\vec{t}$ by replacing each $x_{i_j}$ with $x_j$.
   \item If $P(\vec{t})\in\alpha_2$ with  $\vec{t}\subseteq C\cup \{x_1,\ldots,x_k\}$, then
    $P(\vec{v})\in \alpha_1$, where $\vec{v}$ is obtained from $\vec{t}$ by replacing each $x_j$ by $x_{i_j}$.   
  \end{enumerate}
\end{definition}

  \begin{definition}[Mosaic]\label{def:mosaic} 
    A \emph{mosaic} $\mathcal{M}$ for a sentence $\varphi\in \gftimes$  in normal form is a set  of $\UU$-biquitous types over the signature of
  $\varphi$, satisfying the following:

  \begin{enumerate}[(A)]

  \item \label{itm:X-size} Each $\alpha\in \mathcal{M}$ is an $\ell$-type with $\ell\leq \width(\varphi)$;  
    
  \smallskip\item \label{itm:const-agreement} $\mathcal{M}$ is consistent, i.e., for all
  $\alpha, \alpha'\in \mathcal{M}$, $\alpha$ and $\alpha'$ agree on ground atomic formulae and negations of ground atomic formulae, i.e.\,they can only differ on formulae that involve variables;
  
  \smallskip\item \label{itm:local-consistency} $\overline{\alpha}\models \psi$ for all
    $\alpha \in \mathcal{M}$ and all $\psi\in \mathbf{A}(\varphi) $;

  \item \label{itm:existential}   Assume a type $\alpha\in \mathcal{M}$,  a formula
    $\forall \vec{x}. (R(\vec{t}) \rightarrow \exists
    \vec{y}. H(\vec{v}) )\in
    \mathbf{E}(\varphi)$ such that  $R(g(\vec{t}))\in \alpha$ for some
    $\vec{x}$-substitution $g$, and let $\vec{z}=(x_{i_1},\ldots,x_{i_\ell})$ be  some enumeration of the variables in $\alpha$ that appear in  $\{g(x)\mid x\in
      \vec{x}\cap \vec{v}\} $. Then  there exists some $\alpha'\in \mathcal{M}$
    such that:
    \begin{enumerate} 
    \item
      $ H(h(\vec{v}'))\in  \alpha'$ for some  $\bar{y}$-substitution $h$, where $\vec{v}'$ is obtained from $g(\vec{v})$ by replacing each $x_{i_j}$ with $x_j$.
    \item $\alpha$ overlaps with $\alpha'$ at $\vec{z}$.
    \end{enumerate}
    
    \item \label{itm:cross-types} If $\alpha_1,\alpha_2\in \mathcal{M}$ such that $y_1$   is a variable in $\alpha_1$ and $y_2$   is a variable in $\alpha_2$, then there exists a 2-type $\alpha\in \mathcal{M}$  such that 
      \begin{inparaenum}[(i)]
        \item $\alpha_1$ agrees with $\alpha$ on $(y_1,x_1)$,
      \item $\alpha_2$ agrees with $\alpha$ on $(y_2,x_2)$.
      \end{inparaenum}

  \end{enumerate}

\end{definition}

Intuitively, the conditions (A-E) ensure the following:
(\ref{itm:X-size}) requires that only a small number of
``placeholder'' variables are used in types. Note that types in mosaics
only refer to the original signature of the formula and the small number of placeholder
variables. The condition (\ref{itm:X-size}) is important to ensure the
relatively small size of mosaics. The condition
(\ref{itm:const-agreement}) forces the types to agree on the
participation of constants in predicates.  The condition
(\ref{itm:local-consistency}) requires each type to (locally) satisfy
all sentences from $\mathbf{A}(\varphi)$. The condition
(\ref{itm:existential}) ensures that for each type locally satisfying
the body of some sentence from $\mathbf{E}(\varphi)$, we find a
matching type where also the head of that sentence is
satisfied. Recall that all types in a mosaic are $\UU$-biquitous, i.e. we
require $\ur$ to be correctly interpreted locally (i.e., within the
individual types). Using (\ref{itm:cross-types}) we make sure that any
two representatives of unnamed domain elements (in terms of unary
types) found across the types also occur together in one type.

The following soundness and completeness theorems show that mosaics
properly characterize satisfiability of equality-free $\gftimes$ formulae (and, due to Proposition~\ref{prop:twoguarded}, of equality-free $\TGF$ formulae).

\begin{theorem}[Completeness]\label{thm:completeness}
  Let $\varphi\in \gftimes$ be a formula in normal form. If $\varphi$
  is satisfiable, under the SNA
  then there exists a mosaic $\mathcal{M}$ for
  $\varphi$.
\end{theorem}

\begin{proof}
  Let $\sigma$ be the signature of $\varphi$, i.e.  $\sigma$ is the
  set of all predicate symbols and constants that appear in
  $\varphi$. Assume a structure $\str{A}$ over $\sigma$ such that
  $\str{A}$ is a model of $\varphi$.  Let $\str{B}$ be the doubling of
  $\str{A}$. Recall that $\str{B}$ is over $\sigma$, and that it is a
  model of $\varphi$ due to Lemma~\ref{l:doubling}. We show how to
  extract from $\str{B}$ a mosaic $\mathcal{M}$ for $\varphi$. We let
  \[\mathcal{M}=\big\{\type{\str{B}}{b_1,\ldots,b_k} \ \big|\   \{b_1,\ldots,b_k\} \subseteq \check{B} ,\ |\{b_1,\ldots,b_k\}| = k ,\   0 \leq k\leq
    \width(\varphi) \big\}.\] 
  We now proceed to show that the constructed
  $\mathcal{M}$ is a mosaic for $\varphi$. Note that, since $\str{B}$
  is a model of $\varphi$ and is over the signature $\sigma$, we
  immediately obtain that all types in $\mathcal{M}$ are
  $\UU$-biquitous types over the signature of $\varphi$.

      Now, we separately verify the properties from
      Definition~\ref{def:mosaic}:
  \begin{enumerate}[(A)]
  	
  \item  Due to the construction of $\mathcal{M}$ above, each
    $\alpha\in \mathcal{M}$ is trivially an $\ell$-type for which 
    $\ell\leq \width(\varphi) $ holds.
  	
  \item Independently of the specific $\bar{b}\in (\check{B})^k$ picked, 
    any $\type{\str{B}}{\bar{b}}$ will contain all constants that appear
    in $\varphi$ by definition. Then, due
    to the way types are extracted from $\str{B}$, it must be the case that $\alpha$
    and $\alpha'$ agree on ground atomic formulae and negations of
    ground atomic formulae, for any $\alpha, \alpha'\in \mathcal{M}$.

 \item 	$\overline{\alpha}\models \psi$ for all
  	$\alpha \in \mathcal{M}$ and all $\psi\in \mathbf{A}(\varphi) $.
  	
        Suppose $\alpha= \type{\str{B}}{\bar{b}} $ for a tuple
        $\bar{b}=(b_1,\ldots,b_k)$ of elements in $\check{B}$, where
        $k\leq \width(\varphi)$. Observe that
        $\str{B}\restr (\hat{B} \cup \bar{b})\models \psi$ iff
        $\overline{\alpha}\models \psi$. Recall that $\str{B} \models \psi$ by
        assumption. Since $\psi$ is a formula in prenex form with only
        universal quantifiers, it holds that
        $\str{B}\restr D \models \psi $ for any $D\subseteq B$. Thus
        $\overline{\alpha}\models \psi$.

      \item Assume a type $\alpha\in \mathcal{M}$ and a formula
    $\forall \vec{x}. (R(\vec{t}) \Rightarrow \exists
    \vec{y}. H(\vec{v}) )\in
    \mathbf{E}(\varphi)$ such that  $R(g(\vec{t}))\in \alpha$ for some
    $\vec{x}$-substitution $g$. Let $\vec{z}=(x_{i_1},\ldots,x_{i_\ell})$ be  some enumeration of the variables in $\alpha$ that appear in  $\{g(x)\mid x\in
      \vec{x}\cap \vec{v}\} $. 
    We must show that there exists  $\alpha'\in \mathcal{M}$
    such that:
    \begin{enumerate} 
    \item
      $ H(h(\vec{v}'))\in  \alpha'$ for some  $\bar{y}$-substitution $h$, where $\vec{v}'$ is obtained from $g(\vec{v})$ by replacing each $x_{i_j}$ with $x_j$.
    \item $\alpha$ overlaps with $\alpha'$ at $\vec{z}$.
    \end{enumerate}
   
    Suppose $\alpha= \type{\str{B}}{\bar{b}} $ for a tuple
    $\bar{b}=(b_1,\ldots,b_k)$ of elements in $\check{B}$, where
    $k\leq \width(\varphi)$. Let $u$ be the function that maps every
    constant $c$ in $\varphi$ to $c^{\str{B}}$ and every variable
    $x_1,\ldots,x_k$ to $b_1,\ldots,b_k$, respectively. From
    $R(g(\vec{t}))\in \alpha$ it follows that
    $u(g(\bar{t})) \in R^\str{B}$.  Then, since $\str{B}$ satisfies
    the above formula, there exists an extension $f$ of $u$ that
    additionally maps every existentially quantified variable $y$ in
    $\bar{y}$ to an object $f(y)\in B$, and such that
    $f(\bar{v})\in H^{\str{B}}$. Let
    $\bar{d}=(b_{i_1},\ldots,b_{i_{\ell}},d_{1},\ldots,d_w)$ be such
    that $d_{1},\ldots,d_w$ is an enumeration of domain elements that
    appear in
    $\{e\in \check{B} \mid f(y)=e\land y\in \bar{y}\cap\bar{v} \}$.
    Let $\alpha'=\type{\str{B}}{\bar{d}}$. It is not difficult to
    verify that $\alpha'$ is a desired type that satisfies (a) and
    (b). For (a), $\alpha$ overlaps with $\alpha'$ at $\vec{z}$ due to
    fact that $\alpha= \type{\str{B}}{\bar{b}} $ and
    $\alpha'=\type{\str{B}}{\bar{d}}$ with the first $\ell$ elements
    of $\vec{d}$ being $b_{i_1},\ldots,b_{i_{\ell}}$.  For (b), take
    the $\vec{y}$-substitution $h$ such that, for all $y\in \vec{y}$ we
    have: (i) if $f(y)=d_i$ for some $1\leq i \leq w$, then
    $h(y)=x_{\ell + i}$, and (ii) if $f(y)=c^{\str{B}}$ for some
    constant $c$, then $h(y)=c$. Then $ H(h(\vec{v}'))\in \alpha'$
    follows from the fact that $f(\bar{v})\in H^{\str{B}}$.

        \item %
        Assume types $\alpha_1,\alpha_2\in \mathcal{M}$ such that $x_i$  is a variable in $\alpha_1$ and $x_j$   is a variable in $\alpha_2$. We must show that there  exists some $\alpha\in \mathcal{M}$ with variables $x_{i'},x_{j'}$ such that 
      \begin{inparaenum}[(i)]
      \item $x_{i'}\neq x_{j'}$,
      \item $\alpha_1$ agrees with $\alpha$ on $(x_{i},x_{i'})$,
      \item $\alpha_2$ agrees with $\alpha$ on $(x_{j},x_{j'})$.
      \end{inparaenum}
      Suppose $\alpha_1= \type{\str{B}}{\bar{c}} $ and
      $\alpha_2= \type{\str{B}}{\bar{d}} $ for tuples
      $\bar{c}=(c_1,\ldots,c_k)$ and $\bar{d}=(d_1,\ldots,d_{\ell})$
      of elements in $\check{B}$, where $k,\ell \leq \width(\varphi)$.
      
      First, assume  $c_i\neq d_j$.  Then take the type $\alpha= \type{\str{B}}{(c_i,d_j)}$ and let
      $x_{i'}=x_1$, $x_{j'}=x_2$. Then $\alpha$ immediately satisfies
      all requirements (i)-(iii).

      Now, suppose $c_i = d_j$. Since $\str{B}$ is the doubling of $\str{A}$,  we can find an element $e\in \check{B}$ such that $d_j\neq e$ while $d_j$ and $e$ realize the same 1-type in $\str{B}$.  Then take the type $\alpha= \type{\str{B}}{(c_i,e)}$ and let
      $x_{i'}=x_1$, $x_{j'}=x_2$. Then again $\alpha$ immediately satisfies
      all requirements (i)-(iii). \qedhere
      	
  \end{enumerate}
  
\end{proof}

\begin{theorem}[Soundness]\label{thm:soundness}
  Let $\varphi\in \gftimes$ be a formula  in normal form. If there
  exists a mosaic $\mathcal{M}$ for $\varphi$, then $\varphi$ is satisfiable under the SNA.
\end{theorem}

\begin{proof}
  Assume a mosaic $\mathcal{M}$ for $\varphi$. %
  Let $C$ be the set of constants
  that appear in $\varphi$ and let $D$ be a countably infinite set
of elements such that $C\subseteq D$. One may think of $D$ as a pool of potential domain elements. An \emph{instantiation} for a type
  $\alpha\in \mathcal{M}$ is any injective function $\delta$ that maps
  all terms of $\alpha$ to objects in $D$ such that $\delta(c)=c$ for
  all $c\in C$. We let $\delta(\alpha)$ denote the result of replacing in $\alpha$ every variable $x$ with $\delta(x)$. Given a sequence
  $\Smc = (\alpha_0,\delta_0),(\alpha_1,\delta_1),\ldots$ of pairs
  $(\alpha_j,\delta_j)$, where $\alpha_j\in \mathcal{M}$ and $\delta_j$ is
  an instantiation for $\alpha_j$, we use $\str{A}_\Smc$ to denote the
  following structure:
  \begin{enumerate}[-]
  \item the domain of $\str{A}_\Smc$ is
    $A_\Smc=C\cup \bigcup_{i \geq 0 } ran(\delta_i)$,
  \item $c^{\str{A}_\Smc}=c$, for all $c\in C $, and
  \item
    $R^{\str{A}_\Smc}= \bigcup_{i \geq 0 } \{\delta_i(\vec{t}) \mid
    R(\vec{t})\in \alpha_i \} $, for all predicates $R$ of $\varphi$.
  \end{enumerate}
  For a finite initial segment $(\alpha_0,\delta_0),\ldots,(\alpha_i,\delta_i)$ of $\Smc$, we define $\str{A}_{\Smc}^{i}$ in the obvious way: it is the substructure of $\str{A}_{\Smc}$ obtained by restricting its domain to $C\cup \bigcup_{0 \leq k \leq i} ran(\delta_k)$. %
  
  Our goal is to inductively construct a possibly infinite sequence $\Smc$ as above such that  $\str{A}_\Smc \models \varphi$.

  \smallskip
  
  As base case, the first element $(\alpha_0,\delta_0)$ of $\Smc$
  is obtained by setting $\alpha_0$ to be an arbitrary type from
  $\mathcal{M}$, and letting $\delta_0$ be any instantiation for
  $\alpha_0$.

  \smallskip For the inductive step, suppose
  $(\alpha_0,\delta_0),\ldots,(\alpha_{i-1},\delta_{i-1})$ have been
  defined, where $i > 0$. We show how to define the next segment
  $(\alpha_i,\delta_i),\ldots,(\alpha_m,\delta_m)$ of $\Smc$, where $m\geq i$
  (we indeed may attach to $\Smc$ multiple new elements in one step). To
  this end, choose the smallest index $0 \leq j \leq i -1$ satisfying
  the following condition: there is
  $\forall \vec{x}. (R(\vec{t}) \Rightarrow \exists
  \vec{y}. H(\vec{v}))\in
  \mathbf{E}(\varphi)$ such that $g(\vec{t})  \in R^{\str{A}_\Smc^{j}}$ for
  some $\vec{x}$-assignment $g$ and there is no $\vec{x}{\cup}\vec{y}$-assignment $u$ extending $g$ such that $u(\vec{v})\in H^{\str{A}_\Smc^{i-1}}$. Intuitively, $j$ corresponds to the ``earliest'' violation of some formula from $\mathbf{E}(\varphi)$. If such $j$ does not exist, the
  construction of $\Smc$ is complete, and we can proceed to $(\star)$
  below, where we argue that $\str{A}_\Smc \models \varphi$. We
  assume that the above $j$ exists. We first show in ($\dagger$) how
  to define $(\alpha_{i},\delta_{i})$, and then in ($\ddagger$) how to
  define the remaining
  $(\alpha_{i+1},\delta_{i+1}),\ldots,(\alpha_m,\delta_m)$.

  \smallskip ($\dagger$) Take an  $\vec{x}$-substitution $h$ such that $h(x)=\delta_j^{-1}(g(x))$ for all variables $x$ in $\vec{x}$. Due to the definition of $\str{A}_\Smc$, we have that $R(\delta_j^{-1}(g(\vec{t})))\in \alpha_j$. Thus also $R(h(\vec{t}))\in \alpha_j$. Then the condition (D) in Definition~\ref{def:mosaic} holds. Let $\vec{z}=(x_{i_1},\ldots,x_{i_\ell})$ be some enumeration of the variables in $\alpha_j$ that appear in  $\{h(x)\mid x\in
      \vec{x}\cap \vec{v}\} $. Then, by definition, there must exist a  $\alpha'\in \mathcal{M}$
    such that:
    \begin{enumerate} 
    \item $ H(w(\vec{v}'))\in \alpha'$ for some $\bar{y}$-substitution
      $w$, where $\vec{v}'$ is obtained from $h(\vec{v})$ by replacing
      each $x_{i_j}$ with $x_j$.
    \item $\alpha_j$ overlaps with
      $\alpha'$ at $\vec{z}$.
    \end{enumerate}
    We let $\alpha_i=\alpha'$, and define an instantiation $\delta_i$ for $\alpha_i$
    as follows. First, we set $\delta_{i}(x_1)=g(x_{i_{1}}),\ldots,\delta_{i}(x_\ell)=g(x_{i_{\ell}})$. For the remaining variables in $x $ in $\alpha_i$, we set $\delta_i(x)=c_x$ to a a fresh constant from $D\setminus C$.

  \smallskip ($\ddagger$) Let $N$ be the set of all constants that
  were freshly introduced in $\Smc$ by $\delta_i$, i.e., $N$ is the set
  of all $\delta_i(x)$ such that the variable $x$ appears in $\alpha_i$ but not among $x_1,\ldots,x_{\ell}$. Intuitively, in order to properly deal with the
  $\ur$ predicate, we need to find in $\mathcal{M}$ appropriate types to
  connect every $c\in N$ with the each of the constants of the
  sequence $\Smc$ constructed so far.
  Let $(d_1,d_1'),\ldots,(d_n,d_n')$ be an enumeration of all pairs
  $(d,d'$) such that $d\in N$ and
  $d'\in (\bigcup_{0 \leq k \leq i-1} ran(\delta_{k}))\setminus C$. That is, $d'$ is
  any constant that appears in the sequence $\Smc$ constructed so far but
  $d'\not\in N\cup C$.  The definition of the
  segment $(\tau_{i+1},\delta_{i+1}),\ldots,(\tau_m,\delta_m)$ of $S$
  in this inductive step is as follows. We let $m=i+n$, and for each
  $1 \leq k\leq n$, we select $(\tau_{i+1+k},\delta_{i+1+k})$ as
  described next.

  \smallskip  Assume an arbirary $1 \leq k \leq n$.  Note that
  $\delta_{i}^{-1}(d_k)$ is a variable in $\alpha_i$ and thus let
  $y=\delta_{i}^{-1}(d_k)$.  Let $\alpha = \alpha_r$ for some
  $0 \leq r \leq i$ such that $d_k'\in ran(\delta_{r})$. Again,
  $\delta_{r}^{-1}(d_k')$ is a variable of $\alpha_r$, thus let
  $y'=\delta_{r}^{-1}(d_k')$. Due to Condition (\ref{itm:cross-types})
  in the definition of mosaics, there exists a 2-type $\alpha^{*}\in \mathcal{M}$  such that 
      \begin{inparaenum}[(i)]
        \item $\alpha_i$ agrees with $\alpha^{*}$ on $(y,x_1)$,
        \item $\alpha$   agrees with $\alpha^{*}$ on $(y',x_2)$.
      \end{inparaenum}
  Then we set $\alpha_{i+1+k}=\alpha^*$, and let
  $\delta_{i+1+k} = \{x_1 \mapsto d_k, x_2 \mapsto d_k')\}$.
  
  \medskip

 $(\star)$  The above completes the construction of a candidate model
  $\str{A}_\Smc$ for $\varphi$. We now show that  $\str{A}_\Smc$ is indeed a model of $\varphi$.
  
  First, observe that $\vec{v} \in P^{ \str{A}_\Smc}$  if and only if there exists some $i$ such that  $P(\vec{t}) \in \alpha_i$ and $\delta_i(\vec{t})=\vec{v}$. Due to the definition of mosaics and the construction of $\Smc$, we have $\vec{v} \in P^{\str{A}_\Smc}$  if and only if  $P(\delta_j^{-1}(\vec{v})) \in \alpha_j$  for \textbf{all} $j$  such that $\vec{v} \subseteq \ran(\delta_j)$.

We note two consequences of our construction:  
\begin{enumerate}[(i)]
	\item Every $\delta_i(\tau_i)$ is an induced substructure of $\str{A}_\Smc$.
	\item \label{props:twotogether} For any two elements $e_1,e_2 \in A_{\Smc}$, there is at least one $(\alpha_i,\delta_i)$ in $\Smc$ with $\{e_1,e_2\} \subseteq \ran(\delta_i)$.
\end{enumerate}

We now show that the predicate $\ur$ is interpreted in the intended
way. Indeed, we obtain
$\ur^{ \str{A}_\Smc} = A_{\Smc} \times A_{\Smc}$ from the fact that
all types in $\Smc$ are $\UU$-biquitous and the fact
(\ref{props:twotogether}) above, i.e., that any two elements of
$A_{\Smc}$ co-occur in some type of $\Smc$.

Next we show $ \str{A}_\Smc \models \varphi$. 

We start with some sentence $\psi=\forall \vec{x}. (B_1(\vec{t}_1)\land \ldots \land B_n(\vec{t}_n)
\Rightarrow H_1(\vec{v}_1)\lor \ldots \lor H_m(\vec{v}_m)) $ coming from $\mathbf{A}(\varphi)$.
W.l.o.g. we assume $B_1(\vec{t}_1)$ to be the guarding atom, i.e. $\vec{x} \subseteq \vec{t}_1$.
Now assume there is an arbitrary $\vec{x}$-assignment $f$ such that $f(\vec{t}_i) \in B_i^{\str{A}_\Smc}$ for $1 \leq i \leq n$.
Let $(\alpha,\delta)$ be the  element in $\Smc$ for which $B_1(\delta^{-1}(f(\vec{t}_1))) \in \alpha$. Due to the construction of $\str{A}_\Smc$, and, in particular, the fact that $\ran(\delta)$ includes $f(\vec{t})$ and $C$, we have that $B_i(\delta^{-1}(f(\vec{t}_i))) \in \alpha$ for $i>1$. Due to Condition~(\ref{itm:local-consistency}), we have $\alpha \models \psi$, therefore
$H_j(\delta^{-1}(f(\vec{v}_j))) \in \alpha$ for some $j$ with $1\leq j\leq m$, which implies $f(\vec{v}_j) \in H_j^{\str{A}_\Smc}$. Hence we have shown $\str{A}_\Smc \models \psi$.

Now consider some sentence $\psi= \forall \vec{x}. (R(\vec{t}) \Rightarrow \exists
\vec{y}. H(\vec{v}))$ coming from $\mathbf{E}(\varphi)$. 
Assume  an arbitrary $\vec{x}$-assignment $g$ such that $g(\vec{t}) \in R^{\str{A}_\Smc}$.
Let $(\alpha_i,\delta_i)$ be the  element in $\Smc$ for which $R(\delta_i^{-1}(g(\vec{t}))) \in\alpha_i$. By construction ($\dag$) of $\Smc$ there must exist some $j>i$ such that  $(\alpha_j,\delta_j)$ is an element in $\Smc$ for which 
$H(\delta_j^{-1}(f(g(\vec{v}))))\in \alpha_j$ for some $\vec{y}$-assignment $f$. Hence we have shown $\str{A}_\Smc \models \psi$.
\end{proof}

Based on the above characterization of satisfiability via mosaics, we can infer worst-case optimal upper bounds for satisfiability checking in $\gftimes$, and thus in $\TGF$. 
\begin{theorem}[Complexity]\label{thm:upper-bound}
  Deciding satisfiability of $\,\TGF$ and of $\,\gftimes$ formulae 
  is \twonexptime-complete.  The problem is \nexptime-complete under
  the assumption that predicate arities are bounded by a constant.
\end{theorem}

\begin{proof}
Due to Propositions~\ref{prop:twoguarded} and
\ref{prop:normalization}, it suffices to show the two upper bounds for
$\gftimes$ formulae in normal form. Due to Proposition~\ref{prop:SNA},
we can focus on the satisfiability of $\gftimes$ formulae in normal
form under the SNA.  Due to Theorems \ref{thm:completeness} and
\ref{thm:soundness}, we can decide the satisfiability of a formula
$\varphi\in \gftimes$ in normal form under the SNA by checking the
existence of a mosaic $\mathcal{M}$ for $\varphi$. Our approach is to
non-deterministically guess a set $\mathcal{M}$ containing $\ell$-types over
the the signature of $\varphi$ with $\ell \leq \width(\varphi)$, and
then verify that $\mathcal{M}$ is indeed a mosaic for $\varphi$. It is easy to see that, given a candidate $\mathcal{M}$ as input, we can check in
polynomial time whether $\mathcal{M}$ satisfies all the conditions
given in Definition~\ref{def:mosaic}. Note that, if $C$ is the set of
constants in $\varphi$ and $\ell \leq \width(\varphi)$, then the
number of ground atoms over the signature of $\varphi$ with arguments
from $T=C\cup \{x_1,\ldots,x_\ell\}$ is bounded by $m \cdot |T|^k$,
where $m$ is the number of predicates in $\varphi$ and $k$ is the
maximal arity of predicates in $\varphi$. Consequently, we can
restrict ourselves to candidates $\mathcal{M}$ that have no more than
$2^{m \cdot |T|^k}$ types. Since this bound is double exponential in
the size of $\varphi$, but only single exponential under the
assumption that $k$ is a constant, the two upper bounds follow.

  The matching lower bound for the bounded arity case follows from the
  complexity of $\fot$
  \cite{DBLP:journals/bsl/GradelKV97}. For the case of unbounded arity, we establish the lower bound directly hereafter.
\end{proof}  

We dedicate the remainder of the section to the proof of the lower complexity bound in the case of unbounded arities. 
For the case where the predicate arities are not bounded, we will show \twonexptime-hardness 
via a polynomial reduction from the tiling problem of a grid of doubly exponential
size \cite{BGG}.

Let $k$ be a natural number. We will construct a $\gfu$ sentence $\varphi_k$ that has polynomial length wrt. $k$ and describes a tiling of a $2^{(2^k)} \times 2^{(2^k)}$ grid. The sentence will enforce the presence of domain elements corresponding to elements of this grid. To identify the position of each of those grid elements, we have to assign them x- and y-coordinates between $0$ and $ 2^{(2^k)}-1$. We will express these coordinates in binary encoding, i.e., as $2^k$-dimensional bitvectors. In order to express that the $\ell$th position in the bitvector corresponding to the x-coordinate of some grid element $e$ carries a $0$, we let
$\mathit{Sel}_0(e,0,\ell_\mathrm{binary})$ hold, where $\ell_\mathrm{binary}$ is a list of length $k$ containing $0$s and $1$s, expressing the binary encoding of $\ell$. That is, the arity of $\mathit{Sel}_0$ is $k+2$. In order to express that the $\ell$th bit is $1$, we use $\mathit{Sel}_1(e,0,\ell_\mathrm{binary})$. To express the corresponding information for the y-coordinate, we use $\mathit{Sel}_0(e,1,\ell_\mathrm{binary})$ and $\mathit{Sel}_1(e,1,\ell_\mathrm{binary})$, respectively.

As can be seen in the description above, the GFU sentence we create makes use of two distinguished constants, $0$ and $1$.  
In the following, we omit leading universal quantifiers; all formulae are sentences.
The final GFU sentence $\varphi_k$ is then obtained as the conjunction of all the sentences provided below. 

First, we will make sure that for every pair of x- and y-coordinates in the considered range, a corresponding grid element exists. This is achieved by creating a binary tree structure of exponential depth, where at the $\ell$th level two $\mathit{Next}$-children are created: one where the $\ell$th bit is set to $0$ and one where it is set to $1$. All the previously set bits are propagated toward the leaves of the tree. These leaf elements then are endowed with full bit representation of x- and y-coordinates, and hence they will serve as representatives of the grid elements.

We create the tree root.
{\small
	\begin{equation}
		\exists x. \mathit{ToSelect}(x,0^{k+1})
	\end{equation}
}%
At the $\ell$th level ($\vec{z}=\ell_\mathrm{binary}$), two child nodes are created with the $\ell$th bit set to $0$ and $1$, respectively. For guardedness reasons, the $\mathit{Next}$ predicate -- whose purpose is to link parent with child nodes in the tree -- is of arity $k+3$, although only the two leading positions really matter. 
{\small
	\begin{eqnarray}
		\mathit{ToSelect}(x,\vec{z}) & \Rightarrow & \exists y. \mathit{Next}(x,y,\vec{z})
\wedge \mathit{Sel}(y,\vec{z}) \wedge \mathit{Sel}_0(y,\vec{z})\\ 
		\mathit{ToSelect}(x,\vec{z}) & \Rightarrow & \exists y. \mathit{Next}(x,y,\vec{z})
\wedge \mathit{Sel}(y,\vec{z}) \wedge \mathit{Sel}_1(y,\vec{z}) 
	\end{eqnarray}
}%

At the $\ell{+}1$st level (where the $\ell$th bit has just been selected), we indicate that in the next step, the $\ell{+}1$st bit is to be selected ($0\leq i \leq k-1$).
{\small
	\begin{equation}
		\mathit{Sel}(x,\vec{z},0,1^i) \Rightarrow \mathit{ToSelect}(x,\vec{z},1,0^i)
	\end{equation}
}%
The next two rules create versions of the $\mathit{Next}$ predicate which carry all possible $(k+1)$ary bitvectors as additional parameters. Again, this is required for guardedness reasons. 
{\small
	\begin{eqnarray}
		\mathit{Next}(x,y,\vec{z}) & \Rightarrow & \mathit{Next}(x,y,0^{k+1})\\ 
		\mathit{Next}(x,y,\vec{z},0,1^i) & \Rightarrow & \mathit{Next}(x,y,\vec{z},1,0^i) 
	\end{eqnarray}
}%
We propagate earlier choices made for the bits along the $\mathit{Next}$ predicate, making use of the auxiliary positions, to keep everything guarded.
{\small
	\begin{equation}
		\mathit{Sel}_b(x,\vec{z}) \wedge  \mathit{Next}(x,y,\vec{z})
		\Rightarrow \mathit{Sel}_b(y,\vec{z})
	\end{equation}
}%
Once the last bit is set, we indicate that we have reached a leaf.
{\small
	\begin{equation}
		\mathit{Sel}(x,1^{k+1}) \Rightarrow \mathit{Complete}(x)
	\end{equation}
}%
This finishes the creation (and ``coordinatization'') of the grid elements. In the next step, we enforce that any two grid elements with the same y-coordinate and subsequent x-coordinates are connected via a $H$ predicate ($b \in \{0,1\}$ and $0\leq i \leq k-1$). 
{\small
	\begin{eqnarray}
		\ur(x,y)\wedge\mathit{Complete}(x) \wedge \mathit{Complete}(y) & \Rightarrow & \mathit{ChkH}(x,y,0,0^k)
		\label{h22}
		\\
		\mathit{ChkH}(x,y,0,\vec{z},0,1^i) 
		\wedge \mathit{Sel}_1(x,0,\vec{z},0,1^i)
		\wedge \mathit{Sel}_0(y,0,\vec{z},0,1^i)
		&\Rightarrow&
		\mathit{ChkH}(x,y,0,\vec{z},1,0^i)
		\label{h23}
		\\
		\mathit{ChkH}(x,y,0,\vec{z},0,1^i) 
		\wedge \mathit{Sel}_0(x,0,\vec{z},0,1^i)
		\wedge \mathit{Sel}_1(y,0,\vec{z},0,1^i)
		&\Rightarrow&
		\mathit{ChkH}'(x,y,0,\vec{z},1,0^i)
		\label{h24}
		\\
		\mathit{ChkH}'(x,y,0,\vec{z},0,1^i) 
		\wedge \mathit{Sel}_b(x,0,\vec{z},0,1^i)
		\wedge \mathit{Sel}_b(y,0,\vec{z},0,1^i)
		&\Rightarrow&
		\mathit{ChkH}'(x,y,0,\vec{z},1,0^i)
		\label{h25}
		\\
		\mathit{ChkH}'(x,y,0,1^k) 
		\wedge \mathit{Sel}_b(x,0,1^k)
		\wedge \mathit{Sel}_b(y,0,1^k)
		&\Rightarrow&
		\mathit{ChkH}''(x,y,1,0^k)
		\label{h26}
		\\
		\mathit{ChkH}''(x,y,1,\vec{z},0,1^i) 
		\wedge \mathit{Sel}_b(x,1,\vec{z},0,1^i)
		\wedge \mathit{Sel}_b(y,1,\vec{z},0,1^i)
		&\Rightarrow&
		\mathit{ChkH}''(x,y,1,\vec{z},1,0^i)
		\label{h27}
		\\
		\mathit{ChkH}''(x,y,1,1^k) 
		\wedge \mathit{Sel}_b(x,1,1^k)
		\wedge \mathit{Sel}_b(y,1,1^k)
		&\Rightarrow& H(x,y)
		\label{h28}
	\end{eqnarray}
}%
Thereby, the atom $\mathit{ChkH}(x,y,0,\ell_\mathrm{binary})$ is supposed to hold for all grid elements $x$ and $y$ for which the lowest $\ell$ bits of their x-coordinate
have the shape $1^\ell$ and $0^\ell$, respectively. 
Moreover, the atom $\mathit{ChkH}'(x,y,0,\ell_\mathrm{binary})$ is supposed to hold for all grid elements $x$ and $y$ for which the lowest $\ell$ bits of $x$ and $y$ represent consecutive binary numbers. Finally, $\mathit{ChkH}''(x,y,1,\ell_\mathrm{binary})$ is supposed to hold for any $x$ and $y$
with consecutive x-coordinates and coinciding lowest $\ell$ bits of the y-coordinate.
Consequently, $H(x,y)$ follows for every $x$ and $y$ with consecutive x-coordinates and coinciding y-coordinates.

In the same way, we make sure that any two grid elements with the subsequent y-coordinates and equal x-coordinates are connected via a $V$ predicate ($b \in \{0,1\}$ and $0\leq i \leq k-1$). 
{\small
	\begin{eqnarray}
		\ur(x,y)\wedge\mathit{Complete}(x) \wedge \mathit{Complete}(y) 
		&\Rightarrow& \mathit{ChkV}(x,y,1,0^k)
		\\
		\mathit{ChkV}(x,y,1,\vec{z},0,1^i) 
		\wedge \mathit{Sel}_1(x,1,\vec{z},0,1^i)
		\wedge \mathit{Sel}_0(y,1,\vec{z},0,1^i)
		&\Rightarrow&
		\mathit{ChkV}(x,y,1,\vec{z},1,0^i)
		\\
		\mathit{ChkV}(x,y,1,\vec{z},0,1^i) 
		\wedge \mathit{Sel}_0(x,1,\vec{z},0,1^i)
		\wedge \mathit{Sel}_1(y,1,\vec{z},0,1^i)
		&\Rightarrow&
		\mathit{ChkV}'(x,y,1,\vec{z},1,0^i)
		\\
		\mathit{ChkV}'(x,y,1,\vec{z},0,1^i) 
		\wedge \mathit{Sel}_b(x,1,\vec{z},0,1^i)
		\wedge \mathit{Sel}_b(y,1,\vec{z},0,1^i)
		&\Rightarrow&
		\mathit{ChkV}'(x,y,1,\vec{z},1,0^i)
		\\
		\mathit{ChkV}'(x,y,1,1^k) 
		\wedge \mathit{Sel}_b(x,1,1^k)
		\wedge \mathit{Sel}_b(y,1,1^k)
		&\Rightarrow&
		\mathit{ChkV}''(x,y,0,0^k)
		\\
		\mathit{ChkV}''(x,y,0,\vec{z},0,1^i) 
		\wedge \mathit{Sel}_b(x,0,\vec{z},0,1^i)
		\wedge \mathit{Sel}_b(y,0,\vec{z},0,1^i)
		&\Rightarrow&
		\mathit{ChkV}''(x,y,0,\vec{z},1,0^i)
		\\
		\mathit{ChkV}''(x,y,0,1^k) 
		\wedge \mathit{Sel}_b(x,0,1^k)
		\wedge \mathit{Sel}_b(y,0,1^k)
		&\Rightarrow& V(x,y)
	\end{eqnarray}
}%
This way, we have established a doubly exponential grid where $V$ indicates vertical connections and $H$ horizontal connections between grid elements. Encoding a tiling on top of such a grid is standard.

\section{Data Complexity}\label{sec:datacomplexity}

We now consider the so-called \emph{data complexity} of \TGF~and
\GFU. Assume any formula $\varphi$ in \TGF~or \GFU. We consider the
problem of checking whether the formula $\varphi\land \gamma$ is
satisfiable, where $\gamma$ is a conjunction of ground literals (that is, negated or unnegated atoms ) given
as input. Technically speaking, we consider a family of decision
problems (one for every formula $\varphi$), where one instance of the
problem is a conjunction $\gamma$ of ground literals. It is not difficult to see
that there exists some $\varphi$ such that the above problem is
\np-hard. For instance, consider the formula $\varphi_{\mathrm{3COL}}$
defined as follows:
\begin{align*}
  \varphi_{\mathrm{3COL}} = & (\forall x. G(x)\lor R(x) \lor B(x)) \land  \\
                            & (\forall x,y. G(x)\land E(x,y)\Rightarrow \neg G(y)) \\
                            & (\forall x,y. R(x)\land E(x,y)\Rightarrow \neg R(y)) \\
                            & (\forall x,y. B(x)\land E(x,y)\Rightarrow \neg B(y)) 
\end{align*}
It is immediate that checking the satisfiability of
$\varphi_{\mathrm{3COL}}\land \gamma$, where $\gamma$ is an input
conjunction of ground atoms, is an \np-hard problem. Indeed, this
problem allows us to decide 3-colorability of a given graph $G$ that
can be easily converted in polynomial time into a conjunction $\gamma$
using the binary predicate $E$ to represent the edges of $G$. In
common parlance, this shows the \np-hardness of data complexity of
satisfiability \TGF~and \GFU. Our next goal is to show that all of the
above problems belong to \np, and thus to establish \np-completeness
of data complexity. We note that the \np-hardness proof above is mainly
for illustration purposes; the result also directly follows from the data
complexity of satisfiability in expressive Description
Logics~\cite{Scha94}.

The key to the \np membership result is the following notion of
\emph{compatibility} between a structure and a mosaic.

\begin{definition}\label{def:compatibility}
  Assume a signature $\sigma$, a structure $\str{A}$ over $\sigma$,
  and a mosaic $\mathcal{M}$ for a formula $\varphi$ over $\sigma$. We
  say $\str{A}$ and $\mathcal{M}$ are \emph{compatible}, if
  $\{\type{\str{A}}{e_1,\ldots,e_k}\mid  \{e_1,\ldots,e_k\} \subseteq \check{A},\  |\{e_1,\ldots,e_k\}| = k ,\  0 \leq k\leq
  \width(\varphi) \}\subseteq \mathcal{M}$.
\end{definition}

\begin{theorem}\label{thm:compatibility}
Let $\varphi$ be a GFU formula and let $\str{A}$ be a structure, both over the signature $\sigma$.
Then $\varphi$ has a model $\str{B}$ with $\str{B}\restr A=\str{A}$ adhering to the SNA 
iff there is a mosaic $\mathcal{M}$ for $\varphi$ that is 
  compatible with $\str{A}$.
\end{theorem}

\begin{proof}
  While the ``only if'' part is immediate from the proof of
  Theorem~\ref{thm:completeness}, the ``if'' part follows from a small
  adaptation of the model construction in the proof of
  Theorem~\ref{thm:soundness}. Specifically, instead of starting with
  an instantiation of an arbitrary initial type $\alpha_0$ and
  expanding it to a model of the the formula, we start with $\str{A}$
  and using a sequence of type instantiations we expand $\str{A}$ to a
  model of $\varphi$.
\end{proof}

\begin{theorem}
  Assume a formula $\varphi$ in \TGF~or \GFU.  There is a
  non-deterministic polynomial time procedure to decide, given as
  input a conjunction $\gamma$ of ground literals, whether
  $\varphi\land \gamma$ has a model adhering to the SNA.
\end{theorem}
\begin{proof}
  We can w.l.o.g.\,assume that $\varphi$ is a \GFU~formula in normal
  form. Let $C$ denote the set of constants from $\varphi$'s signature $\sigma$. 
  Here we
  see $\gamma$ as a set of ground literals over the signature $\sigma \cup C'$, where $C'$ denotes the set of constants used in $\gamma$. 
  We point out that $C'$ is not necessarily a subset of $C$. 
  Take the set
  $\mathbb{M}$ of all mosaics for $\varphi$. Note that each mosaic is
  of at most double exponential size in the size of $\varphi$, and
  thus $\mathbb{M}$ is at most triply exponential in the size of
  $\varphi$. Since $\varphi$ is fixed (not part of the input), it only
  matters to us that $\mathbb{M}$ can be effectively constructed:
  since the size of $\varphi$ is bounded by a constant, the size of
  $\mathbb{M}$ is also bounded by a constant.

  The procedure to check the satisfiability of $\varphi\land \gamma$
  is as follows:
  \begin{enumerate}
  \item Non-deterministically complete $\gamma$ so that (exclusively)
    either $p(\vec{c})$ or $\neg p(\vec{c})$ belongs to $\gamma$, for
    all predicates $p$ of $\gamma$ and all tuples $\vec{c} \in (C \cup C')^\ell$
    where $\ell$ is the arity of $p$. Moreover, $\UU(c,c')\in \gamma$ is required  for all $c,c'\in C \cup C'$. 
  \item Convert $\gamma$ into a structure $\str{A}$ with
    $A = C \cup C'$,  $\hat{A}=C$, $\check{A} = C' \setminus C$,
    and $c^\str{A}=c$ for all $c \in \hat{A}$. For all predicates $p$
    of $\gamma$, we set
    $p^{\str{A}}=\{\vec{c}\mid p(\vec{c})\in \gamma\}$.
  \item Check if there exists an $\mathcal{M}\in \mathbb{M}$ such that
    $\mathcal{M}$ is compatible with $\str{A}$. If such an
    $\mathcal{M}$ exists, then return ``yes''.
  \end{enumerate}
  Based on Theorem~\ref{thm:compatibility}, it is immediate that
  $\varphi\land \gamma$ is satisfiable iff there exists a run of the
  above procedure that returns ``yes''. The above procedure runs in
  nondeterministic polynomial time, as, after the nondeterministic
  guess in the first step, the other two steps above require only
  polynomial time in the size of $\gamma$.
\end{proof}

It is easy to see that the above result can be combined with Proposition~\ref{prop:SNA} in order to obtain a nondeterministic polynomial time procedure for checking satisfiability in the case where the SNA is not imposed.

\begin{corollary}
  Satisfiability in \TGF~and \GFU~is \np-complete in data complexity, with and without the SNA.
\end{corollary}

\section{Expressing TGF Queries in Disjunctive Datalog}\label{sec:disjdatalog}

In recent years, the approach of \emph{ontology-mediated querying} (OMQ) has gained momentum in database theory.
It is based on the idea that a \emph{database} $D$, usually represented as a finite set of ground facts, is supposed to be queried for information, taking into account background information expressed by a logical theory or formula $\varphi$. The most basic query form are \emph{atomic queries}, where one predicate, say $q$, from $\varphi$'s signature is picked. Assume $q$ is of arity $k$. 
The \emph{answer} to the combined query
$(\varphi,q)$ over a database $D$ is then the set of all $k$-ary tuples
$\vec{c}$ of constants that appear in $D$ such that
$\varphi\land D \models q(\vec{c})$. 
We let $ans(\varphi,q,D)$ denote the answer to $(\varphi,q)$ over $D$. 

Naturally, we are interested in the case of \emph{atomic TGF (resp.,\,GFU) queries}, where $\varphi$ is a formula in TGF (resp.\,GFU). As an aside, we point out that this
query language is quite expressive: e.g., it includes atomic OMQs
based on the description logic $\mathcal{ALCHI}$ studied by Bienvenu et
al.~\cite{BCL14}.

Here, we want to connect our query language to \emph{disjunctive Datalog}, short DDlog,
which is a well-known query language in the area of (deductive)
databases. We choose to present a definition where the formalism is introduced as another type of ontology-mediated queries. A \emph{(DDlog) program} $P$ is a finite set of sentences  of the form:
\begin{equation}
  \label{eq:datalog-rule}
  \forall \vec{x}. B_1(\vec{u}_1)\land \ldots \land B_n(\vec{u}_n) \Rightarrow H_1(\vec{v}_1)\lor \ldots \lor H_m(\vec{v}_m)
\end{equation}
such that each variable that appears in some
$\vec{v}_1, \ldots , \vec{v}_m$ also appears in some
$\vec{u}_1, \ldots,\vec{u}_n$ (this is the so-called \emph{safety
  restriction}).
Since we only have universally quantified variables in Datalog, we will drop ``$\forall \vec{x}.$'' from rules of the form \eqref{eq:datalog-rule}. A rule of the form $\top \Rightarrow H_1(\vec{v}_1)\lor \ldots \lor H_m(\vec{v}_m) $ is a (disjunctive) \emph{fact}, while a rule of form $ B_1(\vec{u}_1)\land \ldots \land B_n(\vec{u}_n) \Rightarrow\bot$ is called a \emph{constraint}. 

Furthermore, we also consider a version of this language that allows
for the built-in inequality predicate $\neq$, denoted DDlog$^\neq$.

Following the schema of ontology-mediated querying, a DDlog$^{(\neq)}$ \emph{query} is a pair $(P,q)$, where $P$ is a
DDlog$^{(\neq)}$ program and $q$ is a predicate symbol of some arity $k$. The
\emph{answer} to $(P,q)$ over a database $D$ is the set of all
$k$-ary tuples $\vec{c}$ of constants such that
$P\land \bigwedge D \models q(\vec{c})$ and we let $ans(P,q,D)$ denote the answer
to $(P,q)$ over $D$. We note that due to the safety restriction, 
the
answer is always limited to the constants that appear in the input
database or the program.

We next show that any atomic TGF query can be converted into an equivalent DDlog$^\neq$ query, that is a query that produces the same answers for any database $D$.
\begin{theorem}
	Given a GFU formula $\varphi$, one can construct a DDlog$^\neq$ program $P^\mathrm{SNA}_{\varphi}$ (over an extended signature) 
	such that, for any set of ground atoms $D$ and a ground atom
	$q(\vec{c})$ over the predicate signature of $\varphi$, we have
	$P^\mathrm{SNA}_{\varphi}\land \bigwedge D\models q(\vec{c})$ iff
	$\varphi\land\bigwedge D \models q(\vec{c})$ under the SNA.
\end{theorem}
\begin{proof}
	We can w.l.o.g.\,assume that $\varphi$ is in normal form and we let $\mathbb{M}$ be the
	set of all mosaics for $\varphi$. The idea is to exploit the notion of compatibility between a structure and a mosaic (see Definition~\ref{def:compatibility} and Theorem~\ref{thm:compatibility}). The construction of
	$P^\mathrm{SNA}_{\varphi}$ is as follows.
	
	For all $\mathcal{M}\in \mathbb{M}$, let $S_{\mathcal{M}}$ be a fresh
	nullary relation symbol. We add to $P^\mathrm{SNA}_{\varphi}$ the following
	disjunctive fact:
	\[\top \Rightarrow \bigvee_{\mathcal{M}\in \mathbb{M}}  S_{\mathcal{M}} \]
	The above rule allows to choose an arbitrary mosaic in
	$\mathcal{M}\in \mathbb{M}$ and this choice is indicated by the truth
	of the (propositional) atom $S_{\mathcal{M}}$.  For every $\ell$-type
	$\alpha$ that appears in $\mathcal{M}$, take a fresh predicate $T_{\alpha}$ of
	arity $\ell$. Add the following rule for every $\mathcal{M}\in \mathbb{M}$ and
	every $\ell\leq \width(\varphi)$:
	\[ S_{\mathcal{M}}\land 
	\bigwedge_{1 \leq i < j \leq \ell}  y_i\neq  y_j  \Rightarrow \bigvee_{\alpha\in \mathcal{M}\mbox{ is an }\ell\mbox{-type}}
          T_{\alpha}(y_1,\ldots,y_\ell)           \]
	For every $k$-ary relation symbol $p$ of $\varphi$, we take a fresh
	predicate $\bar{p}$ of the same arity as $p$ and add the following
	constraint (where $\vec{y}$ is a $k$-ary
	tuple of variables):
	\[ p(\vec{y})\land  \bar{p}(\vec{y}) \Rightarrow \bot\]  
	What remains is to ensure that a type $\alpha$ assigned to some tuple $\vec{c}$ of constants is compatible with the atoms over $\vec{c}$ stored in $\gamma$.
	For every $\ell$-type $\alpha$ that appears in $M$, add the following
	rules:
	\[  T_{\alpha}(x_1,\ldots,x_\ell) \Rightarrow p(\vec{t})  \quad \mbox{for all }
	p(\vec{t})\in \alpha \]
	\[  T_{\alpha}(x_1,\ldots,x_\ell) \Rightarrow \bar{p}(\vec{t})  \quad \mbox{for all }
	\neg p(\vec{t})\in \alpha \]
	This concludes the construction of the DDlog$^\neq$ program.
\end{proof}

We now turn to the case where the SNA is not imposed, which means that some of the constants in $D$ may denote the same individual. Interesting, in this case we can even provide a DDlog program without built-in disequality.

\begin{theorem}
	Given a GFU formula $\varphi$, one can construct a DDlog program $P_{\varphi}$ (over an extended signature) 
	such that, for any set of ground atoms $D$ and a ground atom
	$q(\vec{c})$ over the predicate signature of $\varphi$, we have
	$P_{\varphi}\land \bigwedge D\models q(\vec{c})$ iff
	$\varphi\land\bigwedge D \models q(\vec{c})$ (with the SNA not imposed).
\end{theorem}

\begin{proof}
Let $C$ be the set of constant symbols occurring in $\varphi$ and let $\mathbb{E}(C)$ denote the set of equivalence relations over $C$. For every $\sim$ from $\mathbb{E}(C)$, let $\mathbb{M}_\sim$ denote the set of all mosaics for $\varphi^\sim$ (see the proof of Proposition~\ref{prop:SNA}). Now, for every
$\sim$ introduce a fresh nullary symbol $Q_\sim$. Then, we let $P_{\varphi}$ contain the disjunctive rule 
	\[\top \Rightarrow \bigvee_{{\sim}\in \mathbb{E}(C)}  Q_\sim \]
as well as the rule
\[  Q_\sim \Rightarrow EQ(c,c')  \qquad\mbox{ for all }{\sim}\in \mathbb{E}(C) \mbox{ and } c,c' \in C \mbox{ with } c\sim c' \]
and 
\[  Q_\sim \Rightarrow NEQ(c,c')  \qquad \mbox{ for all }{\sim}\in \mathbb{E}(C) \mbox{ and } c,c' \in C \mbox{ with } c\not\sim c'. \]
Therein, $EQ$ and $NEQ$ are two fresh binary predicates, meant to express equality and disequality, respectively, which is axiomatized as follows:
\[ 
  ADom(x)\land ADom(y) \Rightarrow EQ(x,y) \vee NEQ(x,y)  \qquad \qquad  EQ(x,y)\land  NEQ(x,y)  \Rightarrow \bot  
\]
\[ 
 ADom(x) \Rightarrow EQ(x,x) \qquad\qquad    EQ(y,x) \Rightarrow  EQ(x,y)    \qquad\qquad  EQ(x,y)\land EQ(y,z)\Rightarrow EQ(x,z)  
\]
\[ 
 p(x_1,\ldots x_k)\land  EQ(x_1,x'_1)\land  \ldots \land EQ(x_k,x'_k)  \Rightarrow p(x'_1,\ldots x'_k) \qquad \mbox{ for all  predicates $p$ from }\varphi
\]
	Here $\mathit{ADom}$ is the usual ``active domain'' predicate, which
	can be easily expressed in Datalog. Specifically, for every $k$-ary relation symbol $p$ that appears in $\varphi$, and every $1 \leq i \leq k$, we add the following rule:
	\[p(y_1,\ldots,y_k)  \Rightarrow \mathit{ADom}(y_i) \]

The further rules of $P_{\varphi}$ are similar to the SNA case:
\[ Q_\sim  \Rightarrow \bigvee_{\mathcal{M}\in \mathbb{M}_\sim}  S_{\mathcal{M}} \qquad \mbox{ for every }{\sim} \in \mathbb{E}(C)\]
\[S_{\mathcal{M}}\land 
\bigwedge_{1 \leq i < j \leq \ell}  NEQ(y_i, y_j)  \Rightarrow \bigvee_{\alpha\in \mathcal{M}\mbox{ is an }\ell\mbox{-type}}
T_{\alpha}(y_1,\ldots,y_\ell)  \qquad \mbox{for every $\mathcal{M}$ and
	every $\ell\leq \width(\varphi)$}\]
Note that the built-in disequality predicate $\neq$ has been replaced by the axiomatized one $NEQ$.  
\[p(\vec{y})\land \bar{p}(\vec{y}) \Rightarrow \bot\]  
For every $\ell$-type $\alpha$, add all the following
rules (recall the special form of constants introduced by turning $\varphi$ into $\varphi^\sim$):
\[  T_{\alpha}(x_1,\ldots,x_\ell) \Rightarrow p(\vec{t'})  \quad \mbox{for all }
p(\vec{t})\in \alpha \mbox{ where $\vec{t'}$ is obtained from $\vec{t}$ by replacing every $b_D$ by some $c \in D$}, \]
\[  T_{\alpha}(x_1,\ldots,x_\ell) \Rightarrow \bar{p}(\vec{t'})  \quad \mbox{for all }
\neg p(\vec{t})\in \alpha \mbox{ where $\vec{t'}$ is obtained from $\vec{t}$ by replacing every $b_D$ by some $c \in D$}. \]
 \end{proof}

\section{Undecidability Results} \label{s:undec}
We  review here some further natural extensions of \tgf{} and find that they lead to undecidability.

\paragraph{Relaxing guardedness further} Unguarded quantification of subformulae with three variables would allow to express any formula of the three-variable fragment of $\fo$, denoted ${\mathrm{FO}^3}$, for which satisfiability is  undecidable (as ${\mathrm{FO}^3}$ contains the class of $\fo$ sentences with quantifier prefix $\forall\exists\forall$ which is undecidable \cite{Lewis:Unsolvable}).

\paragraph{Counting} $\fot$ can be extended by counting quantifiers of the shape $\exists^{=n}$, $\exists^{\leq n}$, and $\exists^{\geq n}$, yielding a logic denoted $\mathrm{C}^2$. This extension is useful for logical modeling (it helps to capture description logics with cardinality restrictions and graded modal logics) and does not lead to an increase in complexity of satisfiability checking \cite{PH:C2complex}. Yet, this enrichment is detrimental when mixing it with the guarded fragment: via the $\mathrm{C}^2$ sentence $\forall x. \exists^{=1} y. F(x,y)$ we can enforce that $F$ must be interpreted as a functional binary relation. Yet, adding a functional relation to $\gf$ is known to cause undecidability \cite{Gra99}.

\paragraph{Conjunctive Queries} Instead of asking for satisfiability of a $\TGF${} theory, an often considered problem stemming from database theory is also if it entails a Boolean conjunctive query (i.e., an existentially quantified conjunction of atoms). However, conjunctive query entailment has been shown to be undecidable already for $\fot$ alone \cite{DBLP:conf/icdt/Rosati07}. This also shows that any attempt of extending $\TGF${} such that it incorporates $\fo$ fragments that can express negated Boolean conjunctive queries (such as the unary negation fragment \cite{DBLP:journals/corr/SegoufinC13} or the guarded negation fragment \cite{DBLP:journals/jacm/BaranyCS15}) will lead to undecidability.

\paragraph{Loose guardedness} It has been shown that $\gf$ remains decidable if the guardedness restriction is relaxed, leading to notions such as the loosely guarded fragment, the packed fragment or the clique-guarded fragment. For most restrictive notion of those, the loosely guarded fragment \cite{loosely}, the guard does not need to be one atom containing all free variables, rather it can be a conjunction of atoms with the property that any pair of free variables occurs together in one of those conjuncts. It is not hard to see that in the presence of the $\ur$ predicate (or if such a predicate can be axiomatized as in \TGF), we can create a ``loose guard'' $\bigwedge_{\{x,y\} \subseteq \vec{x}} \ur(x,y)$ for any set $\vec{x}$ of 
free variables. This allows to quantify over the full domain, hence every $\fo$ formula is equivalent to such a loosely guarded one. Consequently, a hypothetical ``loosely triguarded fragment'' would be as expressive as $\fo$, hence undecidable.

\renewcommand{\approx}{=} 
\paragraph{Equality atoms of the form $x \approx y$}

\newcommand{\predfont}{\mathit}

\newcommand{\predH}{\predfont{H}}
\newcommand{\predV}{\predfont{V}}
\newcommand{\predOrig}{\predfont{Orig}}
\newcommand{\predChkFunc}{\predfont{ChkFunc}}
\newcommand{\predChkSq}{\predfont{ChkSq}}

The 
undecidability of satisfiability of \GFU{} (and hence of \TGF{}) in the presence of statements that assert equality between variables can be inferred from~\cite{Kazakov:06:Phd} (Section 4.2.3), where a
reduction from satisfiability in the \emph{Goldfarb class} is
presented, and it can be applied to our fragments. Here
we provide a more direct undecidability proof by a reduction from
the tiling problem for an infinite grid \cite{BGG}. We can
construct a \GFU{} formula with equality such that its universal model
represents an $\mathbb{N} \times \mathbb{N}$ grid. Thereby, the domain
elements of the model correspond to grid positions and every position
is connected to its upper neighbor by a binary predicate $\predV$ and to
its right neighbor by a binary predicate~$\predH$.

In the following, we omit leading universal quantifiers; all formulae are sentences.
We start our modeling by ensuring there is exactly one leftmost, bottommost position of the grid, i.e., the
``origin''.

\begin{eqnarray}
	& \exists x. \predOrig(x) &
	\\
	& \ur(x,y) \wedge \predOrig(x) \wedge \predOrig(y) \to x \approx y &
\end{eqnarray}
Any two domain elements co-occur together with the origin in a ternary
auxiliary predicate $\predChkFunc$.
\begin{eqnarray}
	\ur(x,y) \to \exists z. \predChkFunc(x,y,z) \wedge \predOrig(z)
\end{eqnarray}
Intuitively, $\predChkFunc(x,y,z)$ indicates that we will enforce that if
$z$ is connected with both $x$ and $y$ by predicate $\predV$ (or $\predH$), then
$x$ and $y$ must coincide; in other words, as $x$ and $y$ are
arbitrary elements, $z$ has only one outgoing $\predV$-connection and
one outgoing $\predH$-connection. The following two sentences implement
this.
\begin{eqnarray}
	\predChkFunc(x,y,z) \wedge \predH(z,x) \wedge \predH(z,y) & \to & x \approx y\\
	\predChkFunc(x,y,z) \wedge \predV(z,x) \wedge \predV(z,y) & \to & x \approx y
\end{eqnarray}
In particular, this makes sure that the origin has exactly one right
and one upper neighbor. Also, we propagate this ``local functionality''
enforcing predicate along the (known to be unique) $\predV$- and
$\predH$-connections.
\begin{eqnarray}
	\predChkFunc(x,y,z) & \to & \exists w. \predChkFunc(x,y,w) \wedge \predH(z,w)
	\\
	\predChkFunc(x,y,z) & \to & \exists w. \predChkFunc(x,y,w) \wedge \predV(z,w)
\end{eqnarray}
With these axioms alone, the corresponding universal model would resemble an 
infinite binary tree, with the origin as root and every node
having (exactly) one $\predH$-successor and (exactly) one
$\predV$-successor. The next axioms make sure that for every element $e$ in
our structure, the element reached from $e$ via an $\predH$-$\predV$-path
coincides with the element reached from $e$ via a $\predV$-$\predH$-path, 
using another auxiliary $5$-ary predicate $\predChkSq$ which is
handled in a way that $\predChkSq(x,y,z_1,z_2,z_3)$ is only entailed
whenever $z_1$ has $z_2$ as right neighbor and $z_3$ as upper
neighbor.

Again, we start ensuring this for $e$ being the origin and then work
our way through the structure along the (unique) $\predH$- and $\predV$-
connections.
\begin{eqnarray}
	\ur(x,y) \to \exists z_1z_2z_3. \predChkSq(x,y,z_1,z_2,z_3) \wedge
	\predOrig(z_1) \wedge \predH(z_1,z_2) \wedge \predV(z_1,z_3) & &
	\\
	\predChkSq(x,y,z_1,z_2,z_3) \to \exists w_1w_2. \predChkSq(x,y,z_2,w_1,w_2) \wedge  \predH(z_2,w_1) \wedge \predV(z_2,w_2) & &
	\\
	\predChkSq(x,y,z_1,z_2,z_3) \to \exists w_1w_2. \predChkSq(x,y,z_3,w_1,w_2) \wedge  \predH(z_3,w_1) \wedge \predV(z_3,w_2) & &
\end{eqnarray}
Finally, we ensure that if $\predChkSq(x,y,z_1,z_2,z_3)$ holds and $x$
is the right neighbor of $z_2$ and $y$ is the
upper neighbor of $z_3$, that then $x$ and $y$ must coincide.
\begin{eqnarray}
	\predChkSq(x,y,z_1,z_2,z_3) \wedge \predV(z_2,x) \wedge \predH(z_3,y) \to x \approx y
\end{eqnarray}
This finishes our modeling of the infinite grid. It is now
straightforward to model a tiling on top of this, and we obtain the following theorem.

\begin{theorem}[Undecidability with Equality]
	Checking satisfiability of \TGF{} formulae with equality expressions of the form $x \approx y$ is
	undecidable.
\end{theorem}

\section{Warm-up for FMP: Alternative Decidability Proof for TGF} \label{s:decgfu}

\label{s:altdec}

The main purpose of this section is to prepare the reader to understand our finite model construction in the next section.
We believe, however, that it is also of independent 
interest as it can serve as an alternative, quick decidability proof for \GFU{}, and thus also for \TGF{}. The basic idea is
to augment the input \GFU{} formula with some additional conjuncts (parametrized by a set of $1$-types), enforcing that any model contains a sufficiently rich collection of $\UU$-connections, and then to show that the resulting formula is satisfiable in \GF{}
(over the chosen set of $1$-types) iff the original one is satisfiable in \GFU{}.
We remark that also the decidability proof for the logic
$\gfcross$ in \cite{DBLP:conf/ijcai/BourhisMP17} goes via a reduction to $\GF$. Our approach is different, and, as we believe, conceptually 
simpler; in contrast to the latter it also covers the case of the unrestricted use of constants and equalities of the form $x=c$ (variable = constant).

Let $\varphi$ be a \GFU{} sentence in normal form and $\sigma$ the signature consisting of the symbols used in $\varphi$. Assume that $c_1, \ldots, c_m$
is the list of all constants in $\sigma$. Let $\AAA$  be a set of $1$-types over $\sigma$. 
We construct a  \GF{} $\sigma$-sentence $\varphi^\AAA$ by appending to  $\varphi$ the following conjuncts:
\begin{align}
	\label{expphi1} \forall x \big(x=x \Rightarrow (\bigvee_{i=1}^{m} x=c_i \vee \bigvee_{\alpha \in \AAA} \alpha(x))\big)\\
	\label{expphi2} \bigwedge_{\alpha, \alpha' \in \AAA} \exists xy \big( \UU(x,y) \wedge \UU(y,x) \wedge \alpha(x) \wedge \alpha'(y) \wedge \bigwedge_{i=1}^{m} (x \not= c_i \wedge y \not=c_i) \big)\\
	\label{expphi3}  \bigwedge_{P \in \sigma} \forall \bar{x} \Big( (P(\bar{x}) \ \Rightarrow \bigwedge_{i,j \in \{1,\ldots,|\bar{x}|\}}  \UU(x_i, x_j)\Big) 
\end{align}
saying, respectively, that only $1$-types from $\AAA$ are realized, all $1$-types from $\AAA$ are realized and every pair of $1$-types has a realization both-ways connected by $\UU$, 
and every guarded pair of elements is connected by $\UU$. Note that $\varphi^\AAA$ is in weak normal form.

The following lemma reduces the satisfiability problem for \GFU{} to the satisfiability problem for \GF{}.
\begin{lemma} \label{l:reductiontogf}
Let $\varphi$ be a \GFU{} sentence in weak normal form,
and $\sigma$ the signature consisting of the symbols used in $\varphi$. 
Then $\varphi$ is satisfiable (over a $\UU$-biquitous model) iff there is a  consistent
set of $\UU$-biquitous $1$-types $\AAA$ over $\sigma$ such that the $\GF{}$-sentence $\varphi^\AAA$ is satisfiable (over a not necessarily $\UU$-biquitous model). 
\end{lemma}
The proof for the left-to-right direction is easy:
Let $\str{A} \models \varphi$ be $\UU$-biquitous and let $\AAA$ be the set of $1$-types realized in $\str{A}$ by its all unnamed elements. It is then readily verified that
$\AAA$ is consistent, contains only $\UU$-biquitious $1$-types, and that $\str{A}$ satisfies the
conjuncts (\ref{expphi1})--(\ref{expphi3}), so $\str{A} \models \varphi^\AAA$. 
Let us turn to the proof of the right-to-left direction. For the sake of clarity, we divide it into several paragraphs. However,
we do not try to make it as simple as possible, but rather attempt to make its structure similar to the structure of the finite model
construction that will be presented in the next section.

\medskip \noindent
{\bf Initial infinite non-$\UU$-biquitous model}.  
Let $\AAA$ be a consistent set of $\UU$-biquitous $1$-types over $\sigma$, let $\str{C}_0$ be a model of $\varphi^\AAA$.
 As \GF{} has the finite model property we may assume that $\str{C}_0$ is finite. If $\str{C}_0$ is $\UU$-biquitous then we are done. Otherwise, 
let $\str{C}$ be the doubling of $\str{C}_0$. $\str{C}$ is still finite and by Lemma \ref{l:doubling} we have that  $\str{C} \models \varphi^\AAA$.

Let $\str{C}_1, \str{C}_2, \ldots$ be an infinite sequence of isomorphic copies of $\str{C}$ sharing the
named domain $\hat{C}$ with $\str{C}$ and having pairwise disjoint unnamed domains $\check{C}_1, \check{C}_2, \ldots$.
Obviously this sequence is joinable and each $\str{C}_i$ is a model of $\varphi^{\AAA}$. Let $\str{A}$ be its join.
 $\str{A} \models \varphi^\AAA$ by Lemma \ref{l:join}.

$\str{A}$ is not $\UU$-biquitous, but it is easy to see that the only pairs not connected via $\UU$ must consist of unnamed elements.
This follows from the fact that all $1$-types realized in $\str{A}$ are transferred from some $\str{C}_i$ and
hence they are $\UU$-biquitous and contain $\UU(c,c')$, $\UU(c,x)$ and $\UU(x,c)$ for all constants $c,c'$.
\begin{claim} \label{c:warmup0}
Any pair from $\hat{A} \times A$ is $\UU$-biqutious in $\str{A}$.
\end{claim}

\medskip \noindent
{\bf Outline}. In order to turn $\str{A}$ into a $\UU$-biquitous structures, we will modify the types of some tuples of its domain elements according to the following strategy: for a non-$\UU$-biquitous pair of elements $(a_1,a_2)$ we will ``attach'' it to one of the $\str{C}_i$-substructures by making $\str{A}$'s induced substructure on $\{a_1, a_2 \} \cup C_i$ isomorphic to the partial doubling of $\str{C}_i$ wherein a twin $e_1 \in C_i$ of $a_1$ and a twin $e_2 \in C_i$ of $a_2$ which together form a $\UU$-biqutious pair in $\str{C}_i$ are duplicated. 
This makes the pair $(a_1,a_2)$ (and, possibly, some other pairs $(a_i,b)$, $i=1,2$, $b \in \check{C}_i$) $\UU$-biqutious.
As $\str{A}$ has infinitely many substructures $\str{C}_i$, there is a plenty of space to do this without conflicts. Let us next turn to the details.

\medskip \noindent
{\bf Joining function}.  
Let us fix an injective function $\hD:\check{A} \times \check{A} \rightarrow \N$ returning for every pair $(a_1,a_2) \in \check{A}\times\check{A}$
a number bigger than $\max \{m_1,m_2 \}$ where $m_i \in \N$ is such that $a_i \in C_{m_i}$, for $i=1,2$.
This function indicates to which $\str{C}_i$ the given pair will be attached, if necessary. 
The following property of $\hD$ follows directly from its definition.

\begin{claim}\label{c:warmup1}
For every $a_1, a_2 \in \check{A}$ and every $b_1 \in \check{C}_{\hD(a_1,a_2)}$
there is no $b_2 \in \check{A}$ such that $a_1$ or $a_2$ belongs to $C_{\hD(b_1, b_2)}$.
\end{claim}

\medskip \noindent{\bf Definition of  $D$-structures.} 
For every pair of distinct elements $a_1, a_2 \in \check{A}$ we distinguish three subdomains of $\str{A}$ arising by 
adding one or both of $a_1, a_2$ to $C_{\hD(a_1,a_2)}$.  Formally:
\begin{itemize}
	\item $D^{\sss a_1a_2}_{\hD(a_1,a_2)}:=\{a_1, a_2\} \cup {C}_{\hD(a_1, a_2)}$, 
	\item $D^{\sss a_1\cdot}_{\hD(a_1,a_2)}:=\{a_1\} \cup {C}_{\hD(a_1, a_2)}$, 
	\item $D^{\sss \cdot a_2}_{\hD(a_1,a_2)}:=\{a_2\} \cup {C}_{\hD(a_1, a_2)}$.
\end{itemize}

We now define structures on these distinguished subdomains. In the first of them, the pair $(a_1,a_2)$ will be $\UU$-biquitous.
By (\ref{expphi1}), $\type{\str{A}}{a_i}  \in \AAA$ for $i=1,2$.
Let $e_1, e_2 \in C_{\hD(a_1,a_2)}$ be distinct elements with $\type{\str{C}_{\hD(a_1,a_2)}}{e_i} =\type{\str{A}}{a_i}$ for $i=1,2$, such that $(e_1,e_2), (e_2,e_1) \in \UU^{\str{C}_{\hD(a_1,a_2)}}$. The existence of such $e_1, e_2$ follows by (\ref{expphi2}).  By Lemma~\ref{c:uconnected}, we may indeed assume that $e_1, e_2$ are distinct (even if they have the same $1$-type). 
Let:
\begin{itemize}
	\item $\str{D}^{\sss a_1a_2}_{\hD(a_1,a_2)}$ be isomorphic to the partial doubling of $\str{C}_{\hD(a_1,a_2)}$ duplicating the set $\{ e_1, e_2 \}$,
	\item $\str{D}^{\sss a_1\cdot}_{\hD(a_1,a_2)}$ be isomorphic to the partial doubling of $\str{C}_{\hD(a_1,a_2)}$ duplicating the singleton $\{ e_1 \}$, 
\item $\str{D}^{\sss \cdot a_2}_{\hD(a_1,a_2)}$ be isomorphic to the partial doubling of $\str{C}_{\hD(a_1,a_2)}$ duplicating the singleton $\{ e_2 \}$, 
\end{itemize}
via the natural isomorphism, taking $a_i$ to the twin of  $e_i$ for $i=1,2$.

The above structures will be  called ${D}$-\emph{structures}.
${D}$-structures from the first bullet are called \emph{double-headed}, while those from the two other are \emph{single-headed}
(\emph{left-} and \emph{right-headed}, respectively).  The elements $a_1,a_2$ are called their \emph{heads} and the sets $C_{\hD(d_1,d_2)}$ are their \emph{bodies}.
Observe that, by the definition of $\hD$, for any ${D}$-structure, its head(s) do not belong to its body.
The following claim follows from Lemma \ref{l:doubling}.

\begin{claim}
Every ${D}$-structure is a model of $\varphi^\AAA$.
\end{claim}

\medskip \noindent{\bf Basic properties of $D$-structures.} 
The next four claims follow directly from the construction.

\begin{claim} \label{c:warmup2}
Let $a$ be an element in any ${D}$-structure $\str{D}$. Then $\type{\str{D}}{a}=\type{\str{A}}{a}$.
\end{claim}

\begin{claim} \label{c:warmup3}
Every ${D}$-structure restricted to its body is isomorphic to
the restriction of $\str{A}$ to this body.
\end{claim}

\begin{claim} \label{c:warmup4}
$\str{D}^{\sss a_1\cdot}_{\hD(a_1,a_2)}$ and $\str{D}^{\sss \cdot a_2}_{\hD(a_1,a_2)}$
are induced substructures of $\str{D}^{\sss a_1a_2}_{\hD(a_1,a_2)}$, for all $a_1,a_2 \in \check{A}$.
\end{claim}

\begin{claim} \label{c:warmup5}
The pair of the heads of every double-headed ${D}$-structure $\str{D}$ is $\UU$-biquitous in $\str{D}$.
\end{claim}

Let us also make the following useful observation.
\begin{claim} \label{c:warmup6}
\begin{enumerate}[(i)]
\item
Every single-headed ${D}$-structure is joinable with $\str{A}$. 
\item 
If $a_1,a_2 \in \check{A}$ is a pair of distinct elements that is non-$\UU$-biquitous in $\str{A}$,
then also $\str{D}_{\hD(a_1,a_2)}^{a_1a_2}$ is joinable with $\str{A}$.
\end{enumerate}
\end{claim}

\begin{proof}
Let $\str{D}$ be any ${D}$-structure. Let $\bar{a}$ be a tuple of unnamed elements in $D$ (and hence also in $A$). 
We need to show that it has the same type in both $\str{D}$ and $\str{A}$ or is unguarded in at least one of them.
If $\bar{a}$ is built out of a single element then this follows by Claim \ref{c:warmup2}. So assume $\bar{a}$ contains more than one element. 
If it is contained in the body of $\str{D}$ then its type in $\str{D}$ is the same as in $\str{A}$ by Claim \ref{c:warmup3}.
Assume $\str{D}$ is single-headed, $a$ is its head, and $\bar{a}$ contains $a$ and some element $b$ from the body.
As $a$ does not belong to the body, by the construction of $\str{A}$ we have that the pair $(a,b)$ is non-$\UU$-biquitous in $\str{A}$, so, by the conjunct (\ref{expphi3}) of $\varphi^\AAA$
it follows that $\bar{a}$ is unguarded in $\str{A}$.
Assume now that $\str{D}$ is double headed, $a_1,a_2$ are its heads, and $(a_1,a_2)$ is non-$\UU$-biquitous in $\str{A}$. If $\bar{a}$ contains at most one of 
the heads then it is contained in one of the single-headed substructures of $\str{D}$, and hence we can proceed as in the previous case.
If $\bar{a}$ contains both $a_1,a_2$ then, then recalling that $a_1,a_2$ is is non-$\UU$-biquitous in $\str{A}$, and invoking again the conjunct (\ref{expphi3})
we conclude that $\bar{a}$ is unguarded in $\str{A}$.
\end{proof}

\medskip \noindent
{\bf The join.}
For every pair $(a_1, a_2)$  of distinct elements from $\check{A}$ that is $\UU$-biquitous neither in $\str{A}$ nor in
any of the single-headed ${D}$-structures select one
of the double-headed structures $\str{D}^{\sss a_1a_2}_{\hD(a_1,a_2)}$, $\str{D}^{\sss a_2a_1}_{\hD(a_2,a_1)}$.
Let $\mathcal{J}$ be the set consisting of so selected double headed ${D}$-structures,
all single-headed ${D}$-structures and $\str{A}$.

\begin{claim} \label{warmupjoin}
$\mathcal{J}$ is joinable.
\end{claim}

\begin{proof}
Every ${D}$-structure from $\mathcal{J}$ is joinable with $\str{A}$ by Claim \ref{c:warmup6}. Consider any pair $\str{D}_1$, $\str{D}_2$ 
of distinct ${D}$-structures from $\mathcal{J}$. Let $\bar{a}$ be a tuple of distinct unnamed elements contained in both of them. We need to
show that $\bar{a}$ has the same type in both $\str{D}_1, \str{D}_2$ or is unguarded in at least one of them.
We consider three cases:
\begin{enumerate}
	\item Both $\str{D}_1, \str{D}_2$ are single-headed. If they share the body, then, by the fact that $\hD$ is injective, invoking Claim \ref{c:warmup4}, they must be two single-headed substructures of some double-headed $D$-structure $\str{D}$. Then $\bar{a}$ is in the body of $\str{D}$ and $\type{\str{D}_1}{\bar{a}}=\type{\str{D}_2}{\bar{a}}=\type{\str{D}}{\bar{a}}$.
	If the bodies are different then either the heads are the same element $a$ and then $\bar{a}=a$, or the head $a$ of one structure is in the body of the other 
	and then again $\bar{a}=a$ (note that by Claim \ref{c:warmup1} it cannot happen that, simultanously, the head of one of the structures is in the body of the other and vice versa).
	In both cases $\bar{a}$ is just a single element, and its type in both $D$-structures is the same by Claim \ref{c:warmup2}.

	\item $\str{D}_1$ is single-headed, $\str{D}_2$ is double-headed (or vice versa). If $\bar{a}$ contains at most one head of $\str{D}_2$ then it is contained in one of the two single-headed
	substructure of $\str{D}_2$ and we can proceed as in the previous case. So, assume that $\bar{a}$ contains both heads of $\str{D}_2$. Then both these heads belong to
	the domain  the single-headed structure $\str{D}_1$, and by the choice of double-headed $D$-structures in $\mathcal{J}$, these heads must form a non-$\UU$-biquitous pair in $\str{D}_1$. 
	Using the conjunct (\ref{expphi3}) we infer that the whole $\bar{a}$ is unguarded in $\str{D}_1$.

		\item Both $\str{D}_1$ and $\str{D}_2$ are double-headed. If at least one of the heads of these $D$-structures is not in $\bar{a}$, say a head of $\str{D}_1$, then
	$\bar{a}$ is contained in one of the two single-headed substructure of $\str{D}_1$ and we may proceed as in the previous case. 
		So assume $\bar{a}$ contains both heads $a_1,a_2$ of $\str{D}_1$ and both heads $b_1, b_2$ of $\str{D}_2$. 
		In this case it must be that $\{a_1, a_2 \} = \{b_1, b_2 \}$. Assume to the contrary that this is not the
		case, that is one of the heads of one of the structures belongs to the body of the other structure,
		say $b_1 \in \check{C}_{\hD(a_1, a_2)}$. But then by Claim \ref{c:warmup1} none of $a_1, a_2$ can belong to the body of $\str{D}_2$.
		As they are in $\bar{a}$, they must belong to $D_2$, so they must be the heads of $D_2$. So one of $a_1, a_2$ must be equal to $b_1$, so must
		belong to the body of $\str{D}_1$. But a head of a $D$-structure cannot belong to its body. Contradiction.
		So, $\str{D}_1$, $\str{D}_2$ have the same heads, and since out of $\str{D}_{f(a_1, a_2)}^{a_1a_2}$
		$\str{D}_{f(a_2, a_1)}^{a_2a_1}$ we always choose to $\mathcal{J}$ at most one of them we conclude that actually $\str{D}_1=\str{D}_2$.\qedhere
		\end{enumerate}
\end{proof}

Let $\str{A}^*$ be the join of $\mathcal{J}$. By Lemma \ref{l:join} we have $\str{A}^* \models \varphi^\AAA$.
Let us see that $\str{A}^*$ is indeed $\UU$-biquitious.
By Claim \ref{c:warmup0} (and the fact that the guarded types from $\str{A}$ are retained in the join $\str{A}^*$) we need to consider only pairs of unnamed elements.
Consider a pair of distinct elements $a_1, a_2 \in \check{A}$. 
If they are $\UU$-biquitous in $\str{A}$ or in at least one of the single-headed $D$-structures then 
they are $\UU$-biquitous in $\str{A}^*$ since all the mentioned structures are members of $\mathcal{J}$.
Otherwise, one of the double-headed $D$-structures 
$\str{D}^{\sss a_1a_2}_{\hD(a_1,a_2)}$, $\str{D}^{\sss a_2a_1}_{\hD(a_2,a_1)}$
is a member of $\mathcal{J}$ and in this structure $(a_1,a_2)$ is $\UU$-biquitous by Claim \ref{c:warmup5}. So, 
$(a_1, a_2)$ becomes $\UU$-biquitous in $\str{A}^*$.
This finishes the proof of Lemma \ref{l:reductiontogf}.

\medskip\noindent
{\bf Decidability and complexity.}
Lemma \ref{l:reductiontogf} implies that the satisfiability problem for \GFU{}, and thus also of \TGF{}, is decidable: given a \GFU{} sentence $\varphi_0$ we convert it into weak normal form $\varphi$, 
and verify, for every subset of $1$-types $\AAA$ over the signature of $\varphi$, whether $\varphi^\AAA$, treated as a \GF{} sentence, has a model.

This lemma would be also quite convenient to establish the precise complexity of the problem in the absence of constants.
However, deriving the optimal complexity bounds for the case with constants, via the approach from this section, would be rather awkward, because
of the way in which constants are treated in the original papers in which the complexity of \GF{} is established \cite{Gra99}, \cite{DBLP:journals/jolli/CateF05}. This is why we presented
a tight upper complexity bound in Section \ref{s:compgfu}, working from scratch. It was more
convenient for extracting the optimal upper complexity bound for the case with constants and it also uniformly allowed us to establish the complexity 
in the case without constants and some further complexity results.

\section{The Finite Model Property for TGF (GFU)} \label{s:fmpgfu}
In this section, we establish the finite model property for \GFU{}, and hence also for \TGF{}: any satisfiable sentence in these logics has a finite model.
We also argue that our construction essentially realizes an optimal upper bound on the size of the produced models.

Let us fix a satisfiable \GFU{} sentence $\varphi$ in weak normal form. 
We want to construct a finite $\UU$-biquitous model for it. To this end, we will mimic the 
structure of the proof of the right-to-left direction
in Lemma \ref{l:reductiontogf} from Section \ref{s:altdec}.
There, given a non-$\UU$-biquitous finite model $\str{C}$ of $\varphi^\AAA$, for some $\AAA$,  we constructed an infinite join $\str{A}$ of 
its copies (having pairwise disjoint unnamed domains), then we defined a joinable collection of $D$-structures on some subdomains of $\str{A}$, and finally obtained an  infinite $\UU$-biquitous model of $\varphi$ as the join of this collection.

Here, we will proceed in a similar fashion: after some simple preparations, we will obtain a finite non-$\UU$-biquitous model $\str{C}$ of $\varphi^{\AAA}$, for some 
 $\AAA$, but instead of taking the join of an infinite number of copies of $\str{C}$, this time we will obtain the $\UU$-biquitous model $\str{A}$ by joining only a finite number of those. 
The role of the $D$-structures will be now played by analogous $S$-structures.
As previously, we will define a function $\hS$ which for a given pair of elements of $\check{A}$ will indicate a copy of $\str{C}$ (on a subdomain of $\str{A}$) to which this pair should be attached.
Due to the finiteness of the domain $A$, however, we will not be able to make $\hS$ injective (we recall that its counterpart $\hD$ in the previous section 
was injective). The challenge is to define it in such a way that 
the induced collection of $S$-structures will indeed be joinable. Crucially, we need a counterpart of Claim \ref{c:warmup1} to hold.

\subsection{Finite model construction}

\medskip\noindent{\bf Preparation of building blocks.}
Let $\str{M}$ be a $\UU$-biquitous, possibly infinite, model of $\varphi$. 
Let $\AAA$  be the set of $1$-types realized in $\str{M}$ by unnamed elements. Recall the sentence $\varphi^\AAA$ introduced in Section  \ref{s:decgfu},
obtained  by appending to $\varphi$ the conjuncts (\ref{expphi1})--(\ref{expphi3}), saying that  any model realizes precisely the $1$-types
from $\AAA$ and contains a rich collection of $\UU$-connections. Recall that $\varphi^\AAA$ is in weak normal form.
It is clear that $\varphi^\AAA$, treated as a \GF-sentence, is satisfiable---in fact, $\str{M}$ is a model of it.
Thus, by the finite model  property of \GF{}, 
it also has a finite model, which might, however, not be $\UU$-biquitous. We take such a finite model  $\str{C}_0 \models \varphi^\AAA$, 
and let $\str{C}$ be its doubling.
By Lemma \ref{l:doubling}, we have $\str{C} \models \varphi^\AAA$.

From this point on, the model $\str{C}_0$ will not play any role.
Instead, we will use yet another model $\str{B} \models \varphi^\AAA$,
obtained as the join of $5$ isomorphic copies of $\str{C}$, sharing the named domain and having pairwise disjoint unnamed domains. It remains a model of $\varphi^\AAA$ by Lemma \ref{l:join}.
For convenience, letting $\check{C} = \{1,\ldots,K\}$, we assume that  the unnamed domain of $\str{B}$ is $\check{B}:=\{1, \ldots, 5K \}$; and that for $m=0,\ldots, 4$ 
the structure on $\hat{B} \cup \{mK+1, \ldots, mK+K\}$ is isomorphic to $\str{C}$ (namely via the isomorphism that is the identity on $\hat{B}$ and otherwise maps any $mK+i$ to $i$).

\medskip\noindent
{\bf Initial finite non-$\UU$-biquitous model.} 
Let  $B_{k, \ell}=\hat{B} \cup (\check{B} \times \{ k \} \times \{ \ell \}$,  for $1 \le k, \ell \le 5K$,
and let every $\str{B}_{k, \ell}$ be the structure isomorphic to $\str{B}$, 
via the isomorphism working as the identity on $\hat{B}$ and as the natural projection $(b,k,\ell) \mapsto b$ on the unnamed elements.
As the $\str{B}_{k, \ell}$ have pairwise disjoint unnamed domains they are joinable. Let $\str{A}$ be their join.
Again, by Lemma \ref{l:join}, $\str{A} \models \varphi^\AAA$.

It is helpful to think that $\str{A}$ is organized in a square $5K \times 5K$ table of the unnamed parts (which from now on we sometimes refer to as $B$-\emph{cells}, each of which is further subdivided into $5$ $C$-cells), plus a single named part that is shared between all the joined substructures.  
For $m=0, \ldots, 4$, we denote by $\str{C}_{k,\ell,m}$ the structure $\str{B}_{k,\ell} \restr (\hat{B} \cup \{mK+1, \ldots, mK+K \} \times \{k \} \times \{\ell \})$.
We recall that each $\str{C}_{k,\ell,m}$ is isomorphic to $\str{C}$, so $\str{C}_{k, \ell, m} \models \varphi^\AAA$. 
 We will call the element $(s,k,l)$ the $s$-th element of $B_{k,l}$. 
When referring to the structures $\str{C}_{k,l,m}$ we will sometimes call them ${C}$-structures, and their domains $C$-domains. 
Similarly the $\str{B}_{k,l}$ will be called ${B}$-structures, and their domains $B$-domains.

Obviously, $\str{A}$ is not $\UU$-biquitous, but it is easy to see that only pairs 
consisting of two unnamed elements may be not connected by $\UU$. The following claim 
can be justified exactly as Claim \ref{c:warmup0}.
\begin{claim} \label{c:fmp0}
Any pair from $\hat{A} \times A$ is $\UU$-biqutious in $\str{A}$.
\end{claim}

\medskip \noindent
{\bf The joining function}. We now define the function 
$\hS: \check{A}   \times \check{A} \rightarrow \{1, \ldots, 5K \} \times \{1, \ldots, 5K \} \times \{0, \ldots, 4 \}$ whose purpose is to indicate which 
$C$-cell a given pair of elements can be potentially attached to.

Let $a_1, a_2 \in \check{A}$ be a pair of distinct elements. Assume $a_i \in C_{k_i, \ell_i, m_i}$ and $a_i$ is the $s_i$-th element in its $B$-cell, for $i=1,2$.
Choose $m$ such that $C_{s_1, s_2, m}$ does not contain the $k_1$-, $k_2$-, $\ell_1$-, nor $\ell_2$-th element of $B_{s_1, s_2}$ (as there are $5$ $C$-cells in every $B$-cell, such $m$ exists due to the pigeon hole principle).
Set $\hS(a_1, a_2):=(s_1, s_2, m)$. 
Such choices of $m$ ensure the following property, a counterpart of Claim \ref{c:warmup1}.
\begin{claim}\label{c:fmp1}
For every $a_1, a_2 \in \check{A}$ and every $b_1 \in \check{C}_{\hS(a_1,a_2)}$
there is no $b_2 \in \check{A}$ such that $a_1$ or $a_2$ belongs to $C_{\hS(b_1, b_2)}$.
\end{claim}

\medskip\noindent
{\bf Pattern elements.} We now distinguish in each $C$-cell a pair of distinct elements whose connections with the rest
of this $C$-cell will serve as a template for attaching to this cell the non-$\UU$-biqutious pairs of elements (and making
them $\UU$-biquitous) assigned to this $C$-cell by the function $\hS$.  
For every $1\le k,\ell \le 5K$, let $\alpha^k=\type{\str{B}}{k}$ and $\alpha^\ell=\type{\str{B}}{\ell}$. 
Note that $\hS$ plans to join to a $C$-cell $C_{k,l,m}$ elements $a_1, a_2$ whose numbers in their $B$-cells are $k$ and $\ell$, respectively.
This means that their $1$-types are $\alpha^k$ and $\alpha^{\ell}$.

As $\str{C} \models \varphi^\AAA$, by formula (\ref{expphi2}) and Lemma \ref{c:uconnected}, there are distinct elements $e_1, e_2 \in C$ such that
$\type{\str{C}}{e_1}=\alpha^k$,  $\type{\str{C}}{e_2}=\alpha^{\ell}$, $(e_1, e_2), (e_2, e_1) \in \UU^\str{C}$, and if $\alpha^k=\alpha^\ell$ then 
$e_1$ and $e_2$ are indistinguishable in $\str{C}$.
We choose the pattern elements  $e_1^{k,\ell,m}$, $e_2^{k,\ell,m}$ in $\str{C}_{k,\ell,m}$ to be the corresponding copies of $e_1$ and $e_2$
in each of $\str{C}_{k,\ell,m}$.

\medskip \noindent{\bf Definition of $S$-structures.} 
For every pair of distinct elements $a_1, a_2 \in \check{A}$ we distinguish three subdomains of $\str{A}$ arising by 
adding one or both of $a_1, a_2$ to $C_{\hS(a_1,a_2)}$.  Formally:
\begin{itemize}
	\item $S^{\sss a_1a_2}_{\hS(a_1,a_2)}:=\{a_1, a_2\} \cup {C}_{\hS(a_1, a_2)}$ 
	\item $S^{\sss a_1\cdot}_{\hS(a_1,a_2)}:=\{a_1\} \cup {C}_{\hS(a_1, a_2)}$ 
	\item $S^{\sss \cdot a_2}_{\hS(a_1,a_2)}:=\{a_2\} \cup {C}_{\hS(a_1, a_2)}$
\end{itemize}

We define structures on the above subdomains as follows:
\begin{itemize}
	\item $\str{S}^{\sss a_1a_2}_{\hS(a_1,a_2)}$ is isomorphic to the partial doubling of $\str{C}_{\hS(a_1,a_2)}$ duplicating the set $\{ e_1^{\hS(a_1,a_2)}, e_2^{\hS(a_1,a_2)} \}$,
	\item $\str{S}^{\sss a_1\cdot}_{\hS(a_1,a_2)}$ is isomorphic to the partial doubling of $\str{C}_{\hS(a_1,a_2)}$ duplicating the singleton $\{ e_1^{\hS(a_1,a_2)} \}$, 
\item $\str{S}^{\sss \cdot a_2}_{\hS(a_1,a_2)}$ is isomorphic to the partial doubling of $\str{C}_{\hS(a_1,a_2)}$ duplicating the singleton $\{ e_2^{\hS(a_1,a_2)} \}$, 
\end{itemize}
via the natural isomorphism, making $a_i$ the ``twin'' of  $e_i^{\hS(a_1,a_2)}$ for $i=1,2$.

We transfer the terminology from $D$-structures: ${S}$-structures from the first bullet are called \emph{double-headed}, while those from the two other are \emph{single-headed}
(\emph{left-} and \emph{right-headed}, respectively).  The elements $a_1,a_2$ are called their \emph{heads} and the sets $C_{\hS(a_1,a_2)}$ are their \emph{bodies}.
Observe that, by the definition of $\hS$, for any ${S}$-structure, its heads do not belong to its body.
Note that the definition of left-headed (right-headed) $\str{S}$-structures is sound as it does not depend on
the choice of $a_2$ ($a_1$).
The following claim follows from Lemma \ref{l:doubling} (about partial doublings).
\begin{claim} 
Every $S$-structure is a model of $\varphi$.
\end{claim} 

\medskip\noindent
{\bf Basic properties of $S$-structures.}
We now state a series of claims concerning properties of the $\str{S}$-structures
that  will be used in a moment to prove the correctness of our finite model construction.
 The next four claims correspond to Claims \ref{c:warmup2}-\ref{c:warmup5} and
follow straightforwardly from the construction of ${S}$-structures. 
In particular, the crucial Claim \ref{c:fmp5} is true since the connections between the two heads in $\str{S}$ are obtained from the connections of their respective twins which form a $\UU$-biquitous pair of elements by construction.  

\begin{claim} \label{c:fmp2}
Let $a$ be an element in any $\str{S}$-structure $\str{S}$. Then $\type{\str{S}}{a}=\type{\str{A}}{a}$.
\end{claim}

\begin{claim} \label{c:fmp3}
Every $\str{S}$-structure restricted to its body is isomorphic to
the restriction of $\str{A}$ to this body.
\end{claim}

\begin{claim} \label{c:fmp4}
$\str{S}^{\sss a_1\cdot}_{\hS(a_1,a_2)}$ and $\str{S}^{\sss \cdot a_2}_{\hS(a_1,a_2)}$
are substructures of $\str{S}^{\sss a_1a_2}_{\hS(a_1,a_2)}$, for all $a_1,a_2 \in \check{A}$.
\end{claim}

\begin{claim} \label{c:fmp5}
The pair of the heads of every double-headed ${S}$-structure $\str{S}$ is $\UU$-biquitous in $\str{S}$.
\end{claim}

The following observation can be proved literally exactly as Claim \ref{c:warmup6}. 

\begin{claim} \label{c:fmp6}
\begin{enumerate}[(i)]
\item Every single-headed $\str{S}$-structure is joinable with $\str{A}$. 
\item If $(a_1,a_2) \in \check{A}$ is a pair of distinct elements that is non-$\UU$-biquitous in $\str{A}$,
then also $\str{S}_{\hS(a_1,a_2)}^{a_1a_2}$ is joinable with $\str{A}$.
\end{enumerate}
\end{claim}

\medskip\noindent
{\bf The join.}
For every pair $a_1, a_2 \in \check{A}$ of distinct elements that is $\UU$-biquitous neither in $\str{A}$ nor in
any of the single-headed $\str{S}$-structures select one
of the double-headed structures $\str{S}^{\sss a_1a_2}_{\hS(a_1,a_2)}$, $\str{S}^{\sss a_2a_1}_{\hS(a_2,a_1)}$.
Let $\mathcal{K}$ be the set consisting of so selected double headed $\str{S}$-structures,
all single-headed $\str{S}$-structures\footnote{In the conference version of this paper the finite model construction was presented in a slightly different way, but, essentially, it
could be seen as a process of joining $\str{A}$ successively with some double-headed $S$-structures. That construction had a flaw, since when executed in ``wrong'' order the process 
could lead to some conflicts (as not all pairs of double-headed structures are joinable). Here we repair this flaw by joining with $\str{A}$ all single-headed structures and
only a safely selected subset of the double-headed $S$-structues.\label{note:glitch}}
 and $\str{A}$. 
The following claim is a counterpart of Claim \ref{warmupjoin}. 
Its proof is similar, but we note that the $S$-structures overlap among each other in a less trivial way than the $D$-structures did in the infinite-model construction in Section~\ref{s:altdec}.
In this proof, we need to take special care in the first subcase.

\begin{claim} \label{c:fmpjoin}
 $\mathcal{K}$ is joinable.
\end{claim}

\begin{proof} 
Every ${S}$-structure from $\mathcal{K}$ is joinable with $\str{A}$ by Claim \ref{c:warmup6}. Consider any pair $\str{S}_1$, $\str{S}_2$ 
of distinct ${S}$-structures from $\mathcal{K}$. Let $\bar{a}$ be a tuple of distinct unnamed elements contained in both of them. We need to
show that $\bar{a}$ has the same type in both $\str{S}_1, \str{S}_2$ or is unguarded in at least one of them.
We consider three cases:
\begin{enumerate}
	\item Both $\str{S}_1, \str{S}_2$ are single-headed. If they share the body and have different heads then $\bar{a}$ is contained in this shared body
	and its type is identical in both structures by Claim \ref{c:fmp3}. 
	
	If they share the body but have the same head then they are of the form (a) $\str{S}_{\hS(a,b_1)}^{a\cdot}$, $\str{S}_{\hS(a,b_2)}^{a \cdot}$, 
	(b)	$\str{S}_{\hS(b_1,a)}^{\cdot a}$, $\str{S}_{\hS(b_2, a)}^{ \cdot a}$, or (c) 	$\str{S}_{\hS(a,b_1)}^{a \cdot}$, $\str{S}_{\hS(b_2,a)}^{\cdot a}$ for some $a, b_1,b_2$. 
	In subcases (a) and (b), we obtain $\hS(a,b_1)=\hS(a,b_2)$ ($\hS(b_1,a)=\hS(b_2,a)$) and as already observed both structures are the same as their definition do not depend on the choice of
	$b_i$. In subcase (c), we obtain $\hS(a,b_1)=\hS(b_2,a)=(k,k,m)$
	for $k$ being (the same) number of elements $a, b_1,b_2$ in their $B$-cells, and some $m$. In the first of the considered $S$-structures,  $a$ is the twin of $e_1^{\hS(a,b_1)}$ and in the other it is the twin of $e_2^{\hS(a, b_1)}$. However, both the pattern elements $e_i^{\hS(a_1,a_2)}$
			are indistinguishable in $\str{C}_{\hS(a,b_1)}$ in this case, so again the considered $\str{S}$-structures are in fact identical.		
	
		If the bodies of $\str{S}_1, \str{S}_2$ are different then either the heads are the same element $a$ and then $\bar{a}=a$, or the head $a$ of one structure is in the body of the other 
	and then again $\bar{a}=a$ (note that by Claim \ref{c:fmp1} it cannot happen that the head of one of the structures is in the body of the other and, simultaneously, vice versa).
	In both cases $\bar{a}$ is just a single element, and its type in both $S$-structures is the same by Claim \ref{c:fmp2}.

	\item  $\str{S}_1$ is single-headed, $\str{S}_2$ is double-headed (or vice versa). We reason precisely as in the case of $D$-structures.

	\item Both $\str{S}_1$ and $\str{S}_2$ are double-headed. 
	We reason precisely as in the case of $D$-structures: 

				\end{enumerate}
\end{proof}

Let $\str{A}^\dagger$ be the join of $\mathcal{K}$.  By  Lemma \ref{l:join} we get that $\str{A}^\dagger \models \varphi$.
We argue that it is $\UU$-biquitous exactly as in the case of the structure $\str{A}^*$ from Section \ref{s:decgfu}.

\subsection{Size of models}

We now  estimate the size of finite models that can be produced by means of our construction. 
First, let us recall the result on the size of minimal finite models in pure \GF{}  \cite{BGO14}. 

\begin{theorem}(\cite{BGO14}) \label{t:gfsize}
	\GF{} (with constants and equalities) has the finite model property: Every satisfiable sentence has a model of size bounded doubly exponentially in its
	length. More specifically, the size of a minimal model of a weak normal form sentence is bounded exponentially in the size of its signature and doubly exponentially in
	the maximal arity of this signature.
\end{theorem}

The doubly exponential bound from the first part of the above theorem follows immediately from the statement of Thm.~1.2 (page 21) in \cite{BGO14}, if the trivial query $\bot$ is taken.
The second part of the theorem is not explicitly stated in the original paper, but it is not difficult to derive it inspecting the original proof:
Indeed, given a sentence in normal form over a signature $\sigma$ with maximal arity $s$, the size of a minimal finite model 
 can be bounded by $|\str{J}|^{\mathcal{O}(s)}$ for some so-called 
\emph{invariant}	$\str{J}$ whose size is bounded by the number of guarded  types over $\sigma$. 
By the terminology in \cite{BGO14}, a type is \emph{guarded} if all its variables are contained in a single positive fact. 
	Further, the number of  $\ell$-types over $\sigma$ is $2^{\mathcal{O}(r(\ell+m)^s)}$, where $r$ is the number of relation symbols (the size of $\sigma$) and $m$ the number of
	constants in $\sigma$. Of course, every guarded type is an $\ell$-type for some $\ell \le s$, so the number of guarded types
	is $s\cdot2^{\mathcal{O}(r(s+m)^s)}$, which is  $2^{\mathcal{O}(r(s+m)^s)}$. Thus the size of the model is 
	$2^{\mathcal{O}(sr(s+m)^{s})}$.
	As each of $r$ and $m$ is bounded by $|\sigma|$, the claim follows.
	
\medskip

Our weak normal form is essentially the same as the normal form in \cite{BGO14},
so the estimation in the second part of the theorem applies to models of $\varphi^\AAA$.

Assume now we want to  construct a finite model of a satisfiable formula $\varphi_0$ over a signature $\sigma_0$. 
We first convert $\varphi_0$ into a weak normal form sentence $\varphi$ (over an extended signature $\sigma$) as guaranteed
by Proposition \ref{prop:normalization}.
We take an
arbitrary model $\str{M} \models \varphi$. Next we 
append to $\varphi$ the auxiliary conjuncts obtaining a sentence $\varphi^\AAA$, send $\varphi^\AAA$
to a blackbox producing a finite but generally non-$\UU$-biquitous model $\str{C}_0 \models \varphi^\AAA$, and form models $\str{C}$, $\str{B}$, $\str{A}$.
The last one is then turned into a $\UU$-biquitous model $\str{A}^\dagger$ without changing its domain. 
By Proposition \ref{prop:normalization}, $|\varphi|$ is polynomial in $|\varphi_0|$. $|\varphi^\AAA|$ is
exponential in $|\varphi|$ in the case without constants and doubly exponential in the case with constants.
This follows from the fact that $\varphi^\AAA$ contains the conjuncts (\ref{expphi2}) and (\ref{expphi3}) whose size
is polynomial in the number of $1$-types over $\sigma$. So, we need to be careful and avoid
estimating the size of $\str{A}$ only in terms of the length of $\varphi^\AAA$.

As the external blackbox procedure in the above approach we can use any procedure constructing a finite model of a satisfiable \GF{} formula. 
Let as  assume that we use the model produced by the construction from \cite{BGO14}. 
By Theorem~\ref{t:gfsize}, the size of the finite model we get is bounded exponentially by the size and doubly exponentially by the maximal arity
of the signature of $\varphi^\AAA$, which is the same as the signature of $\varphi$ and $\sigma$. As the size and the maximal arity of $\sigma$ are bounded by $|\varphi|$
which, by Proposition \ref{prop:normalization}, is polynomial in $|\varphi_0|$, eventually  our bound on the size of $\str{C}_0$ is doubly exponential in the size of the input formula $|\varphi_0|$.

Recall that $|C| = 2\cdot |C_0|$, $|B|=5\cdot |C| = 10\cdot |C_0|$, and  $|A_i|=|B|^2 \cdot |B|=(10\cdot |C_0|)^3$ which is still doubly
exponential in $|\varphi_0|$. Thus we get:

\smallskip
\begin{theorem} \label{t:fmpgfu}
	Every satisfiable \TGF{} (\GFU{}) formula $\varphi$  has a finite model of size bounded doubly exponentially by the length of $\varphi$.
\end{theorem}

\medskip
This bound is essentially optimal, since even in \GF{} without constants and equality one can construct a family of satisfiable sentences $\varphi_i$, each of them of length polynomial in
$i$, but having only models of size at least $2^{2^i}$. This is implicit in \cite{Gra99}.

The finite model property of \TGF{} (\GFU) implies that its finite satisfiability problem is equal to its satisfiability problem and
thus it is \TwoExpTime-complete in the absence of constants and \TwoNExpTime-complete with constants as shown in Section \ref{s:compgfu}.

\section{Related Work}
{\bf Prior results on similar logics.}
In terms of expressivity, the fragments $\TGF$ and $\gftimes$ are closely related to the fragment $\gfcross$
proposed by Bourhis, Morak, and Pieris \cite{DBLP:conf/ijcai/BourhisMP17}, which extends $\gf$ with
\emph{cross products} (allowing to capture statements like ``all
elephants are bigger than all mice'' as in
\cite{DBLP:conf/dlog/RudolphKH08}). Still, a closer inspection reveals the substantial difference that
$\gfcross$, inspired by the typical setting in database theory, imposes a separation
into a set of ground facts (the data) and a constant-free theory
(the schema) \cite{perscomm}. This significant restriction was not explicitly mentioned by the authors but becomes obvious after investigating the underlying constructions and proofs.
As a consequence of this restriction, $\gfcross$ is in fact subsumed
by the fragment $\mathcal{GF}{\mid}\mathcal{FO}^2$, studied before by Kazakov in his PhD thesis
\cite{Kazakov:06:Phd} more than a decade earlier, though by means of different techniques. Kazakov used a resolution-based procedure,
to show satisfiability in $\mathcal{GF}{\mid}\mathcal{FO}^2$ to
be in \twoexptime, and in \nexptime in case of bounded predicate
arities~\cite{Kazakov:06:Phd}. Instead of resolution, the proof of
the \twoexptime upper bound for $\gfcross$ by Bourhis et al.~\cite{DBLP:conf/ijcai/BourhisMP17} uses a reduction to
satisfiability in plain $\gf$. These prior results emphasize that the unrestricted
availability of constants is crucial for the \twonexptime-hardness of
full $\gftimes$ and $\TGF$, which is also reflected by our hardness proof.

\medskip\noindent
{\bf Transitive guards.}
An important, practically relevant and theoretically challenging modelling feature is \emph{transitivity} of a binary relation. Neither \FOt, nor \GF, and also not \TGF{} allow for axiomatising transitivity.  As a remedy, it has been suggested to provide a set of dedicated binary predicate names whose transitivity is ``hard-wired'' into the logic, that is, externally imposed by the semantics.
As it turned out, when doing so, one has to be very careful not to lose decidability. \FOt{} with transitive relations was shown undecidable \cite{GOR99}, 
and also the unrestricted use of transitive relations in \GF{} is known to lead to undecidability \cite{Gra99}. Actually, this even holds for \GFt{} ($=\,$\FOt$\,\cap\,$\GF), the two-variable guarded fragment~\cite{GMV99}.

One way out is to restrain the use of transitive relations so that they only are allowed to occur in guards. Indeed, satisfiability of \GF{} with \emph{with transitive guards},
\GFTG{} (with equality), was shown to be decidable and, in fact, \TwoExpTime-complete \cite{ST04}. More recently this result was lifted to \TGFTG{} (without equality) \cite{DBLP:conf/lpar/KieronskiM20}. Regarding the finite model case,
finite satisfiability of \GFtTG{} (with equality) was shown to be decidable and \twoexptime-complete \cite{KT18}. We note that \GFtTG{} does not have the FMP, hence
satisfiability and finite satisfiability are different problems.
In our own work~\cite{DBLP:conf/lics/KieronskiR21}, we first show that finite satisfiability of full \GFTG{} (with equality) is decidable and \TwoExpTime-complete and then use this result to 
show decidability and \twoexptime-completeness of finite satisfiability in \TGFTG{} (without equality). Actually, the proof of the latter goes along the same lines  as the FMP proof for \TGF{} in Section \ref{s:fmpgfu} of this paper, that is, essentially, by a reduction to \GFTG{}.\footnote{Unfortunately, the argument inherits the flaw mentioned in Footnote~\ref{note:glitch}, but can be fixed with a similar repair.}

All the above results concerning logics with TG assume the absence of constants. While \GFTG{} with equality and constants can be easily shown to be undecidable,\footnote{Using a constant and equality one can enforce a transitive relation $T$ to be true everywhere, that is to behave like $\UU$ in \GFU. This way we can simulate \TGF{} with equality%
.} we conjecture that in the absence of equality both \GFTG{} and \TGFTG{} with constants are decidable.

\medskip\noindent
{\bf One-dimensionality.} Let us call a first-order formula \emph{one-dimensional} if every maximal block of existential (or universal) quantifiers in it leaves at most one variable free.
For example every \FOt{} formula and the formula $\forall x \exists yz P(x,y,z)$ are one-dimensional but $\forall xy \exists z P(x,y,z)$ is not. 
The one-dimensional restriction of \TGF{}, denoted \TGF$_1$, has been studied by the first author of this paper \cite{Kie19} and turns out to be decidable and to have the finite (doubly exponential) model property even if equalities are allowed. The satisfiability problem remains \twonexptime-complete (both in the case with and without equality). The complexity for the variant without constants drops down, but, interestingly, it depends on the presence of equality: the satisfiability problem is \twoexptime-complete with equality and only \nexptime-complete without it. In the latter case, every satisfiable formula has a model of just single exponential size.

We note that the one-dimensionality restriction does not hinder  translations from modal/description logics, as they are inherently one-dimensional.
Thus, \TGF$_1$ embeds, e.g., the description logic $\mathcal{ALC}$ plus inverse roles ($\mathcal{I}$), nominals ($\mathcal{O}$), role hierarchies ($\mathcal{H}$), and any Boolean combination of roles (including their negations).

\medskip\noindent
{\bf Alternative FMP proof.} Our FMP proof for \TGF{} in Section \ref{s:fmpgfu} builds an additional level on top of a finite model construction for \GF{}. This makes the overall
construction quite intricate, since, even though we could in principle use any construction for \GF{}, the known ones (by Grädel~\cite{Gra99}; Bárány, Gottlob, and Otto~\cite{BGO14}; and a simplification of the latter in Pratt-Hartmann’s book~\cite{PH23})
are technically quite involved by themselves. A surprisingly simple FMP proof for \GF{} and \TGF{}, proceeding via the \emph{probabilistic method},  was recently discovered by Fiuk \cite{Fiu26}.
A (not very complicated) analysis shows that it also gives a doubly exponential upper bound on the size of constructed models. Fiuk also proposes a derandomisation of his probabilistic algorithm, relying on algebraic hash functions to simulate randomness. This gives an alternative \emph{constructive} proof
of the FMP for \TGF{}, although its analysis is no longer that simple.

\section{Conclusion}
In this article, we presented the triguarded fragment of $\fo$, a decidable fragment of first-order predicate logic which subsumes both the guarded and the two-variable fragment, while requiring a restricted use of equality. As exemplified in the introduction, we expect the fragment to be useful for modeling in knowledge representation and beyond.  

We determined the computational complexity of satisfiability checking in this fragment, both for the bounded and unbounded arity case (\textsc{NExpTime}-complete and \textsc{N2ExpTime}-complete, respectively). We discussed that diverse natural extensions of the fragment lead to undecidability. 

We then investigated complexity and expressive power of the triguarded fragment from the point of view of ontology-based data access, leading to \textsc{NP}-completeness in data complexity and the insight that as a Boolean query language over databases, the triguarded fragment has the same expressivity as disjunctive datalog.

Turning to the computational limitations of the triguarded fragment, we showed that slight increases of expressivity -- including a liberal use of equality or (further) relaxations of guardedness -- will lead to undecidability of satisfiability. Likewise, checking conjunctive query entailment from triguarded theories is undecidable. This apparent brittleness of decidability suggests that the triguarded fragment is already quite on the verge of decidability.
   
We concluded the paper by establishing that the triguarded fragment exhibits the finite model property, and proved a tight doubly exponential bound on the model size.   
 
\medskip

For future research, we see the following avenues: On the theoretical level, the accommodation of transitive guards still poses some open questions regarding decidability in the presence of constants. On a more general note, central logical properties of the triguarded fragment seem to be unknown, such as interpolation and definability.

On the more practical side, it might be worthwhile to investigate practically efficient algorithms (along the lines of tableaux methods or implementations of the (disjunctive) chase) for satisfiability checking. Experience in knowledge representation -- notably in the area of ontological reasoning -- has shown that very high worst-case complexities do not always preclude  implementations which are efficient in the context of practical use cases.

\begin{acks}
  Emanuel Kiero\'nski is supported by Polish National Science Centre
  grant No 2021/41/B/ST6/00996.  Sebastian Rudolph acknowledges
  support by the European Research Council through the ERC
  Consolidator Grant 771779 (DeciGUT) and the
  Wolfgang-Pauli-Institute. Mantas \v{S}imkus has been supported by
  the Austrian Science Fund (FWF) projects P30360 and P30873, and the
  Vienna Business Agency. 
We are also grateful to Pierre Bourhis, Michael Morak, and Andreas Pieris for clarifying some questions regarding their paper \cite{DBLP:conf/ijcai/BourhisMP17}.
\end{acks}

\bibliographystyle{ACM-Reference-Format}
\bibliography{mybib,references}


\begin{thebibliography}{37}


\ifx \showCODEN    \undefined \def \showCODEN     #1{\unskip}     \fi
\ifx \showDOI      \undefined \def \showDOI       #1{#1}\fi
\ifx \showISBNx    \undefined \def \showISBNx     #1{\unskip}     \fi
\ifx \showISBNxiii \undefined \def \showISBNxiii  #1{\unskip}     \fi
\ifx \showISSN     \undefined \def \showISSN      #1{\unskip}     \fi
\ifx \showLCCN     \undefined \def \showLCCN      #1{\unskip}     \fi
\ifx \shownote     \undefined \def \shownote      #1{#1}          \fi
\ifx \showarticletitle \undefined \def \showarticletitle #1{#1}   \fi
\ifx \showURL      \undefined \def \showURL       {\relax}        \fi
\providecommand\bibfield[2]{#2}
\providecommand\bibinfo[2]{#2}
\providecommand\natexlab[1]{#1}
\providecommand\showeprint[2][]{arXiv:#2}

\bibitem[Andr\'{e}ka et~al\mbox{.}(1998)]%
        {ABN98:GF}
\bibfield{author}{\bibinfo{person}{Hajnal Andr\'{e}ka}, \bibinfo{person}{Johan
  F. A.~K. van Benthem}, {and} \bibinfo{person}{Istv\'{a}n N\'{e}meti}.}
  \bibinfo{year}{1998}\natexlab{}.
\newblock \showarticletitle{Modal languages and bounded fragments of predicate
  logic}.
\newblock \bibinfo{journal}{\emph{J.\ of Philosophical Logic}}
  \bibinfo{volume}{27}, \bibinfo{number}{3} (\bibinfo{year}{1998}),
  \bibinfo{pages}{217--274}.
\newblock


\bibitem[Baader et~al\mbox{.}(2007)]%
        {dlhandbook}
\bibfield{editor}{\bibinfo{person}{Franz Baader}, \bibinfo{person}{Diego
  Calvanese}, \bibinfo{person}{Deborah McGuinness}, \bibinfo{person}{Daniele
  Nardi}, {and} \bibinfo{person}{Peter Patel-Schneider}} (Eds.).
  \bibinfo{year}{2007}\natexlab{}.
\newblock \bibinfo{booktitle}{\emph{The Description Logic Handbook: Theory,
  Implementation, and Applications} (\bibinfo{edition}{second} ed.)}.
\newblock \bibinfo{publisher}{Cambridge University Press}.
\newblock


\bibitem[Baader et~al\mbox{.}(2017)]%
        {DBLP:books/daglib/0041477}
\bibfield{author}{\bibinfo{person}{Franz Baader}, \bibinfo{person}{Ian
  Horrocks}, \bibinfo{person}{Carsten Lutz}, {and} \bibinfo{person}{Ulrike
  Sattler}.} \bibinfo{year}{2017}\natexlab{}.
\newblock \bibinfo{booktitle}{\emph{An Introduction to Description Logic}}.
\newblock \bibinfo{publisher}{Cambridge University Press}.
\newblock
\showISBNx{978-0-521-69542-8}


\bibitem[B{\'{a}}r{\'{a}}ny et~al\mbox{.}(2014)]%
        {BGO14}
\bibfield{author}{\bibinfo{person}{Vince B{\'{a}}r{\'{a}}ny},
  \bibinfo{person}{Georg Gottlob}, {and} \bibinfo{person}{Martin Otto}.}
  \bibinfo{year}{2014}\natexlab{}.
\newblock \showarticletitle{Querying the Guarded Fragment}.
\newblock \bibinfo{journal}{\emph{Logical Methods in Computer Science}}
  \bibinfo{volume}{10}, \bibinfo{number}{2} (\bibinfo{year}{2014}).
\newblock


\bibitem[B{\'{a}}r{\'{a}}ny et~al\mbox{.}(2015)]%
        {DBLP:journals/jacm/BaranyCS15}
\bibfield{author}{\bibinfo{person}{Vince B{\'{a}}r{\'{a}}ny},
  \bibinfo{person}{Balder ten Cate}, {and} \bibinfo{person}{Luc Segoufin}.}
  \bibinfo{year}{2015}\natexlab{}.
\newblock \showarticletitle{Guarded Negation}.
\newblock \bibinfo{journal}{\emph{J.\ of the ACM}} \bibinfo{volume}{62},
  \bibinfo{number}{3} (\bibinfo{year}{2015}), \bibinfo{pages}{22:1--22:26}.
\newblock


\bibitem[Bienvenu et~al\mbox{.}(2014)]%
        {BCL14}
\bibfield{author}{\bibinfo{person}{Meghyn Bienvenu}, \bibinfo{person}{Balder
  ten Cate}, \bibinfo{person}{Carsten Lutz}, {and} \bibinfo{person}{Frank
  Wolter}.} \bibinfo{year}{2014}\natexlab{}.
\newblock \showarticletitle{Ontology-Based Data Access: {A} Study through
  Disjunctive Datalog, CSP, and {MMSNP}}.
\newblock \bibinfo{journal}{\emph{{ACM} Trans. Database Syst.}}
  \bibinfo{volume}{39}, \bibinfo{number}{4} (\bibinfo{year}{2014}),
  \bibinfo{pages}{33:1--33:44}.
\newblock


\bibitem[Blackburn and Van~Benthem(2006)]%
        {blackburn:inria-00119856}
\bibfield{author}{\bibinfo{person}{Patrick Blackburn} {and}
  \bibinfo{person}{Johan Van~Benthem}.} \bibinfo{year}{2006}\natexlab{}.
\newblock \showarticletitle{{Modal logic: a Semantic Perspective}}.
\newblock In \bibinfo{booktitle}{\emph{{Handbook of Modal Logic}}},
  \bibfield{editor}{\bibinfo{person}{Frank~Wolter Patrick~Blackburn, Johan
  van~Benthem}} (Ed.). \bibinfo{publisher}{{Elsevier}}, \bibinfo{pages}{1--82}.
\newblock
\urldef\tempurl%
\url{https://hal.inria.fr/inria-00119856}
\showURL{%
\tempurl}


\bibitem[B\"{o}rger et~al\mbox{.}(1997)]%
        {BGG}
\bibfield{author}{\bibinfo{person}{Egon B\"{o}rger}, \bibinfo{person}{Erich
  Gr\"{a}del}, {and} \bibinfo{person}{Yuri Gurevich}.}
  \bibinfo{year}{1997}\natexlab{}.
\newblock \bibinfo{booktitle}{\emph{The Classical Decision Problem}}.
\newblock \bibinfo{publisher}{Springer}.
\newblock


\bibitem[Borgida(1996)]%
        {DBLP:journals/ai/Borgida96}
\bibfield{author}{\bibinfo{person}{Alexander Borgida}.}
  \bibinfo{year}{1996}\natexlab{}.
\newblock \showarticletitle{On the Relative Expressiveness of Description
  Logics and Predicate Logics}.
\newblock \bibinfo{journal}{\emph{Artif. Intell.}} \bibinfo{volume}{82},
  \bibinfo{number}{1-2} (\bibinfo{year}{1996}), \bibinfo{pages}{353--367}.
\newblock
\urldef\tempurl%
\url{https://doi.org/10.1016/0004-3702(96)00004-5}
\showDOI{\tempurl}


\bibitem[Bourhis et~al\mbox{.}({[n.\,d.]})]%
        {perscomm}
\bibfield{author}{\bibinfo{person}{Pierre Bourhis}, \bibinfo{person}{Michael
  Morak}, {and} \bibinfo{person}{Andreas Pieris}.}
  \bibinfo{year}{[n.\,d.]}\natexlab{}.
\newblock \bibinfo{title}{Personal {C}ommunication (23rd of {J}uly 2018)}.
\newblock
\newblock


\bibitem[Bourhis et~al\mbox{.}(2017)]%
        {DBLP:conf/ijcai/BourhisMP17}
\bibfield{author}{\bibinfo{person}{Pierre Bourhis}, \bibinfo{person}{Michael
  Morak}, {and} \bibinfo{person}{Andreas Pieris}.}
  \bibinfo{year}{2017}\natexlab{}.
\newblock \showarticletitle{Making Cross Products and Guarded Ontology
  Languages Compatible}. In \bibinfo{booktitle}{\emph{Proc.\,of {IJCAI} 2017}}.
\newblock
\urldef\tempurl%
\url{https://doi.org/10.24963/ijcai.2017/122}
\showDOI{\tempurl}


\bibitem[Fiuk(2026)]%
        {Fiu26}
\bibfield{author}{\bibinfo{person}{Oskar Fiuk}.}
  \bibinfo{year}{2026}\natexlab{}.
\newblock \showarticletitle{{Random Models and Guarded Logic}}. In
  \bibinfo{booktitle}{\emph{43rd International Symposium on Theoretical Aspects
  of Computer Science (STACS 2026)}} \emph{(\bibinfo{series}{Leibniz
  International Proceedings in Informatics (LIPIcs)},
  Vol.~\bibinfo{volume}{364})}. \bibinfo{publisher}{Schloss Dagstuhl --
  Leibniz-Zentrum f{\"u}r Informatik}, \bibinfo{address}{Dagstuhl, Germany},
  \bibinfo{pages}{37:1--37:21}.
\newblock
\showISBNx{978-3-95977-412-3}
\showISSN{1868-8969}
\urldef\tempurl%
\url{https://doi.org/10.4230/LIPIcs.STACS.2026.37}
\showDOI{\tempurl}


\bibitem[Ganzinger et~al\mbox{.}(1999)]%
        {GMV99}
\bibfield{author}{\bibinfo{person}{Harald Ganzinger},
  \bibinfo{person}{Christoph Meyer}, {and} \bibinfo{person}{Margus Veanes}.}
  \bibinfo{year}{1999}\natexlab{}.
\newblock \showarticletitle{The Two-Variable Guarded Fragment with Transitive
  Relations}. In \bibinfo{booktitle}{\emph{14th Annual IEEE Symposium on Logic
  in Computer Science, {LICS 1999}}}. \bibinfo{pages}{24--34}.
\newblock


\bibitem[Gr{\"{a}}del(1998)]%
        {DBLP:conf/dlog/Gradel98}
\bibfield{author}{\bibinfo{person}{Erich Gr{\"{a}}del}.}
  \bibinfo{year}{1998}\natexlab{}.
\newblock \showarticletitle{Description Logics and Guarded Fragments of First
  Order Logic}. In \bibinfo{booktitle}{\emph{Proc.\,of DL 1998}}.
\newblock


\bibitem[Gr{\"a}del(1999)]%
        {Gra99}
\bibfield{author}{\bibinfo{person}{Erich Gr{\"a}del}.}
  \bibinfo{year}{1999}\natexlab{}.
\newblock \showarticletitle{On The Restraining Power of Guards}.
\newblock \bibinfo{journal}{\emph{J. Symb. Log.}} \bibinfo{volume}{64},
  \bibinfo{number}{4} (\bibinfo{year}{1999}), \bibinfo{pages}{1719--1742}.
\newblock


\bibitem[Gr{\"{a}}del et~al\mbox{.}(1997)]%
        {DBLP:journals/bsl/GradelKV97}
\bibfield{author}{\bibinfo{person}{Erich Gr{\"{a}}del},
  \bibinfo{person}{Phokion~G. Kolaitis}, {and} \bibinfo{person}{Moshe~Y.
  Vardi}.} \bibinfo{year}{1997}\natexlab{}.
\newblock \showarticletitle{On the decision problem for two-variable
  first-order logic}.
\newblock \bibinfo{journal}{\emph{Bulletin of Symbolic Logic}}
  \bibinfo{volume}{3}, \bibinfo{number}{1} (\bibinfo{year}{1997}),
  \bibinfo{pages}{53--69}.
\newblock


\bibitem[Gr{\"a}del et~al\mbox{.}(1999)]%
        {GOR99}
\bibfield{author}{\bibinfo{person}{E. Gr{\"a}del}, \bibinfo{person}{M. Otto},
  {and} \bibinfo{person}{E. Rosen}.} \bibinfo{year}{1999}\natexlab{}.
\newblock \showarticletitle{Undecidability results on two-variable logics}.
\newblock \bibinfo{journal}{\emph{Archiv f{\"{u}}r Mathematische Logik und
  Grundlagenforschung}} \bibinfo{volume}{38}, \bibinfo{number}{4-5}
  (\bibinfo{year}{1999}), \bibinfo{pages}{313--354}.
\newblock


\bibitem[Kazakov(2006)]%
        {Kazakov:06:Phd}
\bibfield{author}{\bibinfo{person}{Yevgeny Kazakov}.}
  \bibinfo{year}{2006}\natexlab{}.
\newblock \emph{\bibinfo{title}{Saturation-Based Decision Procedures for
  Extensions of the Guarded Fragment}}.
\newblock \bibinfo{thesistype}{Ph.\,D. Dissertation}.
  \bibinfo{school}{Universit{\"a}t des Saarlandes},
  \bibinfo{address}{Saarbr{\"u}cken, Germany}.
\newblock


\bibitem[Kieronski(2019)]%
        {Kie19}
\bibfield{author}{\bibinfo{person}{Emanuel Kieronski}.}
  \bibinfo{year}{2019}\natexlab{}.
\newblock \showarticletitle{One-Dimensional Guarded Fragments}. In
  \bibinfo{booktitle}{\emph{44th International Symposium on Mathematical
  Foundations of Computer Science, {MFCS} 2019, Aachen, Germany, August 26-30,
  2019}} \emph{(\bibinfo{series}{LIPIcs}, Vol.~\bibinfo{volume}{138})}.
  \bibinfo{publisher}{Schloss Dagstuhl - Leibniz-Zentrum f{\"{u}}r Informatik},
  \bibinfo{pages}{16:1--16:14}.
\newblock
\urldef\tempurl%
\url{https://doi.org/10.4230/LIPICS.MFCS.2019.16}
\showDOI{\tempurl}


\bibitem[Kieronski and Malinowski(2020)]%
        {DBLP:conf/lpar/KieronskiM20}
\bibfield{author}{\bibinfo{person}{Emanuel Kieronski} {and}
  \bibinfo{person}{Adam Malinowski}.} \bibinfo{year}{2020}\natexlab{}.
\newblock \showarticletitle{The Triguarded Fragment with Transitivity}. In
  \bibinfo{booktitle}{\emph{Logic for Programming, Artificial Intelligence and
  Reasoning 2020}} \emph{(\bibinfo{series}{EPiC}, Vol.~\bibinfo{volume}{73})}.
  \bibinfo{publisher}{EasyChair}, \bibinfo{pages}{334--353}.
\newblock


\bibitem[Kieronski and Rudolph(2021)]%
        {DBLP:conf/lics/KieronskiR21}
\bibfield{author}{\bibinfo{person}{Emanuel Kieronski} {and}
  \bibinfo{person}{Sebastian Rudolph}.} \bibinfo{year}{2021}\natexlab{}.
\newblock \showarticletitle{Finite Model Theory of the Triguarded Fragment and
  Related Logics}. In \bibinfo{booktitle}{\emph{36th Annual {ACM/IEEE}
  Symposium on Logic in Computer Science, {LICS} 2021, Rome, Italy, June 29 -
  July 2, 2021}}. \bibinfo{publisher}{{IEEE}}, \bibinfo{pages}{1--13}.
\newblock
\urldef\tempurl%
\url{https://doi.org/10.1109/LICS52264.2021.9470734}
\showDOI{\tempurl}


\bibitem[Kiero\'{n}ski and Tendera(2018)]%
        {KT18}
\bibfield{author}{\bibinfo{person}{Emanuel Kiero\'{n}ski} {and}
  \bibinfo{person}{Lidia Tendera}.} \bibinfo{year}{2018}\natexlab{}.
\newblock \showarticletitle{Finite Satisfiability of the Two-Variable Guarded
  Fragment with Transitive Guards and Related Variants}.
\newblock \bibinfo{journal}{\emph{ACM Trans. Comput. Logic}}
  \bibinfo{volume}{19}, \bibinfo{number}{2} (\bibinfo{year}{2018}),
  \bibinfo{pages}{8:1--8:34}.
\newblock


\bibitem[Lewis(1979)]%
        {Lewis:Unsolvable}
\bibfield{author}{\bibinfo{person}{Harry~R. Lewis}.}
  \bibinfo{year}{1979}\natexlab{}.
\newblock \bibinfo{booktitle}{\emph{Unsolvable Classes of Quantificational
  Formulas}}.
\newblock \bibinfo{publisher}{Addison-Wesley}.
\newblock


\bibitem[Mortimer(1975)]%
        {DBLP:journals/mlq/Mortimer75}
\bibfield{author}{\bibinfo{person}{Michael Mortimer}.}
  \bibinfo{year}{1975}\natexlab{}.
\newblock \showarticletitle{On languages with two variables}.
\newblock \bibinfo{journal}{\emph{Math. Log. Q.}} \bibinfo{volume}{21},
  \bibinfo{number}{1} (\bibinfo{year}{1975}), \bibinfo{pages}{135--140}.
\newblock


\bibitem[Pratt-Hartmann(2005)]%
        {PH:C2complex}
\bibfield{author}{\bibinfo{person}{Ian Pratt-Hartmann}.}
  \bibinfo{year}{2005}\natexlab{}.
\newblock \showarticletitle{Complexity of the Two-Variable Fragment with
  Counting Quantifiers}.
\newblock \bibinfo{journal}{\emph{J.\ of Logic, Language and Information}}
  \bibinfo{volume}{14} (\bibinfo{year}{2005}), \bibinfo{pages}{369--395}.
\newblock
Issue 3.


\bibitem[Pratt-Hartmann(2023)]%
        {PH23}
\bibfield{author}{\bibinfo{person}{Ian Pratt-Hartmann}.}
  \bibinfo{year}{2023}\natexlab{}.
\newblock \bibinfo{booktitle}{\emph{Fragments of First-Order Logic}}.
\newblock \bibinfo{publisher}{Oxford University Press},
  \bibinfo{address}{United Kingdom}.
\newblock


\bibitem[Rosati(2007)]%
        {DBLP:conf/icdt/Rosati07}
\bibfield{author}{\bibinfo{person}{Riccardo Rosati}.}
  \bibinfo{year}{2007}\natexlab{}.
\newblock \showarticletitle{The Limits of Querying Ontologies}. In
  \bibinfo{booktitle}{\emph{Proc. 11th Int. Conf. Database Theory (ICDT'07)}}
  \emph{(\bibinfo{series}{LNCS}, Vol.~\bibinfo{volume}{4353})},
  \bibfield{editor}{\bibinfo{person}{Thomas Schwentick} {and}
  \bibinfo{person}{Dan Suciu}} (Eds.). \bibinfo{publisher}{Springer},
  \bibinfo{pages}{164--178}.
\newblock


\bibitem[Rudolph(2011)]%
        {rudolph2011fodl}
\bibfield{author}{\bibinfo{person}{Sebastian Rudolph}.}
  \bibinfo{year}{2011}\natexlab{}.
\newblock \showarticletitle{Foundations of Description Logics}.
\newblock In \bibinfo{booktitle}{\emph{Reasoning Web. Semantic Technologies for
  the Web of Data -- 7th International Summer School 2011}},
  \bibfield{editor}{\bibinfo{person}{Axel Polleres}, \bibinfo{person}{Claudia
  d'Amato}, \bibinfo{person}{Marcelo Arenas}, \bibinfo{person}{Siegfried
  Handschuh}, \bibinfo{person}{Paula Kroner}, \bibinfo{person}{Sascha
  Ossowski}, {and} \bibinfo{person}{Peter~F. Patel-Schneider}} (Eds.).
  \bibinfo{series}{LNCS}, Vol.~\bibinfo{volume}{6848}.
  \bibinfo{publisher}{Springer}, \bibinfo{pages}{76--136}.
\newblock


\bibitem[Rudolph et~al\mbox{.}(2008a)]%
        {DBLP:conf/dlog/RudolphKH08}
\bibfield{author}{\bibinfo{person}{Sebastian Rudolph}, \bibinfo{person}{Markus
  Kr{\"{o}}tzsch}, {and} \bibinfo{person}{Pascal Hitzler}.}
  \bibinfo{year}{2008}\natexlab{a}.
\newblock \showarticletitle{All Elephants are Bigger than All Mice}. In
  \bibinfo{booktitle}{\emph{Proc.\,of DL 2008}}.
\newblock


\bibitem[Rudolph et~al\mbox{.}(2008b)]%
        {RudolphKH08}
\bibfield{author}{\bibinfo{person}{Sebastian Rudolph}, \bibinfo{person}{Markus
  Kr{\"{o}}tzsch}, {and} \bibinfo{person}{Pascal Hitzler}.}
  \bibinfo{year}{2008}\natexlab{b}.
\newblock \showarticletitle{Cheap Boolean Role Constructors for Description
  Logics}. In \bibinfo{booktitle}{\emph{Logics in Artificial Intelligence, 11th
  European Conference, {JELIA} 2008, Dresden, Germany, September 28 - October
  1, 2008. Proceedings}} \emph{(\bibinfo{series}{Lecture Notes in Computer
  Science}, Vol.~\bibinfo{volume}{5293})},
  \bibfield{editor}{\bibinfo{person}{Steffen H{\"{o}}lldobler},
  \bibinfo{person}{Carsten Lutz}, {and} \bibinfo{person}{Heinrich Wansing}}
  (Eds.). \bibinfo{publisher}{Springer}, \bibinfo{pages}{362--374}.
\newblock
\urldef\tempurl%
\url{https://doi.org/10.1007/978-3-540-87803-2\_30}
\showDOI{\tempurl}


\bibitem[Rudolph and \v{S}imkus(2018)]%
        {RS18}
\bibfield{author}{\bibinfo{person}{Sebastian Rudolph} {and}
  \bibinfo{person}{Mantas \v{S}imkus}.} \bibinfo{year}{2018}\natexlab{}.
\newblock \showarticletitle{The Triguarded Fragment of First-Order Logic}. In
  \bibinfo{booktitle}{\emph{{LPAR}}} \emph{(\bibinfo{series}{EPiC Series in
  Computing}, Vol.~\bibinfo{volume}{57})}. \bibinfo{pages}{604--619}.
\newblock


\bibitem[Schaerf(1994)]%
        {Scha94}
\bibfield{author}{\bibinfo{person}{Andrea Schaerf}.}
  \bibinfo{year}{1994}\natexlab{}.
\newblock \showarticletitle{Reasoning with individuals in concept languages}.
\newblock \bibinfo{journal}{\emph{Data Knowledge Engineering}}
  \bibinfo{volume}{13}, \bibinfo{number}{2} (\bibinfo{year}{1994}),
  \bibinfo{pages}{141--176}.
\newblock


\bibitem[Scott(1962)]%
        {scott1962decision}
\bibfield{author}{\bibinfo{person}{Dana Scott}.}
  \bibinfo{year}{1962}\natexlab{}.
\newblock \showarticletitle{A decision method for validity of sentences in two
  variables}.
\newblock \bibinfo{journal}{\emph{Journal of Symbolic Logic}}
  \bibinfo{volume}{27}, \bibinfo{number}{377} (\bibinfo{year}{1962}),
  \bibinfo{pages}{74}.
\newblock


\bibitem[Segoufin and ten Cate(2013)]%
        {DBLP:journals/corr/SegoufinC13}
\bibfield{author}{\bibinfo{person}{Luc Segoufin} {and} \bibinfo{person}{Balder
  ten Cate}.} \bibinfo{year}{2013}\natexlab{}.
\newblock \showarticletitle{Unary negation}.
\newblock \bibinfo{journal}{\emph{Logical Methods in Computer Science}}
  \bibinfo{volume}{9}, \bibinfo{number}{3} (\bibinfo{year}{2013}).
\newblock


\bibitem[Szwast and Tendera(2004)]%
        {ST04}
\bibfield{author}{\bibinfo{person}{Wies\l{}aw Szwast} {and}
  \bibinfo{person}{Lidia Tendera}.} \bibinfo{year}{2004}\natexlab{}.
\newblock \showarticletitle{The guarded fragment with transitive guards}.
\newblock \bibinfo{journal}{\emph{Annals of Pure and Applied Logic}}
  \bibinfo{volume}{128} (\bibinfo{year}{2004}), \bibinfo{pages}{227--276}.
\newblock


\bibitem[ten Cate and Franceschet(2005)]%
        {DBLP:journals/jolli/CateF05}
\bibfield{author}{\bibinfo{person}{Balder ten Cate} {and}
  \bibinfo{person}{Massimo Franceschet}.} \bibinfo{year}{2005}\natexlab{}.
\newblock \showarticletitle{Guarded Fragments with Constants}.
\newblock \bibinfo{journal}{\emph{Journal of Logic, Language and Information}}
  \bibinfo{volume}{14}, \bibinfo{number}{3} (\bibinfo{year}{2005}),
  \bibinfo{pages}{281--288}.
\newblock
\urldef\tempurl%
\url{https://doi.org/10.1007/s10849-005-5787-x}
\showDOI{\tempurl}


\bibitem[van Benthem({[n.\,d.]})]%
        {loosely}
\bibfield{author}{\bibinfo{person}{Johan van Benthem}.}
  \bibinfo{year}{[n.\,d.]}\natexlab{}.
\newblock \bibinfo{title}{{Dynamic bits and pieces. Technical Report LP-97-01,
  ILLC, University of Amsterdam, 1997.}}
\newblock \bibinfo{howpublished}{Available at \url{http://www.illc.
  uva.nl/Publications/reportlist.php?Series=LP}}.
\newblock


\end{thebibliography}

\end{document}